\documentclass[smallextended]{svjour3}      

\smartqed  
\usepackage[round]{natbib}
\usepackage{caption}
\usepackage{subcaption}
\usepackage{xcolor}
\usepackage{lineno,hyperref}
\usepackage{dblfloatfix}  
\usepackage{graphicx}  
\usepackage{enumitem}
\usepackage{relsize} 
\usepackage{amsmath}
\usepackage{amssymb}
\usepackage{multirow}
\usepackage{bbm}

\begin{document}

\title{Equality of Learning Opportunity via Individual Fairness in Personalized Recommendations}
\titlerunning{Equality of Learning Opportunity in Personalized Recommendations}      
\author{Mirko Marras \and Ludovico Boratto \and Guilherme Ramos \and Gianni Fenu}
\authorrunning{Marras et al.}

\institute{M. Marras \at University of Cagliari, Cagliari, Italy 
           \email{mirko.marras@unica.it} (Corresponding)         
           \and
           L. Boratto \at EURECAT, Barcelona, Spain 
           \email{ludovico.boratto@acm.org}         
           \and
           G. Ramos \at University of Porto, Porto, Portugal \email{guilhermeramos21@gmail.com}        
           \and
           G. Fenu \at University of Cagliari, Cagliari, Italy 
           \email{fenu@unica.it}         
}

\date{Revision submitted on October 26, 2020}

\maketitle

\vspace{-5mm}

\begin{abstract}
Online educational platforms are playing a primary role in mediating the success of individuals' careers. 
Therefore, while building overlying content recommendation services, it becomes essential to guarantee that learners are provided with equal recommended learning opportunities, according to the platform values, context, and pedagogy. 
Though the importance of ensuring equality of learning opportunities has been well investigated in traditional institutions, how this equality can be operationalized in online learning ecosystems through recommender systems is still under-explored. In this paper, we formalize educational principles that model recommendations' learning properties, and a novel fairness metric that combines them in order to monitor the equality of recommended learning opportunities among learners. 
Then, we envision a scenario wherein an educational platform should be arranged in such a way that the generated recommendations meet each principle to a certain degree for all learners, constrained to their individual preferences. 
Under this view, we explore the learning opportunities provided by recommender systems in a large-scale course platform, uncovering systematic inequalities. 
To reduce this effect, we propose a novel post-processing approach that balances personalization and equality of recommended opportunities. 
Experiments show that our approach leads to higher equality, with a negligible loss in personalization. 
Our study moves a step forward in operationalizing the ethics of human learning in recommendations, a core unit of intelligent educational systems. 
\end{abstract}

\section{Introduction} \label{sec:introduction}

Learning experience selection is at the heart of curriculum development and, consequently, a vital activity towards shaping individuals' knowledge and competencies~\citep{talla2012curriculum,druzhinina2018curriculum}. 
The term \emph{learning experience} generally refers to interactions in courses, programs, or other situations where learning takes place, including traditional and non-traditional settings~\citep{girvan2018virtual}. 
Notable examples of the latter, with an impact on individual experiences, are online course platforms, such as Coursera and Udemy. 
The proliferation of initiatives and the increasing adoption of these platforms have been requiring automated mechanisms of learning experience selection, tailored to the platform's values, context, pedagogy, and needs~\citep{rieckmann2018learning}. 

One aspect receiving special attention to support the learning experience selection on these online platforms is the ranking of courses deemed of relevance to individual learners. 
As a result, recommender systems are being deployed to suggest courses that accommodate learner's interests and needs~\citep{kulkarni2020recommender}. 
These recommended courses can be envisioned as learning opportunities being subjected to the attention of a learner. 
Though optimizing recommendations for learners' interests has been seen for years as the ultimate goal in the context of educational recommender systems, other principles inspected by curriculum-design experts in traditional settings (e.g., validity and affordability of the recommended resources) should be considered to shape online learning opportunities~\citep{talla2012curriculum,druzhinina2018curriculum}. 
Depending on the context, recommendations thus need to meet a trade-off between the interests of learners and the principles of the platform, providing learners with a well-rounded range of learning experiences~\citep{abdollahpouri2020multistakeholder}. 

Recommender capabilities represent a fundamental component of artificial intelligence systems in education. For this reason, ensuring equality among learners according to the recommended educational opportunities is essential, as the suggested courses may translate to educational gains or losses. 
By extension, education significantly influences individuals' life chances in the job market, and these opportunities should not be undermined by arbitrary decisions provided by a recommender system. 
\cite{meyer2016should} has revealed how equal learning opportunities, equal learning outcomes, and equal job opportunities relate to each other and emphasized the indispensable, but at the same time potentially dangerous\footnote{The demand for equal learning opportunities alone can lead to 
(i) attributing unequal learning outcomes that could have been avoided to unequal talent and effort, 
(ii) justify social inequalities by saying that all measures were taken to realize equal learning opportunities, 
(iii) limit efforts to merely realize equal educational opportunities.}, need for equal learning opportunities. 
This aspect has been investigated in traditional educational settings worldwide, such as in China~\citep{golley2018inequality}, Germany~\citep{buchholz2016secondary}, Japan~\citep{fujihara2016absolute}, Korea~\citep{byun2017different}, Spain~\citep{fernandez2017inequality}, and United States~\citep{shields2017equality}.

Thanks to extensive empirical analyses, these studies have identified several variables leading to unequal educational opportunities, with the gap in up-to-date competencies required by the job market and the considerable costs of access to education as two of them. 
Operationalizing these principles and the consequent notion of equal learning opportunity in online ecosystems via recommenders is still under-explored. 
These systems learn patterns from data with biases in terms of imbalances, which end up being emphasized in the recommendations~\citep{boratto2019effect}. 
Certain learners might thus receive low-quality opportunities based on the principles pursued by the targeted educational ecosystem. 
\figurename~\ref{fig:toy-example} shows an example of this phenomenon\footnote{Please note that the figures in this manuscript are best seen in color.}. Hence, it is imperative to mitigate inequalities, while retaining personalization. 
 
\begin{figure}[!t]
\minipage{1\linewidth}
    \centering
    \includegraphics[width=1.0\linewidth]{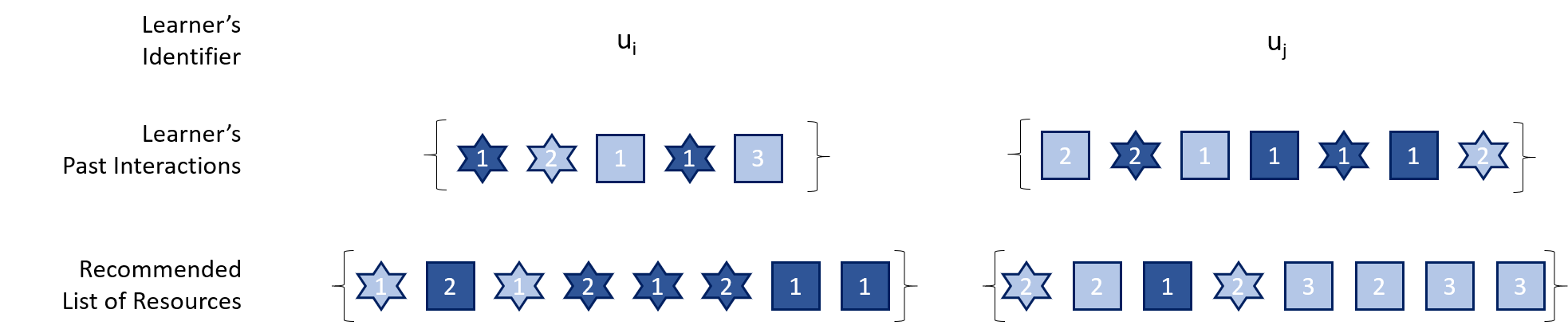}
\endminipage\hfill
\caption{\textbf{Example of Inequality in Recommended Learning Opportunities}. We consider two learners, $u_i$ and $u_j$ (first line). The mid-line shows us that the learners interacted with similar resources in terms of quality (star:high; square:low), validity (light-blue:old; dark-blue:fresh), and affordability level (1:low; 2:mid; 3:high). However, if we consider ranked lists provided by a collaborative algorithm to those two learners (bottom line), $u_i$'s recommendation list consists of mostly fresh, high-quality, and affordable resources, while $u_j$'s recommendations focus on obsolete, low-quality, and expensive resource.}
\label{fig:toy-example}
\end{figure} 

In this paper, we propose the concept of equality of recommended learning opportunities in personalized recommendations. 
To investigate how this concept applies to the online learning ecosystem, we envision a scenario wherein the educational platform should guarantee that a set of learning principles are met for all the learners, to a certain degree, when generating recommendations according to the learner's interests. 
Therefore, the ideal recommender system would 
(i) achieve higher consistency between the principles pursued by the platform and those measured in the recommendations, 
(ii) retain this consistency equal along with the learners' population, and 
(iii) preserve individual interests of learners. 
Under this scenario, we characterize the recommendations proposed by ten algorithms to learners in a real-world online course platform as a function of seven principles built upon knowledge and curriculum literature. 
Our exploratory analysis sheds light on systematic inequalities among learners based on the properties of the suggested courses and motivated us to devise a novel post-processing approach that balances equality and personalization in recommendations. Specifically, the contribution is four-fold:

\begin{itemize}
\item \textbf{Operational}: we define principles that model learning opportunity properties, and combine them in a fairness metric that monitors the equality of recommended learning opportunities across learners. 
\item \textbf{Social}: based on the above principles, we provide observations and insights on learning opportunities in recommendations, leveraging a dataset that includes more than 40K learners and 30K courses. 
\item \textbf{Technical}: we propose a post-processing recommendation approach that aims to balance personalization and equality of learning opportunity, making it possible to optimize different combinations of principles. 
\item \textbf{Ethical}: we evaluate the proposed approach on a real-world public dataset, and show that it leads to higher equality of recommended learning opportunities among learners, with a negligible loss in personalization. 
\end{itemize}

Our study moves a step forward in understanding how equality principles can be operationalized and combined in a formal notion of equal opportunities in educational recommendations. This contribution serves as a foundation for studying learners' notions, preferences, or limits concerning the equality of recommended learning opportunities. 
Therefore, the research community can be in a better position to know what questions to ask as part of interviews with learners, what scenarios to explore to elicit their concepts of fairness, and how to process data in the educational platform to monitor and ensure the equality principles. 
Thanks to its flexibility concerning the underlying principles, our approach can still be applied even if such investigations reveal that the principles captured in our study need to be fine-tuned or modified. 
This paper shapes a blueprint of the decisions and processes to be done, once empirically-validated principles have been defined under the targeted educational scenario.   

The remainder of this paper is structured as follows. Section 2 presents related work. Section 3 introduces the proposed principles and notion of equality, and Section 4 depicts the explorative analysis. Then, Section 5 describes and evaluates our approach for mitigating inequality of recommended opportunities. Finally, Section 6 provides concluding remarks and future research lines.

\section{Related Work} \label{sec:related-work}
This research lies at the intersection among the research communities working on Artificial Intelligence in Education (AIED), Recommender Systems (RecSys), and Fairness, Accountability, Transparency, Explainability (FATE).

\subsection{Educational Recommender Systems in the Artificial Intelligence Context} 
The advances in the area of computing technologies have facilitated the implementation of artificial intelligence applications in educational settings, improving teaching, learning, and decision making~\citep{pinkwart2016another}. 
Learners' behavioral patterns have been analyzed to make inferences, judgments, or predictions, serving for personalized guidance or feedback to students, teachers, or policymakers, for example, as proposed by~\cite{mao2019one, ren2019grade}. 

Our study in this paper treats recommender capabilities as an important component of AIED systems~\citep{khanal2019systematic}. 
Given the increase in resources available in online course platforms and learning management systems, designing personalized recommendations has become a key challenge. 
This challenge, thereby, increased research carried out by the AIED, Educational Data Mining (EDM), and Learning Analytics and Knowledge (LAK) communities. 
The most common objective in prior work is to suggest resources or peers in a given course. 
For instance, \cite{lin2020deep} proposed a deep attention-based model to recommend resources based on learners' online behaviors.  
This work outperformed state-of-the-art baselines in terms of accuracy. 
Similarly, \cite{wang2019adaptive} introduced a recommendation algorithm for textbooks, showing that adding adaptivity significantly increases engagement. 

Beyond recommending learning resources, 
\cite{eagle2018predicting} designed an algorithm for individualized help messages. Further, the work proves that the needs of learners for a lesson can be effectively predicted from their behavior in prior lessons. 
Conversely, \cite{mi2017adaptive} and \cite{chen2020csclrec} showed the important role of personalization while modeling forum discussion recommendations. 
Additionally, \cite{chau2018learning} assisted instructors with recommendations on the most relevant material to teach. 
Both works illustrated the presented algorithms' effectiveness in predicting learners' preferences. 
Reciprocal recommendations have been investigated by \cite{labarthe2016does} and \cite{potts2018reciprocal}. 
These works proposed two recommendation approaches for personalized contact lists. 
Their experiments uncovered that learners are more likely to engage in courses if they received peer recommendations. 
Finally, other tasks dealt with the accuracy of recommender systems in matching learners and job offers \citep{jacobsen2019s}.

Course recommendations have recently received attention due to the increasing number of initiatives carried out online. 
For instance, \cite{pardos2019designing} generated course recommendations that are novel and unexpected, but still relevant to learners' interests. 
Their results revealed that providing services optimized for serendipity allows learners to explore resources without a strong bias towards the learner’s (past) experience. 
Furthermore, \cite{morsomme2019content} found that students find recommendations for courses at other departments very helpful. 
In the university context, \cite{esteban2018hybrid} and \cite{boxuancourse} described two hybrid-methods for discovering the most relevant criteria that affect the course recommendation for university learners. 
Their results confirmed that the overall rating that a learner gives to a course is the most reliable information source. However, when it is complemented with other criteria about the courses, the recommendation accuracy increases. 
Capturing the sequential relationships across courses made it possible to \cite{polyzou2019scholars} to devise a course recommender system. 
This system outperforms other collaborative filtering and baseline approaches. 

Extensive research work has been devoted to mastery learning in intelligent tutoring systems, which select educational resources for learners based on knowledge tracing. 
For instance, \cite{thaker2020recommending} automatically identified the most relevant textbooks to be recommended, by incorporating learner's knowledge states. 
\cite{chanaa2020predicting} proposed a model that predicts learner's needs for recommendation using dynamic graph-based knowledge tracing. 
By learning feature information and topology representation of learners, their model achieved a competitive accuracy of more than 80\%. 
Subsequently, to avoid the mismatch between learners and learning objects, \cite{dai2016course} introduced a recommender system for suggesting learning objects, with a domain knowledge structure to connect learners’ skills and learning objects. 
They showed that the accuracy is higher when texts related to the concerned domain knowledge are involved. In~\cite{ChanRHKSBPCPNSBSDLMH06}, the authors conclude that ready-to-hand access conceives the potential for a new evolution of technology-enhanced learning (TEL) phase. 
This phase is defined as ``seamless learning spaces'' and marked by a succession of the learning experience over different scenarios. Further, it arises from the availability of one (or more) device(s)  per student (``one-to-one''). 
The one-to-one TEL holds the potential to “cross the chasm” from early adopters handling isolated design studies to adoption-based research and extensive implementation.
Finally, \cite{ai2019concept} designed an exercise recommender that considers exercise-concept mappings while tracing learners’ knowledge. This led to a better performance than the heuristic policy of maximizing learners’ knowledge level.

Our contribution differs from prior work in three major ways. 
First, current approaches have been mostly optimized for learners' preference prediction, given their ratings, grades, or enrolments. 
Conversely, our approach aims to control how learning opportunities vary, based on high-level properties directly measurable on the ranked lists (e.g., familiarity and learnability of the recommended courses), going beyond preference prediction accuracy only. 
Second, even though several recommender systems integrated beyond-accuracy aspects, such as learnability and serendipity, combining them with other aspects and decoupling them from the underlying recommendation strategy appears impractical. 
Differently, our post-processing mechanism can be applied to the output of any recommender system to arrange recommendations that meet a range of properties. 
Third, controlling how much the generated recommendations are equally consistent across learners has been rarely investigated. 
Hence, we introduce and operationalize a novel fairness metric that monitors equality among learners concerning the targeted educational principles.   

\subsection{Fairness in Artificial Intelligence for Education} 
Characterizing and counteracting potential pitfalls of data-driven educational interventions is receiving increasing attention from the research community. 
Educational applications of artificial intelligence are not immune to the risks observed in other domains. 
Moreover, the design of systems may often be driven more by profit than by actual educational impact, with serious potential risks (e.g., algorithmic biases, invasion of privacy, or negative social impacts) out-weighting any benefits~\citep{shum2018transitioning,bulger2016personalized,williamson2017decoding}. 

Responding to these concerns may be critical to determine the fairness of AIED systems and to shape how the ethics for human learning are more broadly defined. 
However, only a few works have focused on ethics in education, where the increasing use of learning analytics and artificial intelligence raises unique context-specific challenges~\citep{ocumpaugh2014population}. For instance, while most existing fairness	auditing and de-biasing methods require access to sensitive demographic information (e.g., age,	race, gender)	at an individual level, such information is often unavailable to AIED practitioners~\citep{holstein2019designing}. 
Also, it becomes challenging to define equitable outcomes in contexts where a system results in disparate outcomes across subpopulations, such as learners having lower or higher prior knowledge~\citep{hansen2015democratizing}. 
Although the community has been interested in the ethical dimensions of data-driven educational systems~\cite{drachsler2015ethical,sclater2015code,tsai2017learning}, the focus has often been on policies. 

Despite this widespread attention, fairness has been rarely discussed from a more practical and technical perspective~\citep{holmes2019ethics,holstein2019fairness,mayfield2019equity,porayska2019accountability}. 
Given that designing methods for addressing unfairness challenges can be highly context-dependent~\citep{holstein2019improving,green2018myth,selbst2019fairness}, the education research community has started to explore what fairness, accountability, transparency, and ethics look like in technology-supported education specifically. 
For instance, \cite{yu2020towards} found that combining the profile and material data sources does not fully neutralize biases, and it still leads to high rates of underestimation among disadvantaged groups for learners' success prediction. 
Similarly, \cite{doroudi2019fairer} showed that knowledge tracing algorithms are susceptible to unfairness, but that knowledge tracing with the additive factor may be more equitable. 
\cite{hu2020towards} focused on individually fair models for identifying students at-risk of underperforming. 
The work shows how to effectively mitigate bias in models and make the models useful in aiding all learners. 
Conversely, \cite{abdi2020complementing} investigated the impact of complementing educational recommender systems with transparent justifications for their recommendations. 
This impact leads to a positive effect on engagement and perceived effectiveness and an increasing sense of unfairness due to learners not agreeing with how their competency is modeled. 
Such appraisal is key to enhancing our understanding of fairness, building on knowledge gleaned from AIED research.

However, to the best of our knowledge, controlling equality in educational recommender systems has been so far under-explored. 
As a consequence, we investigated how fairness and ethical aspects have been treated by the general-purpose RecSys community~\citep{barocas2017fairness,RAMOS2020102058}, analyzing whether and how the resulting treatments can be tailored to recommender systems in education. 
Fairness across end-users deals with ensuring that users who belong to different protected classes (group-based) or are similar at the individual level (individual-based) receive recommendations with the same quality. 
Group-based fairness requires that the protected groups are treated similarly to the advantaged groups or the population as a whole. 
For instance, \cite{ZhuHC18} designed an approach that identifies and removes from tensors all gender information about users. This approach leads to fairer recommendations regardless of the user's group membership. 
\cite{RastegarpanahGC19} generated artificial data to balance group representations in the training set and minimize the difference between groups in terms of mean squared error. 
Similarly, \cite{YaoH17} proposed metrics related to population imbalance (i.e., a class of users characterized by a sensitive attribute being the minority) and observation bias (i.e., a class of users who produced fewer ratings than their counterpart). 
Under a similar scenario, for instance, \cite{BeutelCDQWWHZHC19} built a pairwise regularization that penalizes the model if its ability to predict which item was clicked is better for one group than the other. 
These works showed that operationalizing their metrics in the recommender's objective function results in fairer recommendations. 

Group-fairness is, unfortunately, inadequate as a notion of fairness, given that there exist circumstances wherein group fairness is maintained. 
However, from an individual point of view, the outcome is blatantly unfair. 
Hence, our study cares about learners as individuals, not as belonging to a class based on a certain sensitive attribute. 
This condition also fits with educational scenarios where sensitive demographic attributes (e.g., age, race, gender) at an individual level are unavailable to learning analytics practitioners. 
Examples of individual fairness notions proposed by the RecSys community~\citep{biega2018equity,lahoti2019ifair,singh2019policy} imply that similar users should have similar outcomes. 
They capture fairness by defining that any two individuals similar concerning a particular task should be treated likewise, assuming that a similarity metric between individuals exists.
For instance, in a health-related recommender system, two patients who have a similar pathology should receive recommendations with the same quality.

Our study generalizes the original definition of individual fairness and applies it to the educational context. 
Specifically, we aim to provide all learners, indistinctly, with recommended learning opportunities that are equally consistent concerning the targeted principles. 
We do not rely on any notion of similarity across pairs of learners based on how the targeted principles were met in the past. 
These notions could even emphasize existing inequalities (e.g., two learners who similarly experienced less learnable recommended courses in the past could end up receiving low learnable courses more and more, though the recommender would have been fair under the original individual fairness notion). 
Therefore, our notion of equality may be seen as an \emph{individual statistical parity}, which equalizes the consistency for the principles across all learners. 

\section{Problem Formulation} \label{sec:problem-formulation}
In this section, we formalize recommendation concepts, educational principles, and metrics that respectively monitor consistency and equality of recommended learning opportunities among learners. 

\subsection{Preliminaries}
Given a set of learners $U$ and a set of educational resources $I$, we assume that learners expressed their interest for a subset of resources in $I$. The collected feedback from learner-resource interactions can be abstracted to a set of  pairs ($u$, $i$), implicitly obtained from user activity, or triplets ($u$, $i$, $rating$) explicitly provided by learners, shortly denoted by $R_{u,i}$. We denote the learner-resource feedback matrix by $R \in \mathbb{R}^{M*N}$m where $R_{u,i} > 0$  indicates that learner $u$ interacted with resource $i$, and $R_{u,i}=0$ otherwise. Furthermore, we denote the set of resources that learners $u\in U$ interacted with by $I_u=\{i\in I\,:\,R_{u,i} > 0\}$.  

We assume that each resource $i \in I$ is represented by a m-dimensional feature vector $F_i = (f_1, \ldots, f_m)$ over a set of features $F = \{F_{i,1}, F_{i,2}, \ldots, F_{i,m}\}$. Each dimension $F_j$ can be viewed as a set of values or labels describing a feature of a resource $i$, $f_{i,j} \in F_j$ for $j=1, \ldots, m$. In our experiments, we considered five features, i.e., instructional level (discrete), resource category (discrete), last update timestamp (discrete datetime), number of enrolled learners (continuous), and price (continuous). Furthermore, we assume that each resource $i \in I$ is composed by a set of assets $L_i$. Each $l_{i,j} \in L_i$ has a type $t_{i,j} \in T$. In our study, we considered $T=\{ Video, Article, Ebook, Podcast \}$, due to their popularity and their availability in the public datasets.        

We assume that a recommender estimates relevance for unobserved entries in $R$ for a given learner and uses them to rank resources. It can be abstracted as learning $\widetilde{R}_{u,i} \in [0,1]$, which represents the predicted relevance of resource $i$ for learner $u$. Given a certain learner $u$, resources $i \in I \setminus I_u$ are ranked by decreasing $\widetilde{R}_{u,i}$, and top-$k$, with $k\in\mathbb N$ and $k>0$, resources are recommended. Finally, we denote the set of $k\in\mathbb N$ resources recommended to user $u$ by $\Tilde{I}_u$. 

\subsection{Modeling Recommended Learning Opportunity through Principles} \label{sec:principles}
Given that the recommendation capabilities are an important part of AIED systems, investigating whether educational recommender systems are fair and how they can be made a vehicle for making our educational systems fairer is essential. Capturing, formalizing, and operationalizing notions of equality can shape our understanding of the extent to which the educational offerings available to learners provide them with equal opportunities and how recommender systems influence the normal course of educational business. To this end, defining the variables to be equalized constitutes a natural pre-requisite. 

Organizing learning opportunities in classroom settings has been traditionally a responsibility of instructional designers or teachers. 
To this end, they rely on a range of principles coming from the curriculum design field, including significance, self-sufficiency, validity, interest, utility, learnability, feasibility \citep{talla2012curriculum,druzhinina2018curriculum}. 
Hence, our study assumes that the notion of equality needs to consider these principles derived from the instructional design beliefs as those to be equalized in recommendations, given their real-world validity for learners' educational experiences from the instructional perspective. 
However, we do not argue that this approach and the consequent principles are the unique, right ones, as they strongly depend on the educational context and the adopted eliciting processes. 
For instance, alternatively, to the one based on curriculum design, which are properties of the recommended resources that might affect learners' experiences, a user-centered eliciting process might come up with the principles to be equalized from the learner's notions on the fairness of the educational resource selection decisions they make or are presented with. 
Nonetheless, our curriculum-design principles would be a great starting point for a human-centered eliciting process, which is thus not the focus of our study and left as a natural future work. 
Our principle modeling aims to serve as a solid ground for researchers, posing them in a  better position to know what questions, scenarios, and data explore while eliciting their concepts of fairness. 
Therefore, we argue that our study represents a blueprint of the decisions and processes to be done, even when the principles are defined from learner's beliefs in a given educational scenario.   

Human inspection of curriculum-design-based principles is usually based on textual guidelines, and the translation into numerical indicators, when available, is dependent on the specificity of the educational context. 
Given the unique characteristics of the online educational context and the constraints the platforms introduced in the collection of learners' data, we assume that the principles are based on data that would typically be available in a platform in which an educational recommender system would be embedded. 
For this reason, not all the principles and not all the guidelines can be directly operationalized\footnote{Section \ref{sec:limitations} identifies a range of limitations derived from these assumptions.}. 
Specifically, we envision a scenario wherein only a representative subset of principles are embedded in the recommender system's logic. 
The educational platform is thus empowered with the capability of controlling the extent to which the list of courses recommended to learners meets each principle. 
While the high-level conception of the selected principles is assumed to be relevant, their operationalization into the recommender's system logic is strongly dependent on the platform, turning to simplified implementations in some cases. 
While we provide formulations that are as general as possible, we will ensure that our approach can be extended or adapted to any principle. 

Formally, we consider a set $C$ of functions $c_{\Tilde{I}_u}(\cdot) \, : \, I^k \xrightarrow{} [0,1]$. Each function receives a set of $k$ resources $I^k$ and returns a value indicating how much the set of resources meets that principle. The higher the value, the higher the extent to which the principle is met. Specifically, we consider the following seven principles, whose mathematical formulation is provided in Appendix \ref{sec:math-formulation}.

\begin{definition}[Familiarity] \textit{Familiarity is defined as whether the subject is familiar with the recommended content, as measured by whether the relative frequency of the course categories in a recommended set is proportional to that in the courses the learner took.} \end{definition}

Familiarity is at the heart of a learner-centered education. Learners might be more comfortable if the subject matter is meaningful to them, and it is assumed that it becomes meaningful if they are familiar with that subject. \cite{xie2009selection} supported this observation through a descriptive and statistical analysis, uncovering that the familiarity was correlated to the content searching behavior. Similarly, \cite{qiu2017content} showed that participants were behaviourally and cognitively more engaged in tasks with familiar topics as well as having a more positive affective response to them.

Our study models familiarity by means of the category of the resources in a recommended list, encoded into a pre-defined taxonomy. If the relative frequency of the course categories in a recommended set is proportional to that in the courses the learner took, we assume that the familiarity is high (a value of $1$). Conversely, the minimum familiarity of $0$ is achieved when the recommender suggests resources in the opposite direction concerning the learner's most familiar categories. This principle is related to the concept of calibrated recommendations, which aim to reflect the various interests of a user in the recommended list with their appropriate proportions \citep{steck2018calibrated}\footnote{Differently from that work, which used the Kullback–Leibler divergence as a non-symmetric, unbounded, and computationally unstable distance function, we adopted the \emph{Hellinger} distance, which is symmetric and bounded in the range $[0,1]$.}. 

\begin{definition}[Validity] \textit{Validity is defined as whether the course is likely to be up-to-date and not obsolete, as measured by when content was last updated. A subject is assumed to be more meaningful if it has been newly updated.} \end{definition} 

Controlling the validity of the learning content is one of the major axes of education, since learners would not find information not valid anymore in the courses. One way for maintaining the course content valid is continuously updating it, either with more recent content or with new versions of the same contents (e.g., adapted based on the learners' feedback). This practice also shows learners that the course is alive. Curriculum-design experts usually seek to follow current trends and carefully consider the validity of a curriculum \citep{druzhinina2018curriculum}; otherwise, the opportunity becomes obsolete. Similarly, \cite{bulathwela2019towards} highlighted that freshness is one of the main factors shaping the content validity perceived by learners. Hence, we assume that validity should be taken into account in the offered recommendations. 

Our scenario operationalizes the validity principle by controlling that learners are presented with recommended courses that have been recently or frequently updated. Values close to 0 imply that the recommended list includes courses no longer updated for a long time, while values close to 1 are achieved by recommended courses with recent updated\footnote{It should be noted that using recency of updates as a proxy for validity does not consider that, for instance, a course on foundational material updated many times in the past does not benefit from recent updates.}.

\begin{definition}[Learnability] \textit{Learnability is defined as whether the recommended courses present an opportunity coherent with the learner's capability, as measured by whether the set of courses varies in terms of instructional level.} \end{definition} 

Learnability is associated with the ease, efficiency, and effectiveness with which learners can perform an activity of knowledge acquisition. Our study assumes that the subject matter to be recommended should be within the knowledge schema of the learners and their experiences. The literature indicated that learnability impacts learner motivation to learn \citep{conaway2016keys}, motivating us to monitor this principle in the recommended lists.

Our scenario includes this concept to ensure that courses are presented at diverse instructional levels, to maximize the possibility that learners can find an opportunity coherent with their capability. While our assumption might appear over-simple compared to other knowledge-tracking methods (e.g., the zone of proximal development), the current online course platforms impose constraints that should be met. Specifically, data on mid-term quizzes and final exams are often not recorded internally, and this leaves the implementation of more advanced knowledge-tracking techniques as a primary open issue. Therefore, we rely on the recommendation of courses at different levels for students with different abilities. Learnability values close to 0 imply inequality among levels, while the high balance is obtained with values close to 1.

\begin{definition}[Variety] \textit{Variety is defined as whether the recommendation takes into account that learners are different and learn in different ways based on their interests and ability, as measured by the degree to which the recommended courses present a mix of different asset types.} \end{definition}

Providing course material in a variety of formats represents a primary objective. For instance, by studying the online course design and teaching practices of award-winning teachers, \cite{kumar2019award} uncovered that including video, audio, reading, and interactive content made courses more engaging. In another study, \cite{papathoma2020guidance} highlighted how this variety of formats favors also the accessibility of a course, given that learners may struggle with a particular medium (e.g., due to a reading barrier such as dyslexia or a video barrier such as hearing or attention problem). Therefore, monitoring whether a course provides learners with a large variety of content formats is important. 

Our operationalization of variety assumes that varied asset types may be provided to help learners comprehend the subject from various perspectives. Hence, variety values close to 0 mean that the learning opportunities are focused on only one asset type, while types greatly vary for values close to 1. 

\begin{definition}[Quality] \textit{Quality is defined as the perceived appreciation of the recommended resources by the learners, as measured by the ratings that the learners assign to resources after interacting with them.} \end{definition}

Student evaluation of teaching quality is important to assess current teaching experiences. Teaching evaluation helps to promote a better learning experience for learners and provide information to future learners while deciding for attending a course. However, defining quality in online learning is challenging because there is no real consensus on its true meaning. Consequently, quality is evaluated differently depending on the organization in charge of measuring it. For instance, \cite{darwin2017contemporary} showed that student ratings are perceived as a valuable, though fragile, source of intelligence about the effectiveness of curriculum design, teaching practices and assessment strategies. On the other hand, \cite{gomez2016measuring} observed that learners considered other core variables in defining quality in online programs, such as the ability to transfer, knowledge acquisition, learner satisfaction, and course design.

Our study operationalizes quality by leveraging the learners' ratings. Rating values close to 0 mean that the learning opportunities are of low-quality, while values close to 1 are measured for high-quality recommended opportunities. Even though some studies demonstrated that learners' ratings do not often correlate with other measures of quality (e.g., learning outcomes), this design choice makes it possible to meet the current constraints in data gathering in large-scale educational platforms, and allows us also to maintain this principle as much general as possible.

\begin{definition}[Manageability] \textit{Manageability is defined as whether the online classes are large or small, as measured by the number of learners enrolled in the recommended courses, with small classes considered more manageable.} \end{definition} 

Organizational aspects are critical for shaping learners' experiences. In this context, class size differences may influence academic interactions between students and their professors and peers. For instance, with a large number of learners, the instructor may work harder to combat student passivity and encourage participation, as learners feel an increasing sense of anonymity. This point is confirmed by the study of \cite{beattie2016connecting}, which uncovered that the likelihood of academic interactions about course material and assignments with professors was diminished in larger classes, as was the probability of talking to peers about ideas from classes. Similarly, \cite{lowenthal2019does} revealed that online courses with smaller enrollments are seen better for student learning and faculty satisfaction, by learners and instructors.

Our study embeds the notion of manageability, associating it with the size of the course class where the recommended opportunities take place. This principle is relevant to offer opportunities under smaller and controlled classes. Hence, manageability values close to 0 mean that the learning opportunities include very large classes, while values close to 1 refer to small classes\footnote{It should be noted that other operationalizations might refer to a teacher’s ability to manage students or how topics are selected by teachers. Other aspects of manageability (e.g., number of teaching assistants, number of recitation sections) can be captured under this principle, depending on the educational context and platform capabilities.}.

\begin{definition}[Affordability] \textit{Affordability is defined as the cost of accessing to the recommended opportunities, as measured by the enrolment fees of the suggested courses, with less expensive courses having higher affordability value.} \end{definition} 

Dealing with the increasing costs of education is critical, given that lots of learners need access to vastly more affordable and quality education opportunities, including tuition-free course options. For instance, \cite{mohapatra2017adopting} emphasized that the affordability of the offering is one of the prime predictors of the learners' perception, while \cite{joyner2016unexpected} uncovered that providing more affordable courses has led to a learner population more intrinsically motivated to learn, more experienced, and more professionally diverse in some contexts. Institutions and platforms are thus under increasing pressure to provide more affordable learning without sacrificing optimal learning outcomes. For this reason, we monitor the affordability principle in the recommendations generated in the online platform. 

Our notion of affordability aims to control the degree of economic accessibility for the recommended opportunities, measured as by their enrolment fees. Specifically, we consider how much the learning opportunities cover a range of fee ranges. A value close to 0 means that the learning opportunities are expensive, while a value close to 1 corresponds to free-of-charge opportunities.  

\vspace{2mm} Though each of the principles has relevance for students' educational experiences from the instructional design perspective, the set of principles could be expanded. Additionally, the proposed set is not meant to be the unique right set. Furthermore, it should be noted that massive online course platforms are often targeted for profitability and large coverage, and a business plan should be provided, and making a profit must be considered as a primary goal for such courses. 
It thus remains unanswered the question on how to integrate business and educational principles, which would deserve a broader and specific discussion beyond touching the technical details. Our study in this paper assumes that the educational platform is primarily driven by the positive impact on the involved learners. 
Furthermore, another point to mention is that there might be several principles left out, but relevant for certain educational scenarios or specific platforms (e.g., the time of day a course is offered). 
This observation brings into question that there is no one-size-fits-all set of principles to be equalized in educational recommender systems. 
Contextual, technical, social, and pedagogical aspects play a key role in shaping the resulting principles, and, again, we point out that our goal is to provide a blueprint of processes to be done, even when the principles differ from the ones presented in this paper.

\subsection{Equality of Recommended Learning Opportunity}
The notion of equality presented in this paper hinges on low variability across learners in the degree to which a set of recommended educational resources meets each principle, to a targeted degree, where the target is learner-specific or defined at the system level. 
Specifically, to formalize the equality of recommended learning opportunities, we first need to define how much the list recommended to each learner meets the principles targeted by the educational platform. 
In this paper, we propose to operationalize the concept of consistency across principles as the similarity between 
(i) the degree all principles are met into the recommended list and 
(ii) the principles degree targeted by the educational platform. 
The higher the similarity, the higher the extent to which the principles are met. 
We resorted to the operationalization of this metric locally on each ranked list so that it will be possible to optimize such a metric on a pre-computed recommended list through a post-processing function (see Section \ref{sec:eq-control}). For the ranked list $\Tilde{I}_u$ of learner $u$, we assume the platform aims to ensure a targeted degree for each principle $c\in C$ for each learner: 
\begin{equation}
    p_u(m)\in[0,1], \,\,\, \forall \, m \in \{0, \cdots, |C|-1\} 
\end{equation}

\noindent When a recommender computes the top-$k$ resources $\Tilde{I}_u$ to be suggested to learner $u \in U$, we define the extent to which the principles are met in $\Tilde{I}_u$ as: 

\begin{equation}
    q_{\Tilde{I}_u}(m) = c_{\Tilde{I}_u}(m) \in[0,1], \,\,\, \forall \, m \in \{0, \cdots, |C|-1\} 
\end{equation}

\noindent where the value corresponding to each principle $q_{\Tilde{I}_u}(m)$ is computed by applying the formulas formalized in the previous section. 

For the ranked list of a learner $u$, the principles targeted by the educational platform are met if the values in $p_u$ and $q_{\Tilde{I}_u}$ are aligned with each other. 
Hence, to achieve the principles' goals targeted by the educational platform, we compare the vectors $p_u$ and $q_{\Tilde{I}_u}$. 
Specifically, we define the notion of \textbf{Consistency} between 
(i) target principles and 
(ii) how much principles are achieved in recommendations, by the complement of the Manhattan (M) distance\footnote{The Manhattan distance between two vectors is equal to the one-norm of the distance between the vectors. 
Specifically, it represents the distance between two points in a grid-based on a strictly horizontal and/or vertical path, i.e., along the grid lines.}, a symmetric and bounded distance measure. The consistency for each learner and for the entire learners' population is formulated as follows:  

\begin{equation}
\begin{split}
Consistency(u|w) = 1 - M(p_u,q_{\Tilde{I}_u}|w) =  1 - \frac{1}{|\Tilde{I}|} \sum_{i = 1}^{|\Tilde{I}|} \; w_i \; \left|[p_u]_i-[{q}_{\Tilde{I}_u}]_i\right|
\end{split}
\label{eq:local-consistency}
\end{equation}

\begin{equation}
\begin{split}
Consistency(U|w) = \frac{1}{|U|} \sum_{u \in U} Consistency(p_u,q_{\Tilde{I}_u}| w)
\end{split}
\label{eq:global-consistency}
\end{equation}
\vspace{2mm}

\noindent where $w$ is a vector of size $|C|$; the element $w_i$ is the weight assigned to the principle $i$, between 0 and 1. 
\emph{Consistency} is $1$ if $p_u$ and $q_{\Tilde{I}_u}$ are perfectly balanced, meaning that the principles pursued by educational platform are met. 
Conversely, the lowest \emph{Consistency} $0$ is achieved when $p_u$ assigns value 0 to every principle that $q_{\Tilde{I}_u}$ assigns value $1$ (or vice versa), so that the distributions are completely unbalanced. 
In the latter situation, the recommender suggests resources opposite to the educational platform's goals. 
Given that principles are context-sensitive, our notion of consistency might provide different target degrees of principles for each learner or each time period. 
For instance, concerning to familiarity, different learners might have a different propensity towards familiar content, and the same learner may, at different times, have different preferences. 
The above formulation allows modeling these circumstances.  

Given the notion of consistency, we can formalize the notion of \textbf{Equality} across consistencies as the complement of the Gini index\footnote{The Gini index is a statistical measure of distribution that aims to model inequality among a population. 
The coefficient ranges from 0 to 1, with 0 representing perfect equality and 1 representing perfect inequality. 
Formally, it is defined based on the Lorenz curve, which plots the proportion of the total consistency of the population cumulatively earned by the bottom x\% of the population. 
The Gini index is the ratio of the area that lies between the line of equality and the Lorenz curve over the total area under the line of equality.} over the consistencies across learners. 
The Gini index ranges between 0 and 1, with higher values representing distributions with high inequality. 
It is used as: 
\begin{equation}
  \begin{split}
    Equality(U|w) = 1 - GINI\left(\{Consistency(u|w) \; | \; \forall \, u \in U\}\right) 
  \end{split}
\label{eq:equality}
\end{equation}
where a value of $0$ represents the largest inequality across consistencies, and a value of $1$ means that the recommender systems are perfectly equal across learners. 
Differently from \cite{lahoti2019operationalizing,biega2018equity}, we count as a positive effect when learners achieve high consistency in recommendations, regardless of the consistency in their past interactions. 
Thus, the ideal recommender system would be the one that ($i$) achieves the higher consistency between the principles pursued by the platform and those measured in the recommendations, ($ii$) keeps it equal over the learners' population, and ($iii$) retains individual interests of learners.  

It should be noted that our notion of equality is defined as providing the same consistency on principles to all learners, without leveraging any information on learners' sensitive features, e.g., gender. 
The targeted degree for each principle for each student is assumed to be set by or known by the educational platform. 
Our approach enables a platform to set the same targeted degree for all students or apply student-specific targeted degrees set based on the previous learners' preference or elicited from learners. 
To better focus on the core contribution of this paper, Section \ref{sec:exp-analysis} will investigate whether all the principles can be maximized for all learners, leaving student-specific targeted degrees as part of a human-centered study\footnote{We argue that quantifying student-specific targeted degrees from the previous learners' preferences encoded in the training set might be unreliable, given that the preference of each learner might have been biased by the recommender system itself.}. 
Furthermore, the reliance on stakeholders empowered with decision-making capabilities to configure the platform with the considered principles and their different targeted degrees represents an essential element towards implementing our notion of equality. 
This primary responsibility of stakeholders is in addition to all the others involved in the educational ecosystem (e.g., selecting the preferred system, deciding the recommendation strategy, and defining the visual interface).   

\section{Exploratory Analysis} \label{sec:exp-analysis}
To illustrate the trade-off between learners' interests and the considered principles, and further emphasize the value of our analytical modeling, we characterize the learning opportunities proposed by ten algorithms to learners of a real-world educational dataset, as a function of the proposed principles. 
Specifically, we aim to evaluate how these principles play out in each algorithm.  

\subsection{Data} \label{sec:data}
We analyze data from the educational context, exploring the role of the proposed principles in recommendations. 
We remark that the experimentation is made difficult because there are very few large-scale educational datasets coming from this specific field of online education. 
To the best of our knowledge, COCO \citep{dessi2018coco} is the widest educational dataset with all the attributes required to model the proposed principles and with enough data to assess performance significantly. 
Collected from an online course platform, it includes 43,045 courses and 4,123,127 learners who gave 6,564,870 ratings. 
Other educational datasets proposed by \cite{feng2019understanding,zhang2019hierarchical,qiu2016modeling} generally include $(learner,course,rating)$ triplets only, as needed in traditional recommendation scenarios.

\begin{figure}[!b]
\begin{subfigure}[t]{0.32\linewidth}
    \centering
    \includegraphics[width=1.0\linewidth]{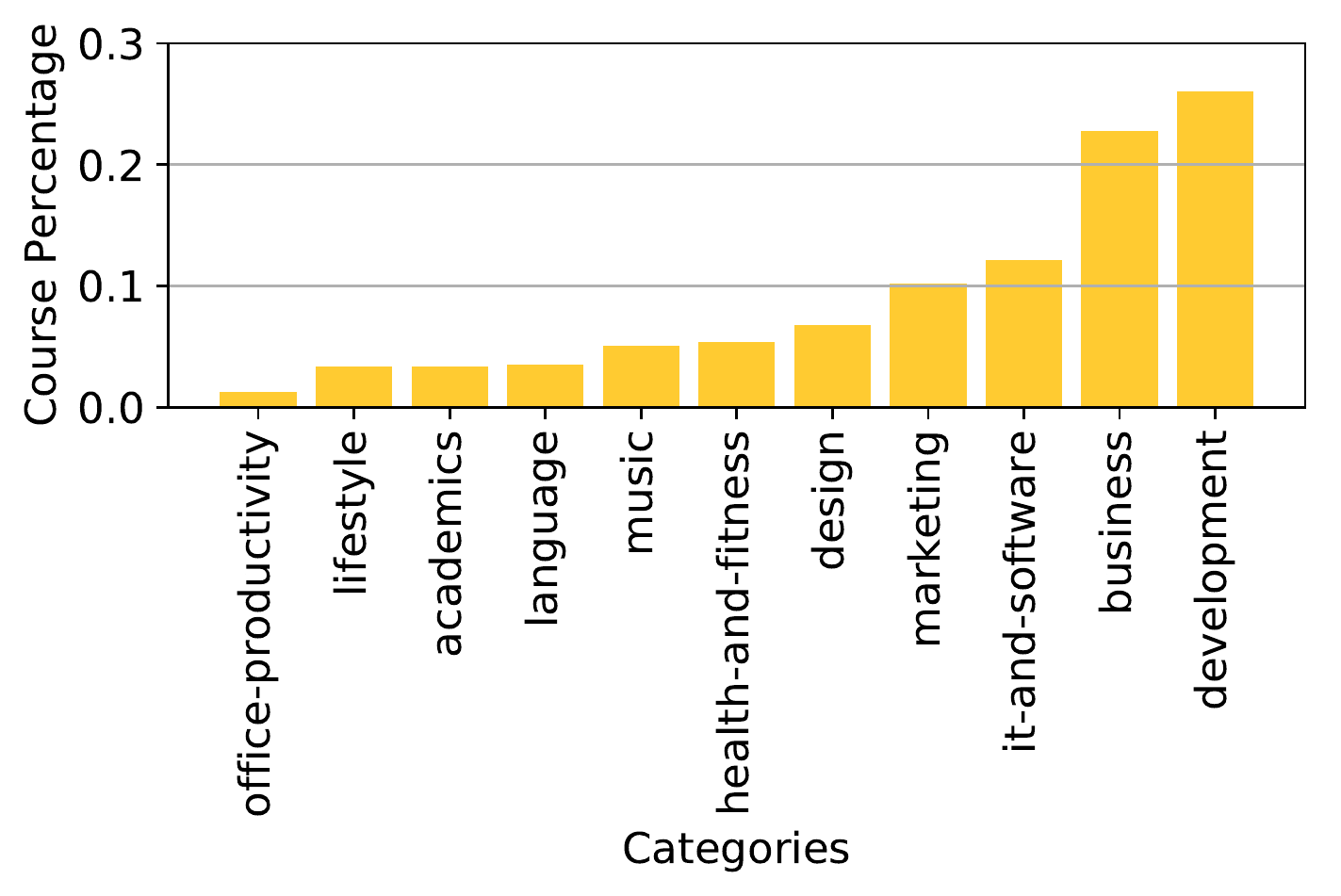}
    \caption{Categories : Interest}
\end{subfigure}
\begin{subfigure}[t]{0.32\linewidth}
    \centering
    \includegraphics[width=1.0\linewidth]{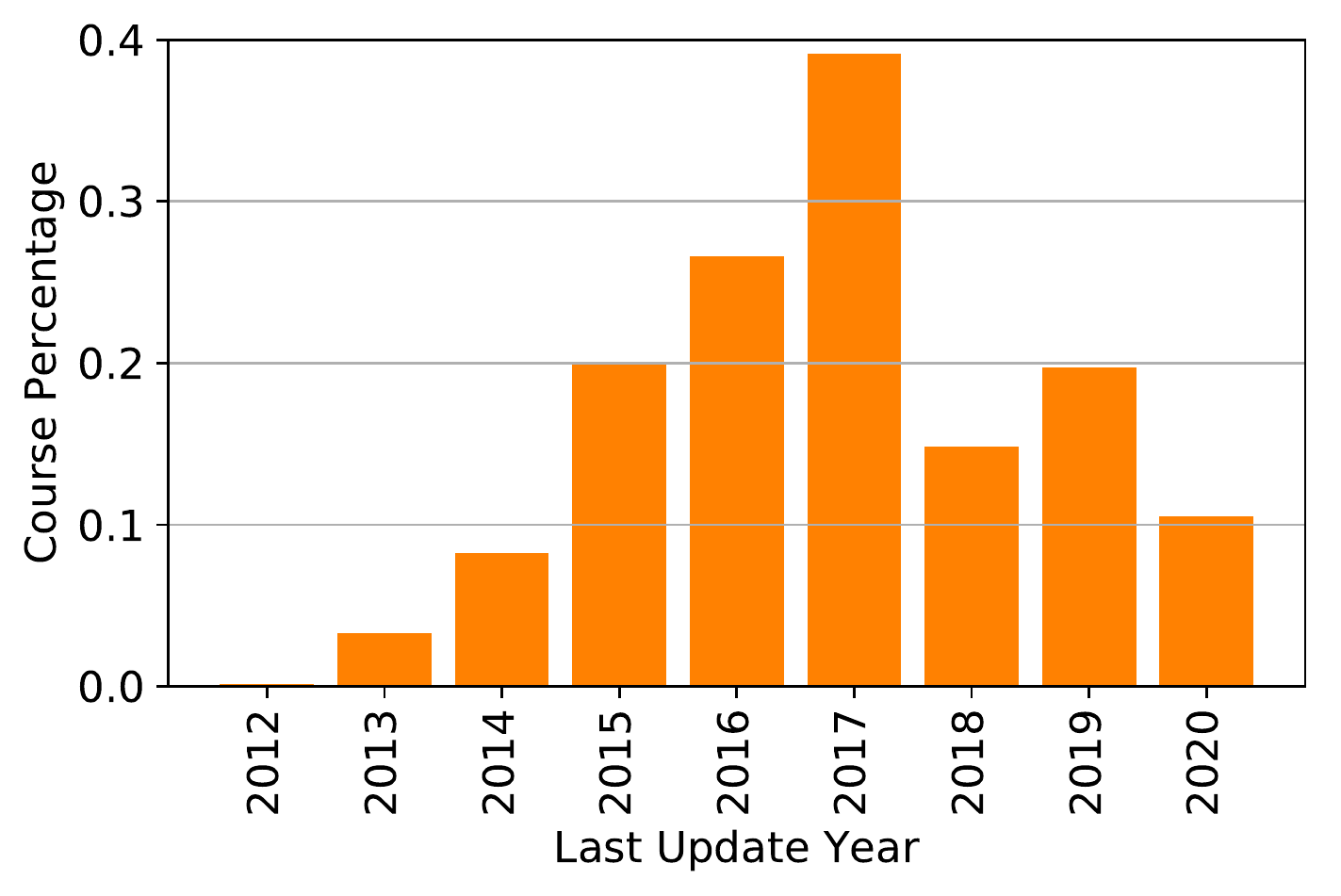}
    \caption{Last Update : Validity}
\end{subfigure}
\begin{subfigure}[t]{0.32\linewidth}
    \centering
    \includegraphics[width=1.0\linewidth]{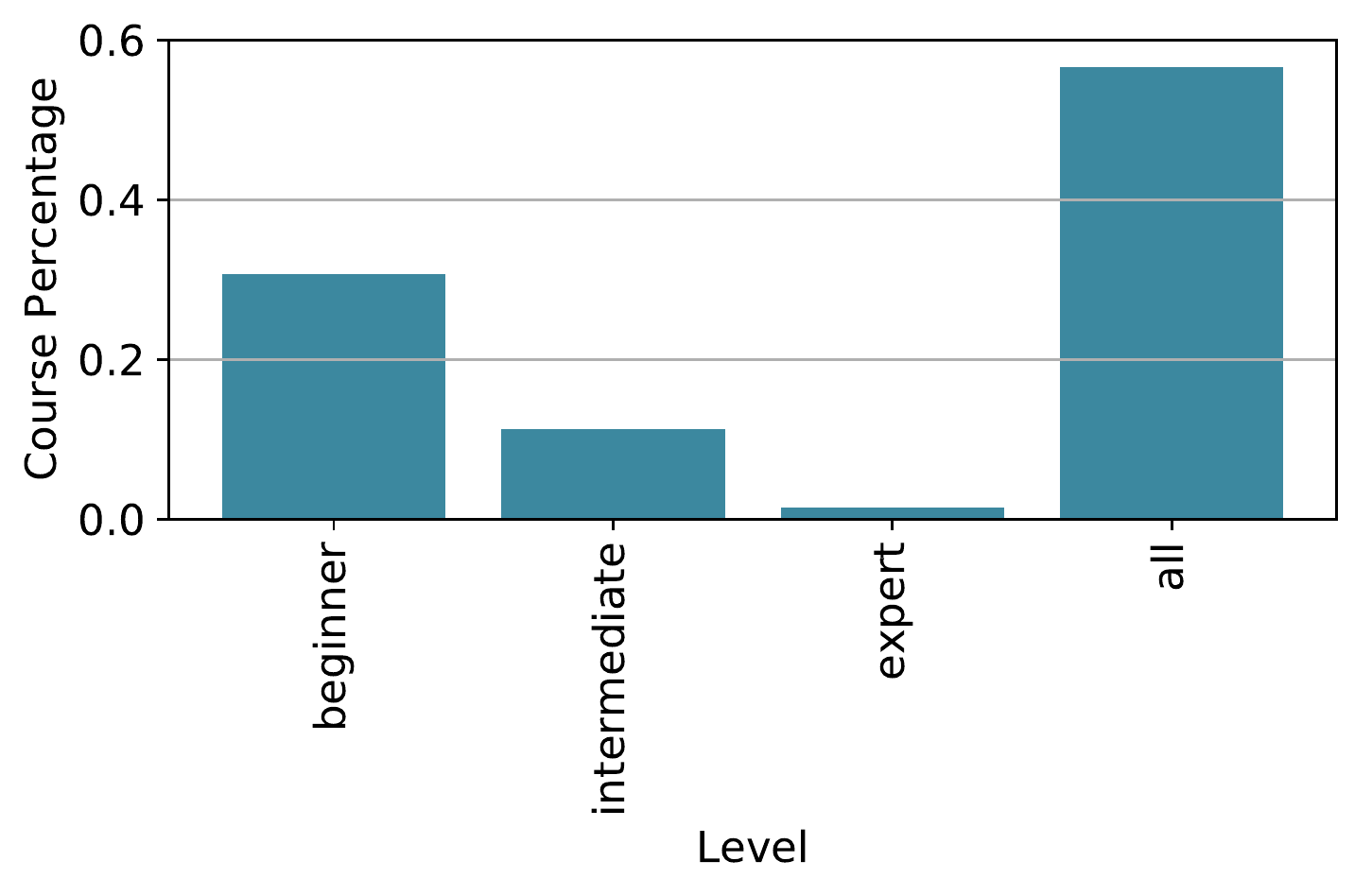}
    \caption{Level : Learnability}
\end{subfigure}
\begin{subfigure}[t]{0.32\linewidth}
    \centering
    \includegraphics[width=1.0\linewidth]{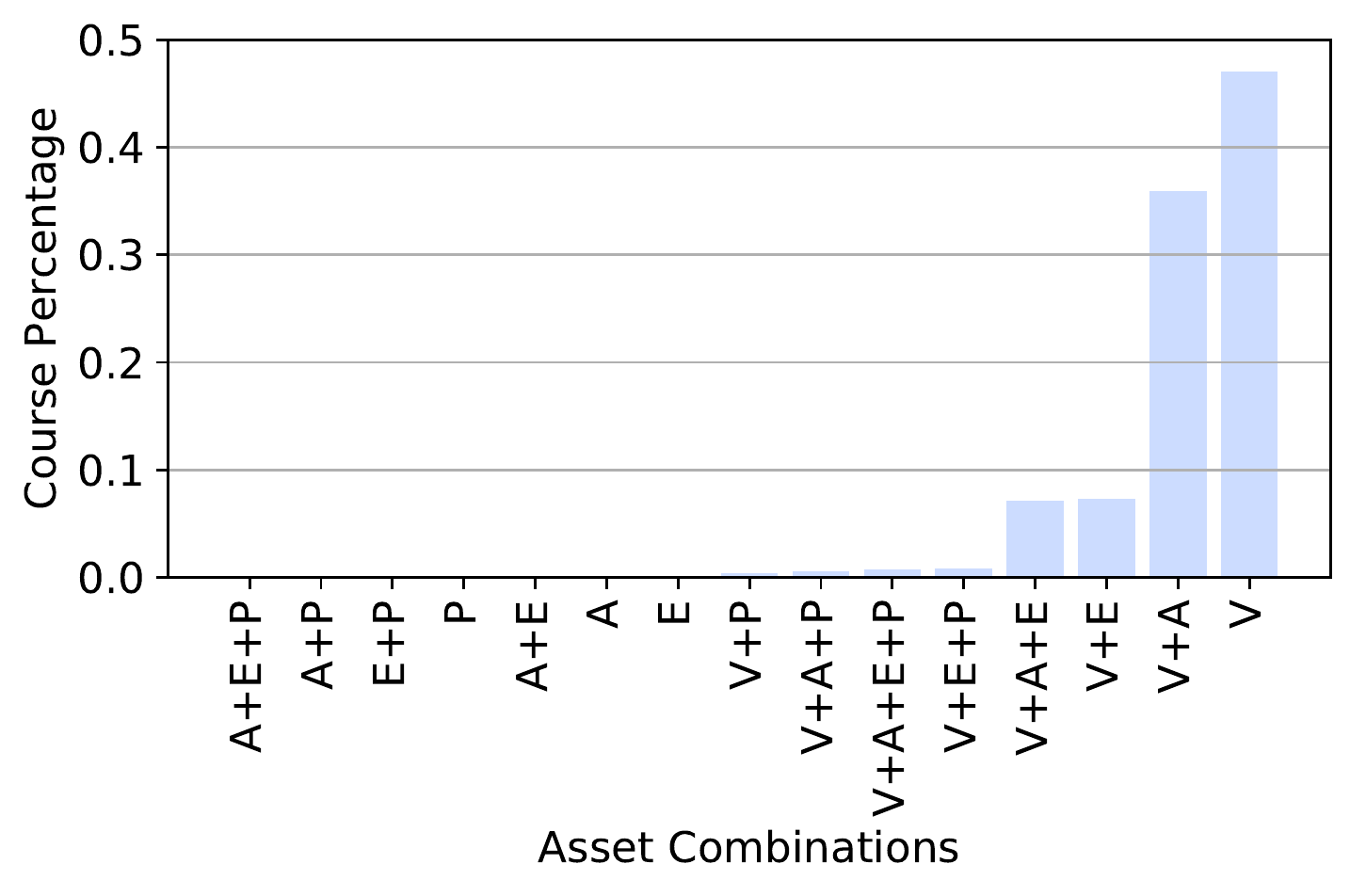}
    \caption{Asset Types : Variety}
\end{subfigure}
\begin{subfigure}[t]{0.32\linewidth}
    \centering
    \includegraphics[width=1.0\linewidth]{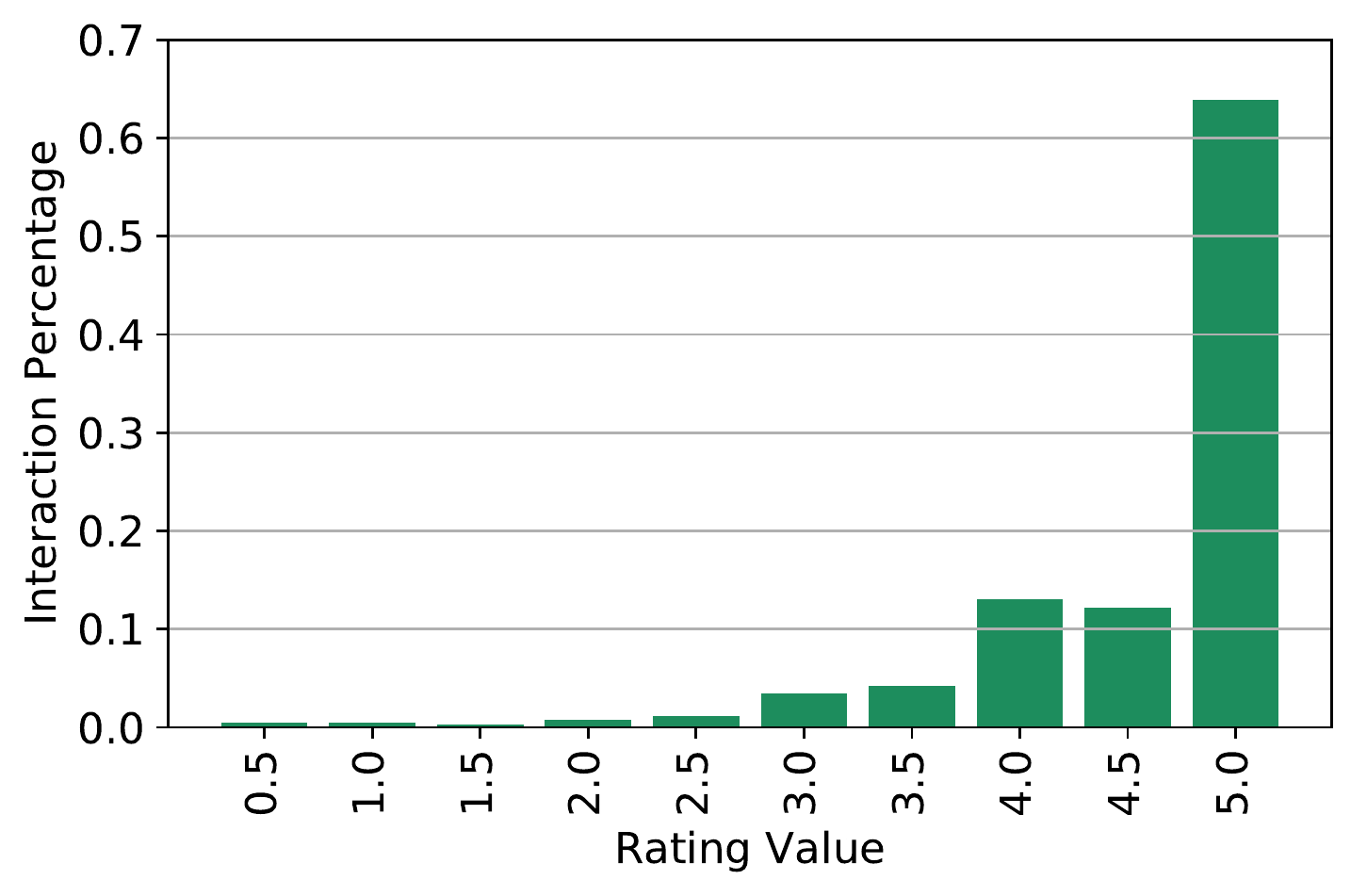}
    \caption{Rating Value : Quality}
\end{subfigure}
\begin{subfigure}[t]{0.32\linewidth}
    \centering
    \includegraphics[width=1.0\linewidth]{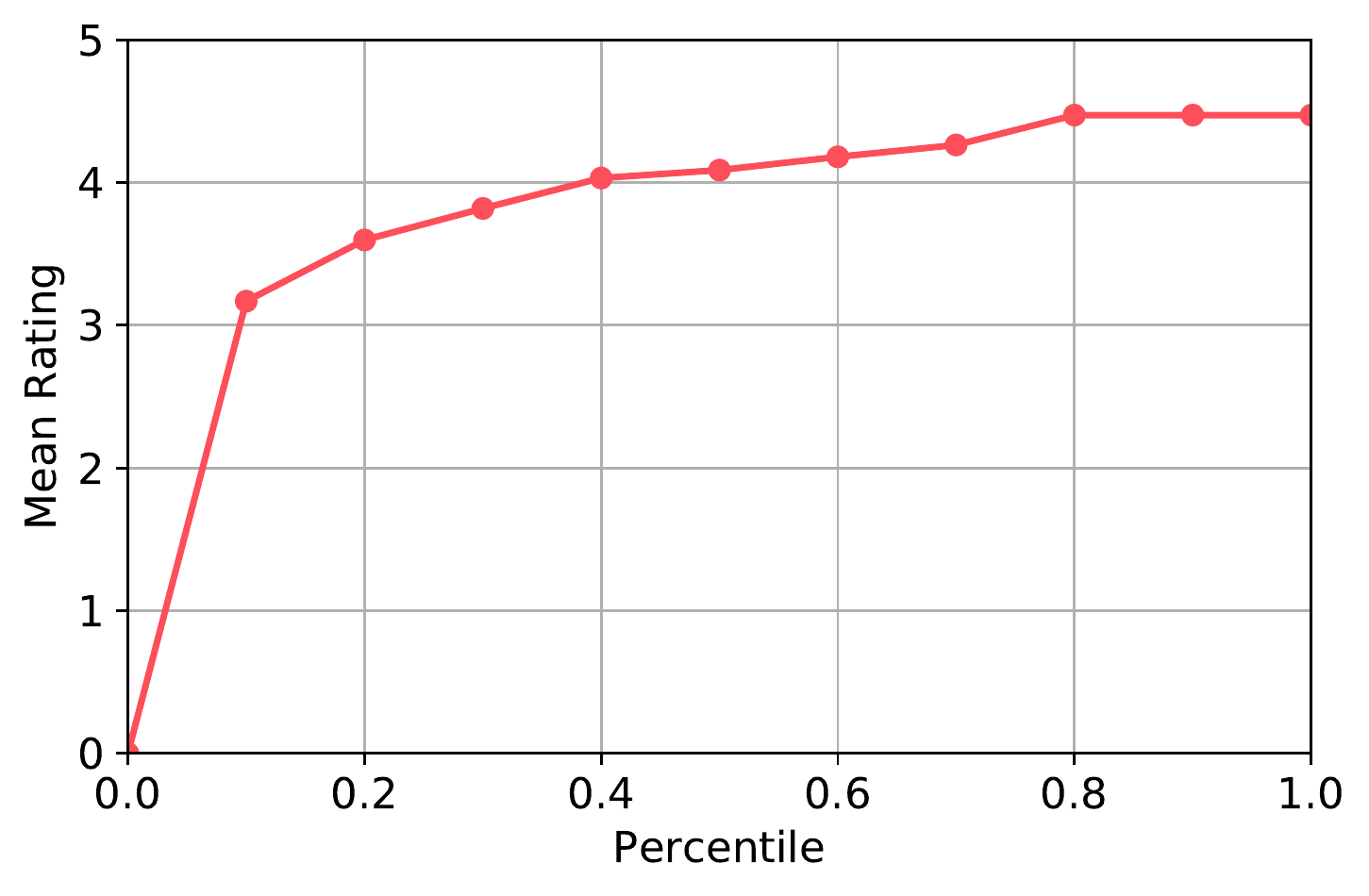}
    \caption{Mean Rating : Quality}
\end{subfigure}
\begin{subfigure}[t]{0.32\linewidth}
    \centering
    \includegraphics[width=1.0\linewidth]{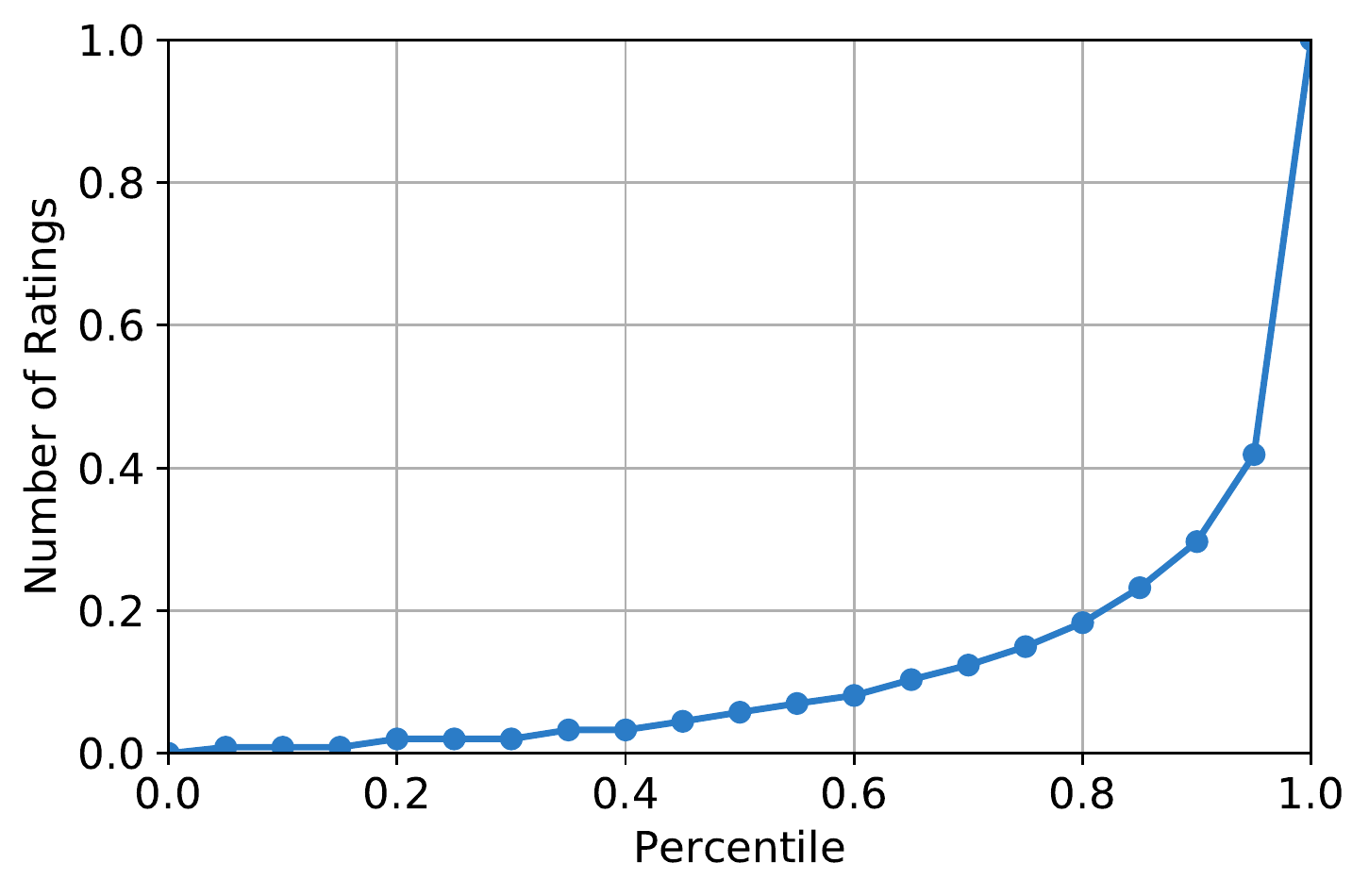}
    \caption{Popularity : Massiveness}
\end{subfigure}
\begin{subfigure}[t]{0.32\linewidth}
    \centering
    \includegraphics[width=1.0\linewidth]{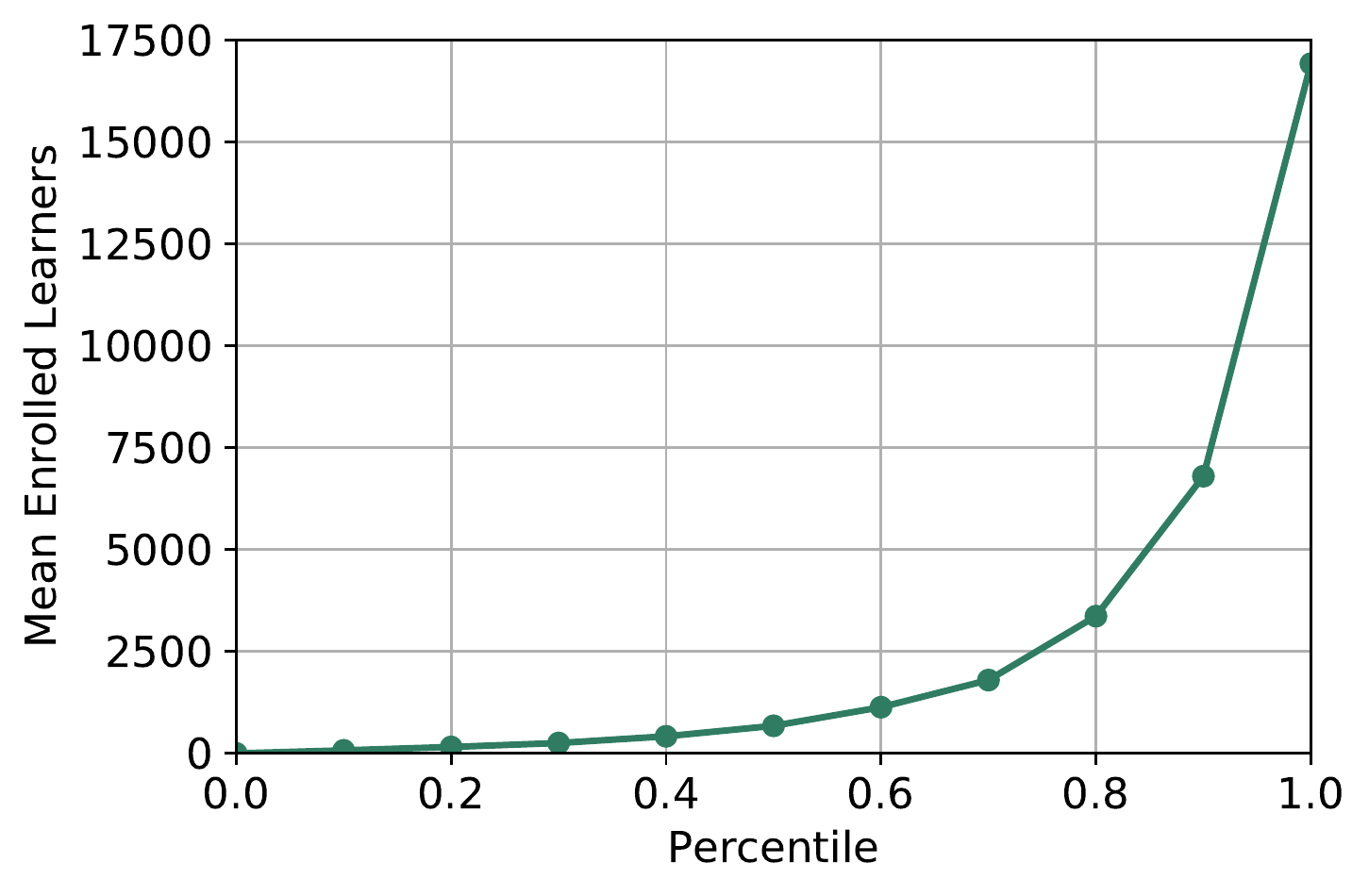}
    \caption{Enrolments : Massiveness}
\end{subfigure}
\begin{subfigure}[t]{0.32\linewidth}
    \centering
    \includegraphics[width=1.0\linewidth]{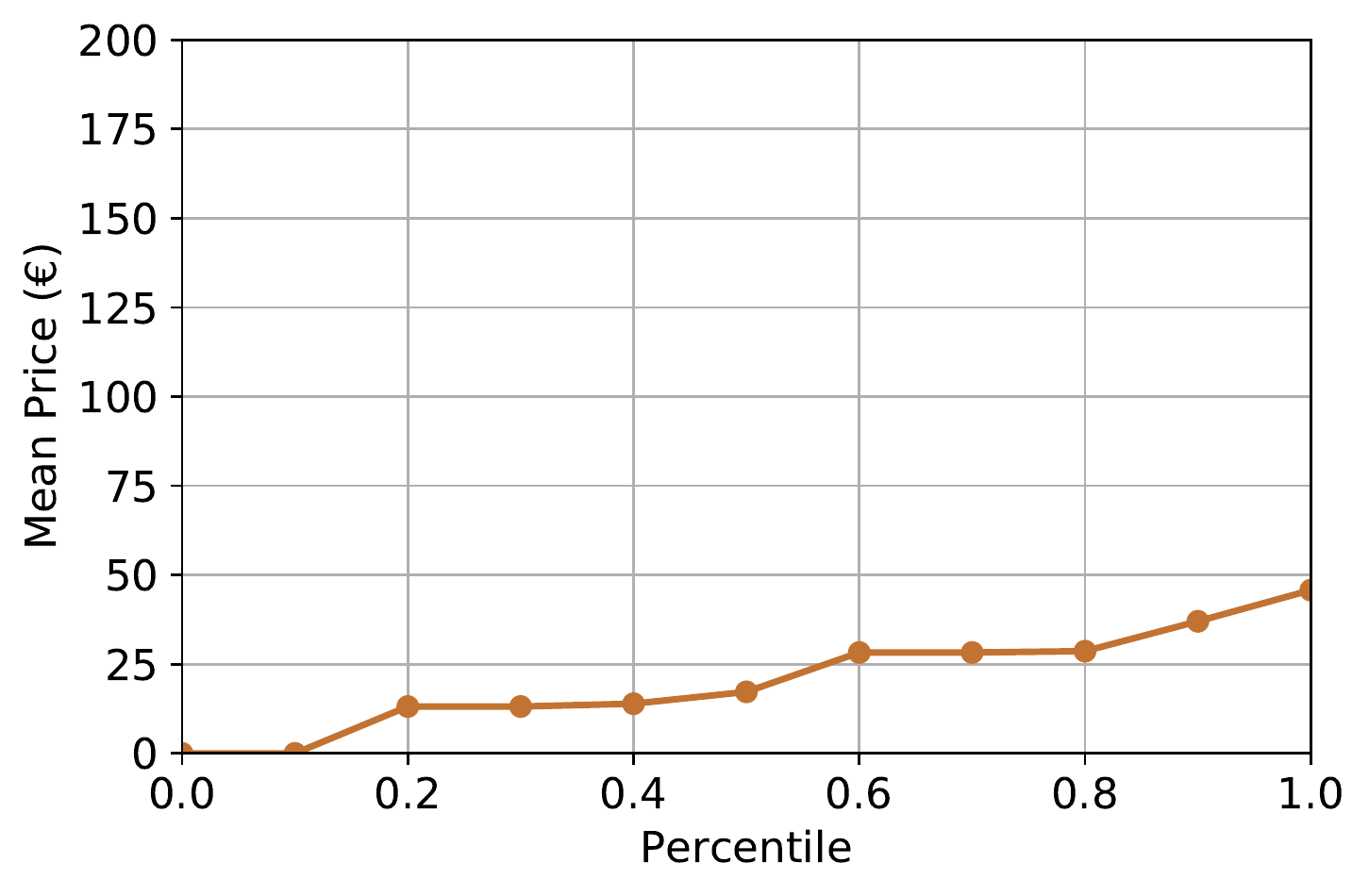}
    \caption{Prices : Affordability}
\end{subfigure}
\caption{\textbf{Data Statistics}. Characteristics of the real-world dataset relevant to the learning opportunity principles proposed by this paper: course popularity, rating values, last update timestamp, thematic category, instructional level, asset types (V:Video; A:Article; E:Ebook; P:Podcast), prices, number of enrolments per course, and average rating per course. Subfigure captions specify the feature and the interested principle as $<Feature>$:$<principle>$.} \label{fig:train-test-data}
\end{figure}

\subsection{Recommendation Algorithms and Protocols} \label{sec:algo}
We considered ten methods and investigated the recommendations they generated. 
Two of them are baseline recommenders, and the other eight are state-of-the-art algorithms, chosen due to their performance, wide adoption, and core applicability in learning contexts~\citep{kulkarni2020recommender}. 
These algorithms are:

\begin{itemize}
\item Non-Personalized: Random and TopPopular.
\item Neighbor-based: UserKNN and ItemKNN \citep{sarwar2001item}.
\item Matrix-Factorized: GMF \citep{he2017neural}, NeuMF \citep{he2017neural}.
\item Graph: P3-Alpha \citep{cooper2014random} and RP3-Beta \citep{paudel2016updatable}.
\item Content: ItemKNN-CB \citep{lops2011content}.
\item Hybrid: CoupledCF \citep{zhang2018coupledcf}.
\end{itemize} 

Based on hyperparameter tuning, UserKNN and ItemKNN relied on cosine metric and 100 neighbors. 
GMF and NeuMF used 10 factors and were trained on 4 negative samples per positive instance. 
This means that, for each observed user-item interaction, we added to the training set four user-item pairs where the selected item has been never observed by that user in the dataset. 
P3-Alpha was executed with 0.8 alpha and 200 neighbors, while RP3-Beta adopted 0.6 alpha, 0.3 beta, and 200 neighbors. 
ItemKNN-CB mapped course descriptions to Term-Frequency Inverse-Document Frequency (TF-IDF) features. 
The TF-IDF features of courses into the user's profile were averaged, and their cosine similarity with the TF-IDF features of other courses is used during ranking. 
CoupledCF embedded user-item associations, the user tendency to interact with each category of courses, and the category of the course in the current user-item pair. 
To be as close as possible to a real scenario, we used a fixed-timestamp split \citep{campos2014time}. 
The basic idea is to choose a single timestamp that represents the moment in which test learners are on the platform waiting for recommendations. 
Their past corresponds to the training set, and the performance is evaluated with data coming from their future. 
In this work, we select the splitting timestamp \emph{2017-06-08}, which maximizes the number of learners involved in the evaluation, by setting two constraints: the training set must keep at least 4 ratings per user, and the test set must contain at least 1 ratings. 
This split led to 43,021 learners, 24,321 courses, and 529,857 interactions (\figurename~\ref{fig:train-test-data}). Normalized Discounted Cumulative Gain (NDCG)\footnote{Discounted cumulative gain (DCG) is a measure of ranking quality. 
The premise of DCG is that highly relevant courses appearing lower in a recommend list should be penalized, given that the graded relevance is logarithmically proportional to the ranking position. 
To let DCG be independent of the ranking length, DCG is normalized by scaling the results based on the best possible value, i.e., ideal DCG. 
The latter is computed by sorting all relevant courses in the test set by their relevance, producing the maximum possible DCG.} is used as effectiveness metric. 
As a measure of relevance for calculating NDCG, the binarized (u, i) tuples formalized in Section \ref{sec:problem-formulation} were used\footnote{It may happen that many more courses would be relevant to a given learner than that learner will have interacted with. 
Some courses that end up high on a recommended list for a given learner but that this learner did not see or have time for it would not be relevant.  Dealing with this well-known problem of missing-not-at-random interactions is an open problem in the recommendation landscape~\citep{nakagawa2008missing}.}. 

\subsection{Real-World Observations} 
We characterize how the proposed principles were met in the lists of courses suggested by the considered algorithms. 
Student-specific targeted weights for each principle would be elicited through user groups, surveys, or implicit preferences observed in the collected data. 
However, due to the absence of this form of feedback in COCO and given that the preference of each learner derived from historical data might have been biased by the recommender system itself, we consider a scenario where the educational platform aims to maximize all the targeted principles, i.e., $p_u = \mathbbm{1}^{|C|}, \, \forall \, u \in U$. 
To this end, we assume to give the same maximum weight to all the principles, i.e., $w_i = \mathbbm{1}^{|C|}, \, \forall \, u \in U$. 
While this assumption comes with some limitations described in Section \ref{sec:limitations}, given that each learner does not always prefer maximum familiarity, for example, such setup allows us to quantify the extent to which each principle is met. 
We leave experiments on learner-specific weights elicited through interviews or surveys as future work. Three research questions drove our analysis:

\begin{enumerate}
\item Does a relation exist between consistency and equality? If so, which one? 
\item Which principles have the largest impact on consistency and equality? 
\item Are consistency and equality influenced by the past interactions of learners? 
\end{enumerate}

\vspace{2mm} \noindent \textit{Equality Analysis}. We provided recommendations to all learners, suggesting to each learner $k=10$ courses; then we measured consistency across the whole learners' population , i.e., how much the principles were met in the recommendations of all learners (Eq. \ref{eq:global-consistency}), and equality, i.e., how much the consistencies were similar across learners (Eq. \ref{eq:equality}). Table \ref{tab:ndcg-cons-equal} reports global performance. A higher value indicates that a recommender better drives consistency or equality, respectively. A first observation coming from Table \ref{tab:ndcg-cons-equal} is the following:

\vspace{2mm} \setlength{\fboxsep}{10pt}
\noindent \textbf{Observation 1}. \textit{Recommenders that embed content metadata ensure higher equality across learners. When the recommender uses only user-item interactions, the equality is reduced. This holds regardless of the algorithm's subfamily.}\\

Though the observation above holds under our setting, the values associated with the equality of the recommender systems and the mean consistency values associated with each principle do not reveal much about how consistency estimates are equal across individual learners. 

To provide a more detailed picture, we plot consistencies across learners for each algorithm, sorted by increasing values (\figurename~\ref{fig:consistency-distr-a}). 
It can be observed that ItemKNN-CB and CoupledCF showed more equal consistencies across learners. 
This result might depend on the fact that, in the presence of principles related to the course content, the content-based and hybrid methods may involuntarily increase those principles and lead to higher consistency. 
In other words, their equality could be biased by the fact that they capitalize on input information that is related to some principles. 

While it may happen that certain principles are being optimized by a traditional recommendation algorithm involuntarily, it is  generally impractical to arrange the internal logic of an algorithm to accommodate all the targeted principles. 
Figure~\ref{fig:consistency-distr-b} plots then the consistency error bars for each algorithm, with mean, standard deviation, minimum, and maximum values. 
It could be observed that there is a link between the magnitude of the mean and of the standard deviation. 
More precisely, the higher the mean consistency guaranteed by the algorithm, the lower the standard deviation across consistency values is (\figurename \ref{fig:consistency-distr-c}). Hence, we can draw a subsequent observation:

\begin{table}[!t]
\resizebox{\textwidth}{!}{%
\begin{tabular}{c|cc|cc|cc|cc|c|c}
\hline
\multirow{2}{*}{} & \multicolumn{2}{c}{\textbf{Non-Personalized}} & \multicolumn{2}{c}{\textbf{Neighborhood}} & \multicolumn{2}{c}{\textbf{Graph}} & \multicolumn{2}{c}{\textbf{Matrix-Factor}} & \textbf{Content} & \textbf{Hybrid} \\ 
 & Random & TopPopular & UserKNN & ItemKNN & P3Alpha & RP3Beta & NeuMF & GMF & ItemKNNCB & CoupledCF \\ 
\hline 
\textbf{NDCG} & 0.000 & 0.035 & \textit{0.072} & 0.021 & 0.001 & 0.000 & 0.008 & 0.010 & 0.042 & 0.013 \\
\textbf{Consistency} & 0.586 & 0.516 & 0.618 & 0.615 & 0.578 & 0.572 & 0.662 & 0.652 & 0.699 & \textit{0.717} \\
\textbf{Equality} & 0.872 & 0.795 & 0.885 & 0.891 & 0.850 & 0.847 & 0.917 & 0.906 & 0.959 & \textit{0.969} \\
\hline
\end{tabular}}
\caption{\textbf{Global Indicators}. Normalized Discounted Cumulative Gain (NDCG), consistency across the whole learners' population , and equality produced by different families of recommenders. Italic values highlight the highest value for each metric across algorithms. The highest NDCG is achieved by UserKNN, while the highest consistency and equality was observed for CoupledCF.}
\label{tab:ndcg-cons-equal}
\end{table}

\begin{figure}[!t]
\begin{subfigure}[t]{0.325\linewidth}
    \centering
    \includegraphics[width=1.0\linewidth]{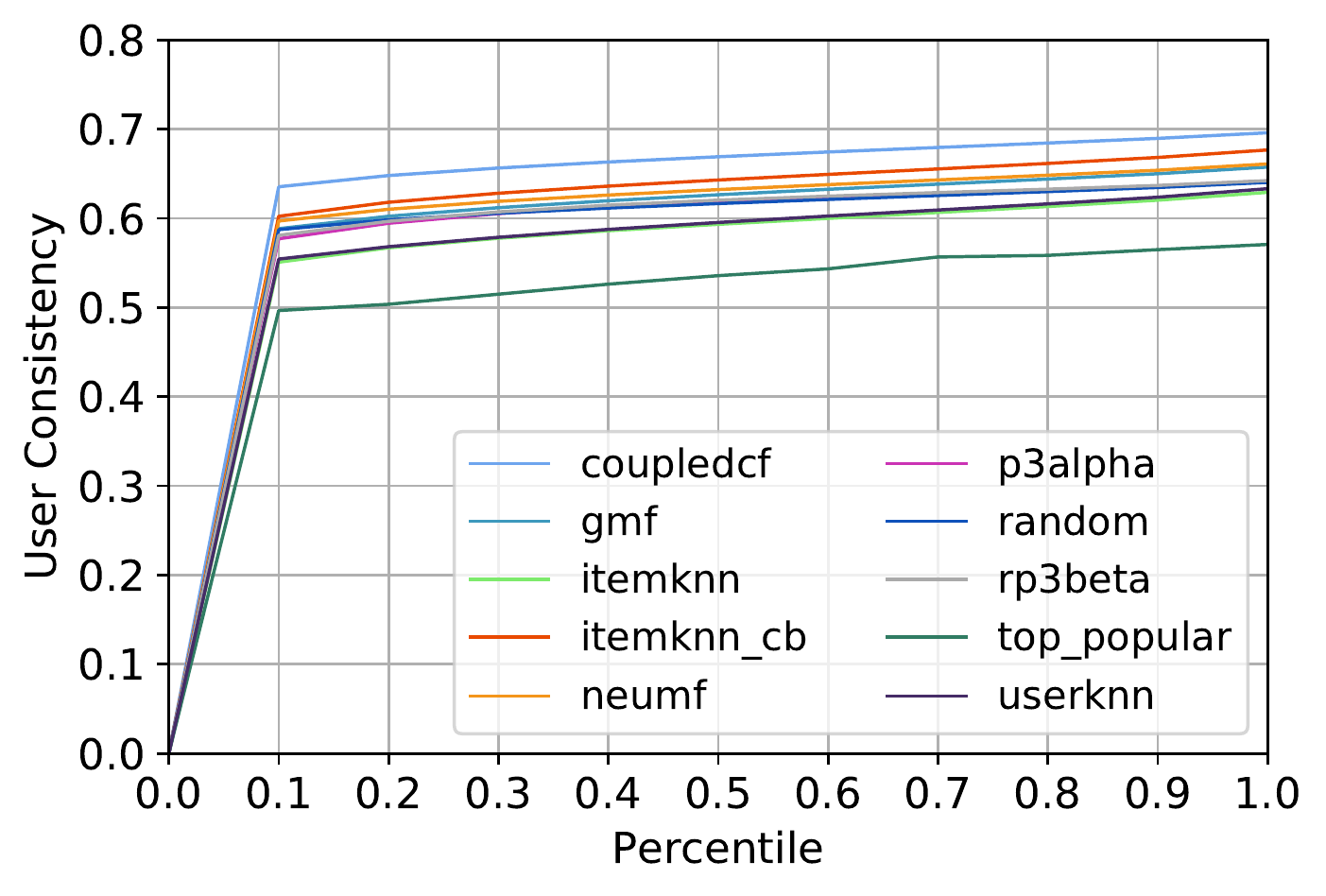}
    \caption{Consistency distribution.}
    \label{fig:consistency-distr-a}
\end{subfigure}
\begin{subfigure}[t]{0.325\linewidth}
    \centering
    \includegraphics[width=1.0\linewidth]{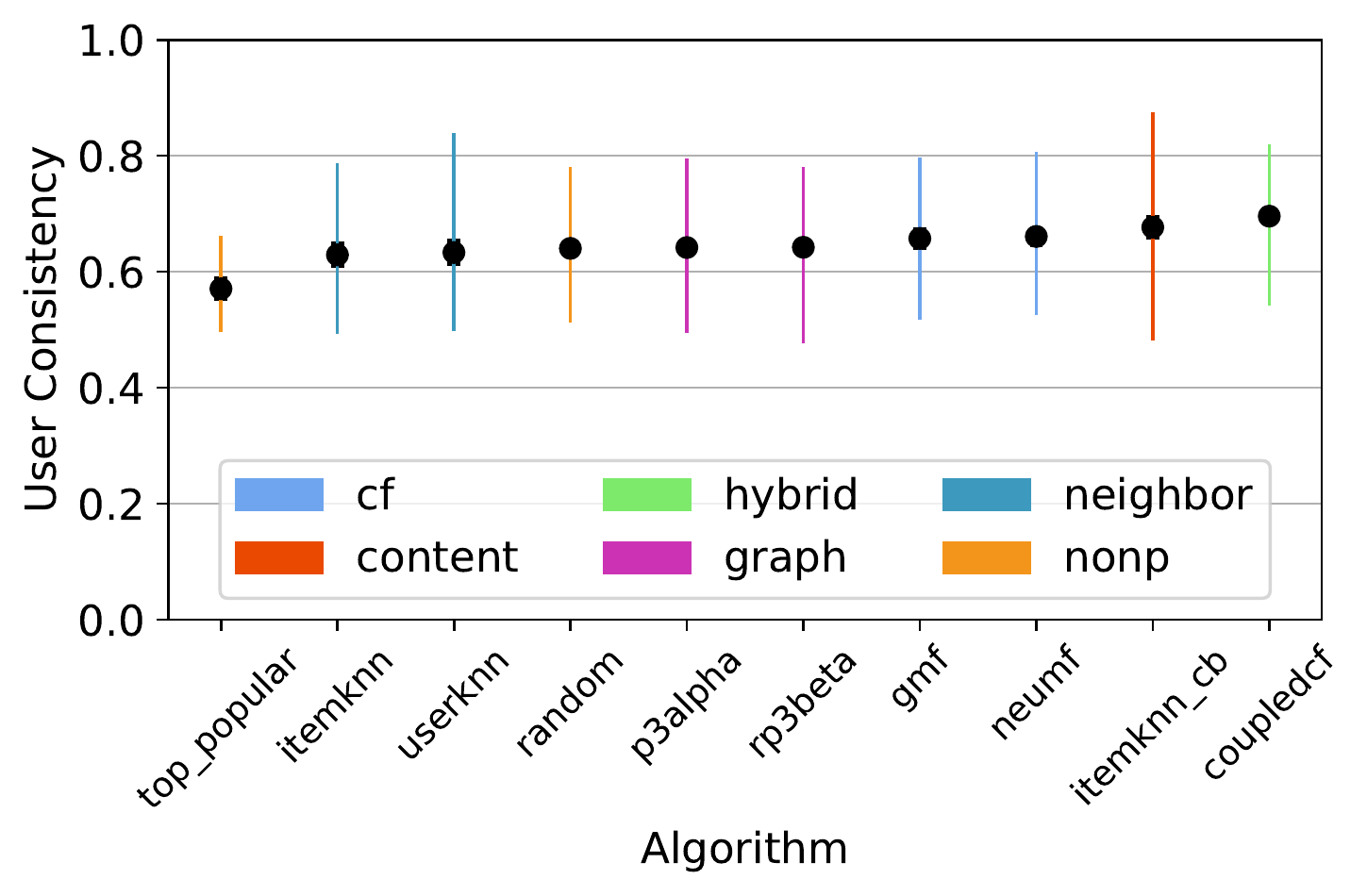}
    \caption{Consistency errors bars.}
    \label{fig:consistency-distr-b}
\end{subfigure}
\begin{subfigure}[t]{0.325\linewidth}
    \centering
    \includegraphics[width=1.0\linewidth]{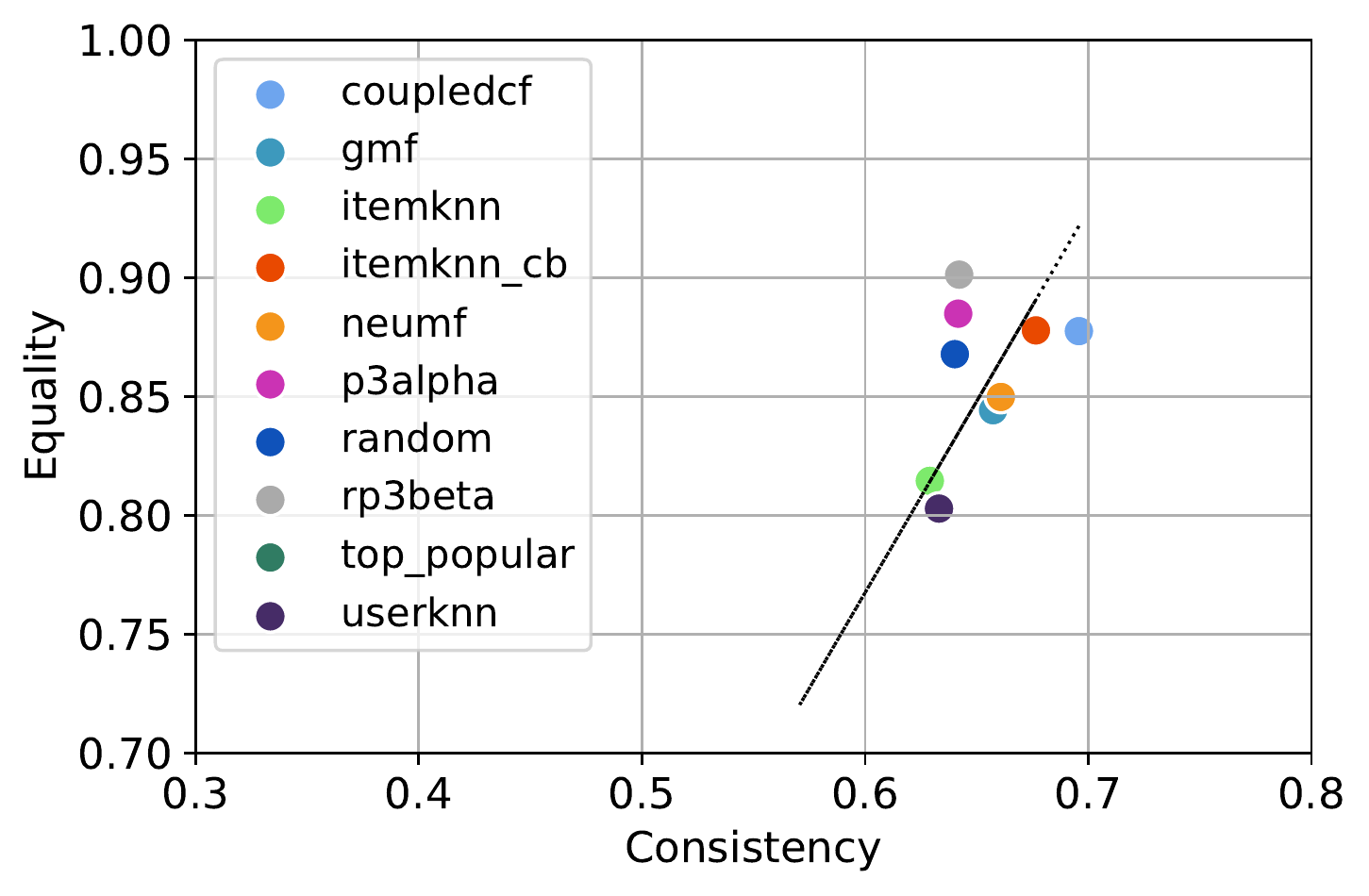}
    \caption{Equality w.r.t. Consistency.}
    \label{fig:consistency-distr-c}
\end{subfigure}
\caption{\textbf{Consistency over the Entire Population}. On the left plot, lines represent the consistency distribution over learners, sorted in increasing order. On the center plot, each error bar includes mean (dot), std deviation (black solid line), and min-max values (colored thick line). The right plot highlights the direct relation between consistency and equality.}
\label{fig:consistency-distr}
\end{figure}

\vspace{2mm} \setlength{\fboxsep}{10pt}
\noindent \textbf{Observation 2}. \textit{Recommenders with high consistency lead to higher equality of recommended learning opportunities. Such property is stronger for neural collaborative, content-based, and hybrid recommenders.}\\

\begin{figure}[!b]
\centering
\includegraphics[width=1.0\linewidth]{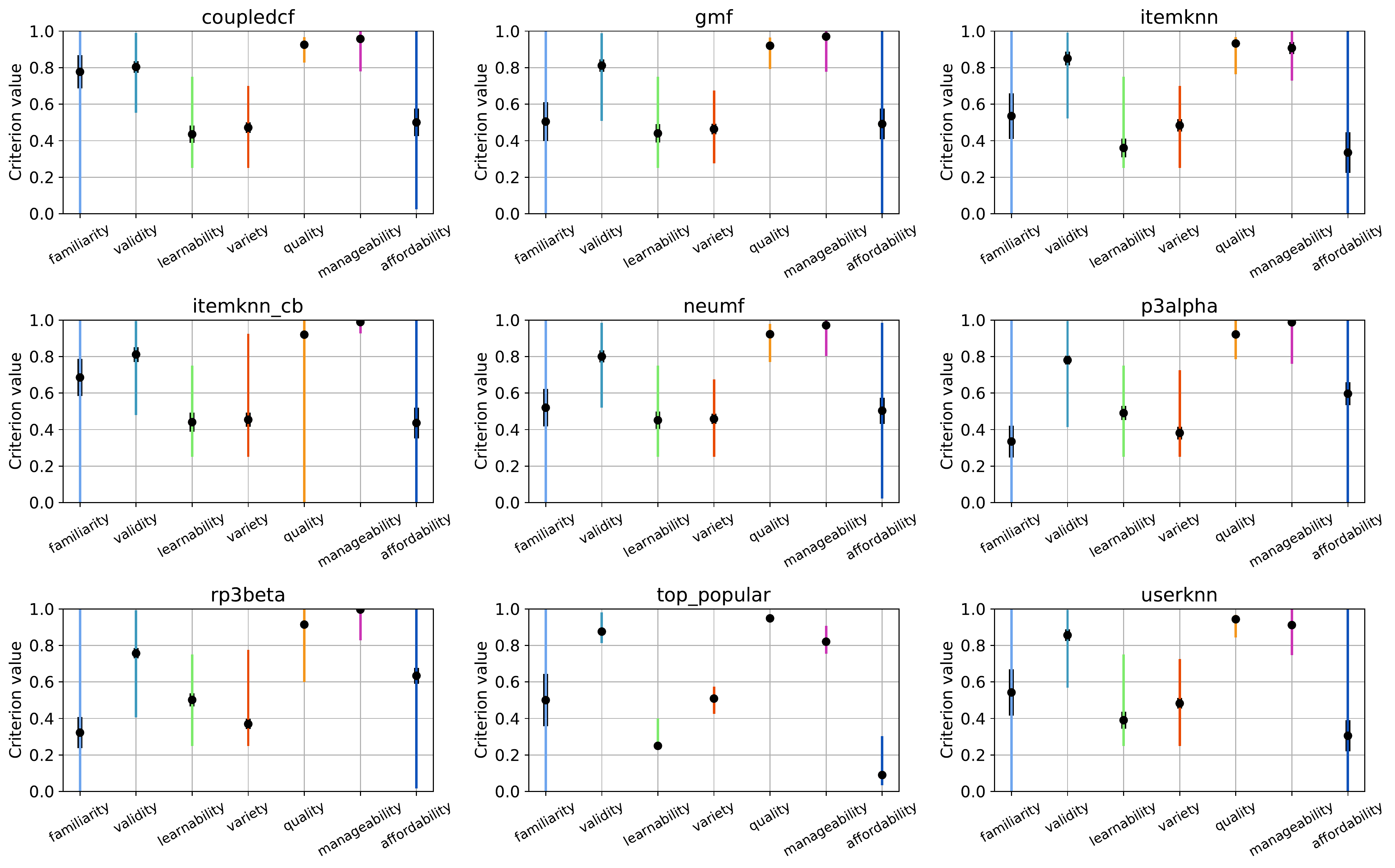}
\caption{\textbf{Algorithm over Principle}. For each recommendation algorithm, the corresponding plot reports an error bar for each principle as measured for that algorithm, including mean (dot), std deviation (solid black line), and min-max values (thick colored line).}
\label{fig:model-by-principles}
\end{figure}

\begin{figure}[!b]
\centering
\includegraphics[width=1.0\linewidth]{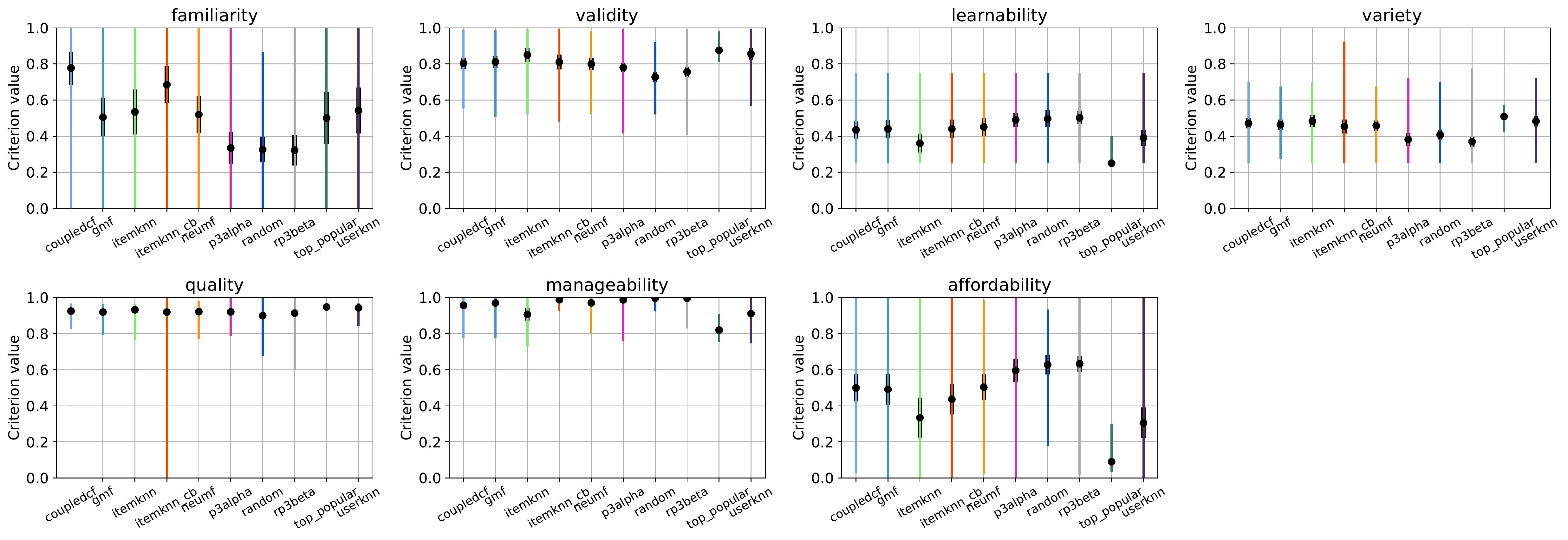}
\caption{\textbf{Principle over Algorithms}. For each principle, the corresponding plot reports an error bar for each algorithm as measured for that algorithm, including mean (dot), std deviation (solid black line), and min-max values (thick colored line).}
\label{fig:principles-by-model}
\end{figure}

Uncovering such a link between a metric that requires knowledge about the whole learner population (i.e., equality) and a metric that can be directly optimized on a single ranked list (i.e., consistency) makes it possible to apply a non-NP-Hard re-ranking procedure to solve our task. This suggests that we should investigate the interplay between (i) the average consistency across principles and (ii) the consistency achieved for each principle individually, when a given learner and algorithm are considered.

\begin{figure}[!b]
\centering
\includegraphics[width=1.0\linewidth]{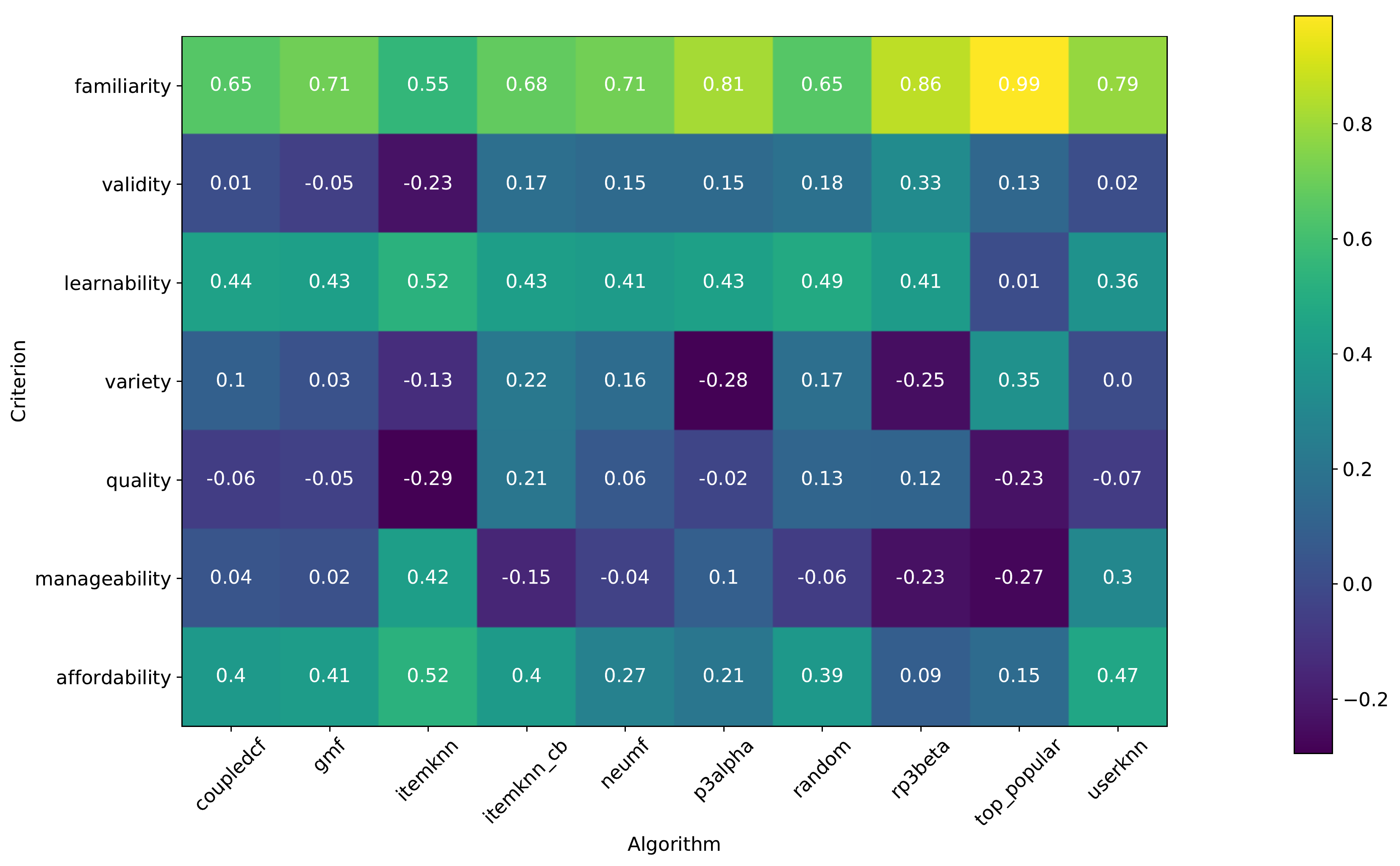}
\caption{\textbf{Principle-Consistency Relation}. Heatmap of correlations between the consistency for a given principle and the consistency across the whole learners' population  computed for each user, over different algorithms. Each value ranges in [-1, +1], and for each principle and algorithm, the Spearman correlation is computed over a distribution of (principle value, user consistency) pairs.}
\label{fig:cri-vs-cons}
\end{figure}

\vspace{2mm} \noindent \textit{Individual Principle Analysis}. 
We further investigated the first two observations by showing the consistency of each recommendation algorithm over every single principle. 
For the sake of readability and conciseness, we do not further consider the Random algorithm over the analysis. \figurename~\ref{fig:model-by-principles} reports the mean, standard deviation, minimum, and maximum values over each principle on that recommender. For instance, the \texttt{coupledcf} plot shows that the \texttt{familiarity} score has a mean of $0.80$, a standard deviation of $\pm 0.05$, and spans the whole range (min: $0.00$; max: $1.00$). The first observations can be made for the top popular algorithm, whose results reveal that popular courses are mostly fresh (high validity) and of high quality. 
However, the consistency on these two principles comes at the price of low familiarity, learnability, variety, and affordability. 
Considering algorithms that capitalize on course metadata (CoupledCF and ItemKNN-CB), similar patterns can be observed across principles, except on variety and quality. 
For the latter principles, embedding user-item interactions in CoupledCF made it possible to reduce the min-max gap. 
Hence, situations where few learners experience very high or very low values can be avoided. 
Other algorithms achieved a more stable consistency.  

To assess whether certain algorithms favor or hurt a given principle, \figurename~\ref{fig:principles-by-model} reports for each principle how it varies over algorithms. It can be observed that familiarity and affordability suffer from high deviations, while more stable values were measured for other principles, over algorithms. We conjecture that the stability observed on quality comes from the highly unbalanced rating value distribution. Indirectly, this effect could come from the fact that learners tend to evaluate courses with high ratings, when they decide to rate them. \figurename~\ref{fig:principles-by-model} also confirmed this intuition. On principles like affordability, manageability, and learnability, the two mentioned algorithms got lower values. 

\vspace{2mm} \setlength{\fboxsep}{10pt}
\noindent \textbf{Observation 3}. \textit{Quality, validity, and manageability are guaranteed to high extent by different recommenders, regardless of the family. Familiarity, affordability, learnability, and variety experience low absolute values and substantial deviations over algorithms, regardless of the algorithm's subfamily.}\\

To further confirm the role of each principle over principles' consistency, we looked at the correlation between the consistency achieved for a given principle and the consistency achieved by including all the principles. 
In \figurename~\ref{fig:cri-vs-cons}, we report the results for each principle and algorithm pair. 
Observed values higher than 0 are expected when the consistency at principle level is directly related to the high consistency achieved when all the principles are considered. 
Hence, the overall consistency is more likely to be met when that specific principle is met. 
Conversely, values lower than 0 result in the opposite behavior. 
No relation is found when the reported value is close to 0. 
This allows us to make another observation: 

\vspace{2mm} \setlength{\fboxsep}{10pt}
\noindent \textbf{Observation 4}. \textit{Familiarity, learnability, and affordability are the most influencing principles on the overall consistency across principles. This effect is stronger for content-based and hybrid recommenders.}\\        

\begin{figure}[!b]
\centering
\includegraphics[width=0.5\linewidth]{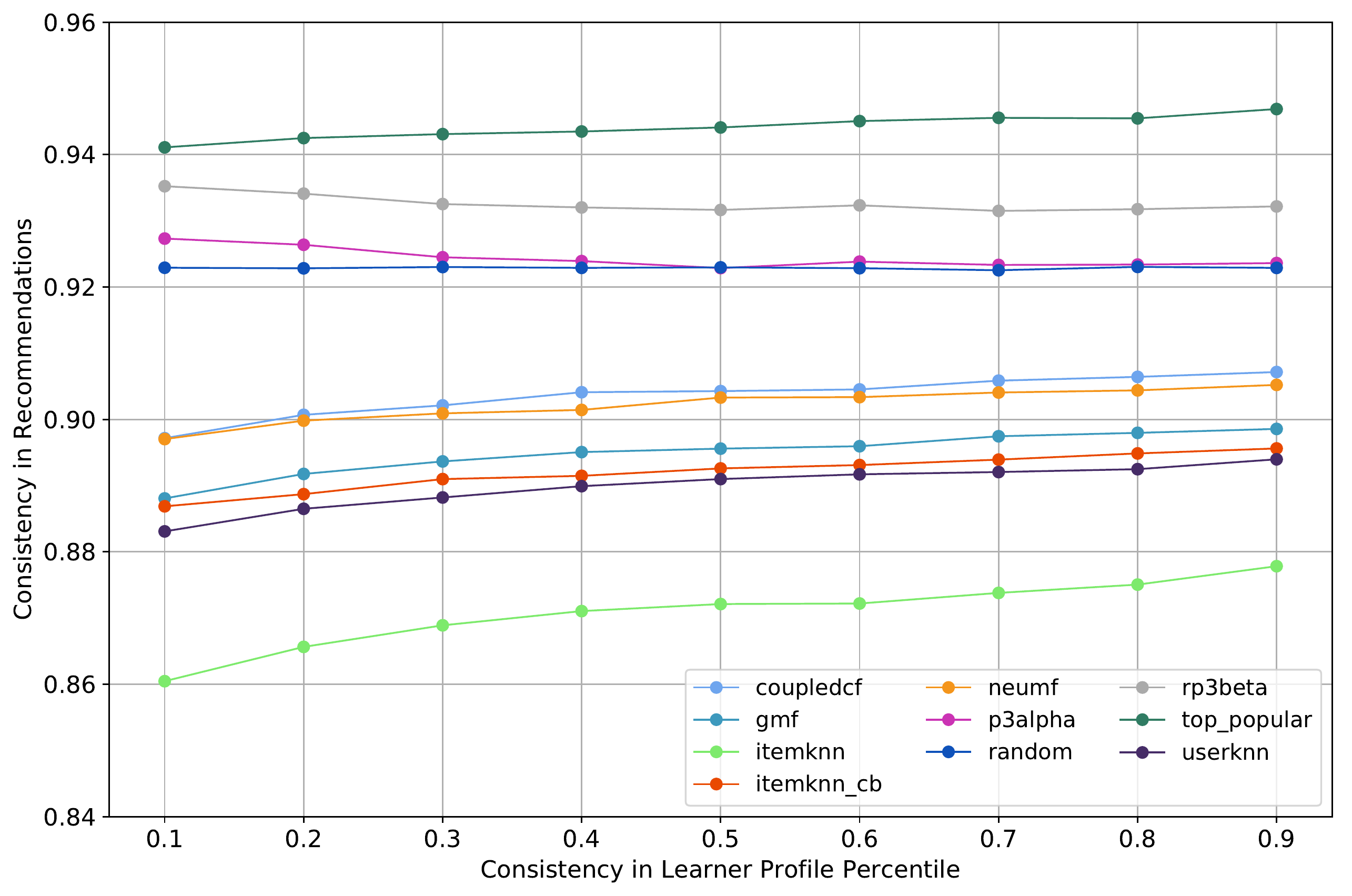}
\caption{\textbf{Consistency in Profile and Recommendation}. The lines show the difference in recommendation consistency over random pairs of learners, with values sorted by increasing difference in consistency in profiles.}
\label{fig:learner-fair}
\end{figure}

\begin{figure}[!b]
\centering
\includegraphics[width=1.0\linewidth]{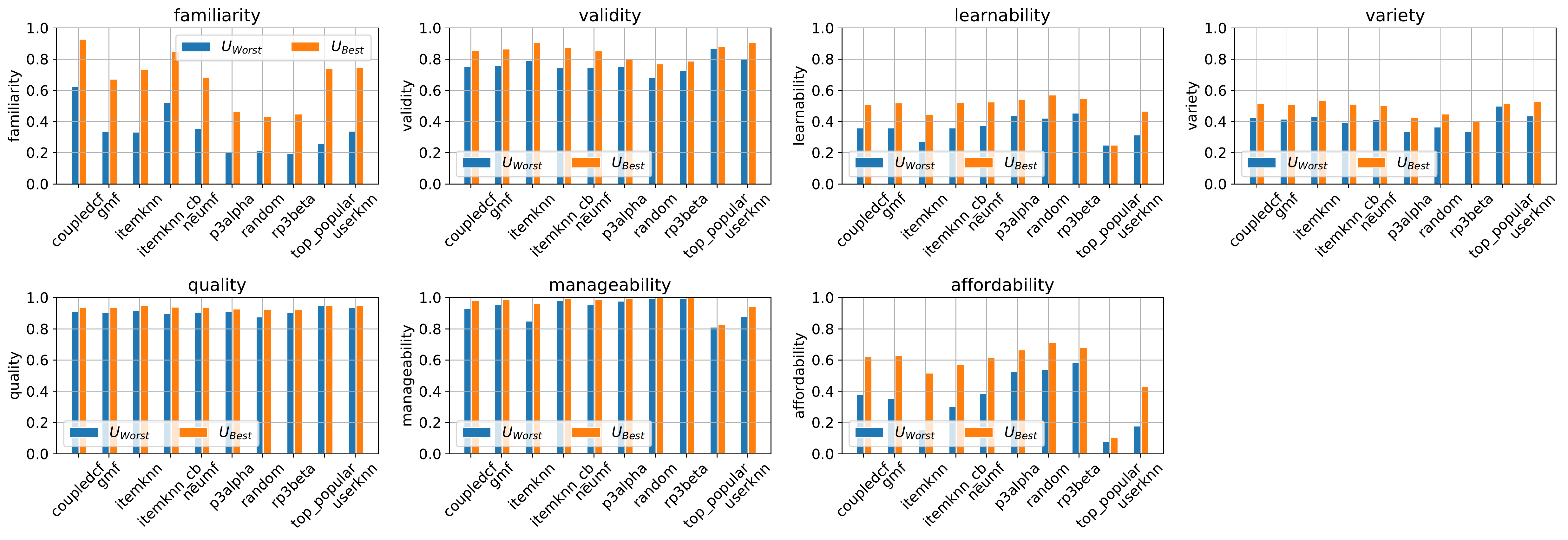}
\caption{\textbf{Learners with More (Less) Consistency}. For each principle, the corresponding plot reports the mean consistency achieved in recommendations by learners with high consistency in their profile (orange) and low consistency in their profile (blue).}
\label{fig:learner-best-worst}
\end{figure}

Most of the observations seen so far are based on the fact that the observed consistency values are averaged over learners. 
However, it is interesting to ask whether, for two learners with similar past interactions concerning the considered principles, we should expect a similar consistency. 
In other words, it is interesting to ask whether similar learners get similar consistency. 

\vspace{2mm} \noindent \textit{Past Interaction Analysis}. 
We next take a traditional individual fairness standpoint, i.e., the principle according to which similar individuals should receive a similar treatment. 
In our setting, for learners, we assume that being similar means having similar consistency in their past interactions. 
Therefore, we computed the consistency metric defined in Eq. \ref{eq:local-consistency} by substituting the vector $q_{\Tilde{I}_u}$ with the vector $q_{I_u}$ so that we can quantify how much the targeted principles were met in the set of past interactions of each learner.    

To this end, for all the possible pairs of learners, $u_1$ and $u_2$, we computed the difference of consistency in their profile and their recommendations. 
Figure~\ref{fig:learner-fair} depicts pairs of results by increasing the difference of consistency in their profiles. 
It can be observed that, except for the graph-based P3Alpha and RP3Alpha, a higher similarity of consistency between the profiles results in a higher similarity of consistency over the recommendations. 
Figure~\ref{fig:learner-best-worst} also shows, for each principle and algorithm, the best and worst consistency across learners, according to the above definition. 
It is confirmed that familiarity, learnability, variety, and affordability play a key role in the overall consistency. 

\vspace{2mm} \setlength{\fboxsep}{10pt}
\noindent \textbf{Observation 5}. \textit{Learners who interacted with courses aligned with the principles are likely to receive recommendations that meet those principles. Similar learners in terms of consistency in the courses they took are likely to receive a similar treatment in terms of future consistency.}\\

With the observations made so far, we conjecture that re-ranking each list of recommendations to maximize the considered principles will lead to  higher consistency, and, consequently, to higher equality. 

\section{Optimizing for Equality of Learning Opportunities} \label{sec:framework}
In this section, we describe, evaluate, and discuss the approach proposed in this paper to favor consistency of principles in recommendations (\figurename~\ref{fig:approach}).  

\begin{figure}[!b]
\minipage{1\linewidth}
    \centering
    \includegraphics[width=1.0\linewidth]{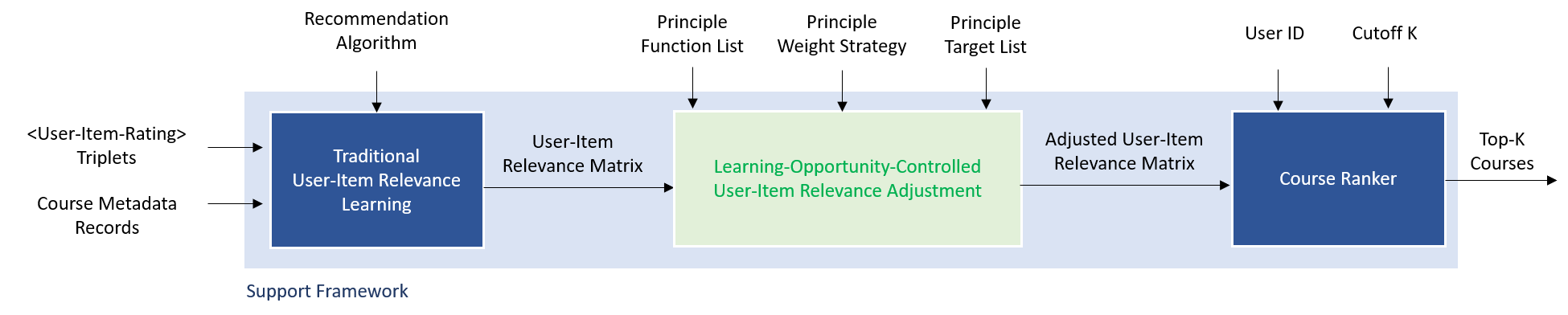}
\endminipage\hfill
\caption{\textbf{Support Framework}. First, given interactions and metadata, a recommendation algorithm computes a user-item relevance matrix. 
Then, given a user-item relevance matrix, a list of principle functions, a principle weight strategy, and a list of principle targets, our approach returns a user-item relevance matrix that meets the input principles. Finally, a ranking step, given the adjusted user-item relevance, outputs the recommended list.}
\label{fig:approach}
\end{figure} 

\subsection{The Proposed Post-Processing Approach}\label{sec:eq-control}
To meet the principles pursued by the platform for each learner and optimizing for equality of opportunities across learners, we introduce a recommendation procedure that seeks to maximize the consistency formalized in Eq. \ref{eq:local-consistency}. 
Since, in general, it is hard to plug-in the balancing phase into the internal logic of a recommender system, we propose to re-arrange the recommended lists returned by a generic recommender system, a common practice in recommendation known as re-ranking~\citep{potey2017personalization}. 
Specifically, for each learner $u\in U$, our goal is to determine an optimal set $\mathcal I^\ast$ of $k$ courses to be recommended to $u$, so that the principles pursued by the platform are met while preserving accuracy. 
To this end, we capitalize on a \emph{maximum marginal relevance}~\citep{carbonell1998use} approach, with Eq. \ref{eq:local-consistency} as the support metric. 
The set $\mathcal I^\ast$ is obtained by solving the following optimization problem: 

\begin{equation}\label{eq:opt_prob}
    \mathcal I^\ast(u|k,w) = \mathop{\text{argmax}}_{\mathcal I\subset I,|\mathcal I|=k}\,(1-\lambda) \sum_{i\in \mathcal I}\widetilde R_{ui}+\lambda\,Consistency(p_u,{q}_{\mathcal I}| w),
\end{equation}
\vspace{1mm}

\noindent where ${q}_{\mathcal I}$ is $q$ when the top-$k$ list includes items $\mathcal I$, and $\lambda\in[0,1]$ is a parameter that expresses the trade-off between accuracy and learning opportunity consistency. With $\lambda=0$, we yield the output of the recommender, not taking consistency optimization into account. Conversely, with $\lambda =1$, the output of the recommender is discarded, and we focus only on maximizing consistency. 

This greedy approach yields an ordered list of resources, and the resulting list at each step is $(1 - 1/e)$ optimal among the lists of equal size. 
The proof of the optimality of the proposed approach is provided in Appendix \ref{sec:proof-opt}. 
This property fits with the real world, where learners may initially see only the first $k$ recommendations, and the remaining items may become visible after scrolling. 
Our approach also allows controlling more than one learning opportunity principle in the ranked lists, with no constraints on the size of $C$. 

\subsection{Evaluation Scenario and Experimental Results}  \label{sec:supporting-scenario-results}
In this section, we assess the impact of controlling consistency and equality of learning opportunities across learners after applying our procedure to pursue the platform's principles (i.e., maximizing all the principle indicators). 
It is important to note that we considered the same setup described for the exploratory analysis, including the same datasets (Section \ref{sec:data}), protocols (Section \ref{sec:algo}), and metrics (Section \ref{sec:problem-formulation}), to answer four key research questions: 

\begin{itemize}[leftmargin=*]
\item RQ1. Which weight setup achieved the best accuracy-equality trade-off?
\item RQ2. Which principles have experienced the largest gain in consistency?
\item RQ3. Which is the influence of the original relevance score distribution?
\item RQ4. How do the recommended lists differ, before and after our approach? 
\end{itemize}

\vspace{2mm} \noindent \textbf{Influence of Weight Setup}. 
We run experiments to assess 
($i$) the influence of our procedure and the weight-based strategy on accuracy, consistency, and equality, and 
($ii$) the relation between a loss in accuracy and a gain in consistency and equality while applying our procedure. 
To this end, we envisioned three approaches of principle weight assignment, while applying our procedure:   
\begin{itemize}[noitemsep,topsep=0pt,leftmargin=*]
\item \texttt{Glob} assigns the same weight to all the principles, for all users. This method would not take into account the level of consistency the recommended list to a given user already achieved and will treat all the principles equally. 
\item \texttt{User} assigns, to a principle, a weight proportional to the consistency gap for that principle concerning the target of the platform, computed during the exploratory analysis. The consistency gap for a principle has been obtained by averaging the individual consistency gaps across users. 
\item \texttt{Pers}, given a user, assigns the weight for a principle by considering only their (individual) consistency gap for that principle. Thus, different weights are used along with the user population. 
\end{itemize}

\noindent For each model, we run an instance of our re-ranking procedure for each weight assignment strategy, assigning to $\lambda$ a value in $[0.0, 0.25, 0.50, 0.75, 0.99]$. 

\begin{figure}[!t]
\centering
\begin{subfigure}[t]{0.32\linewidth}
    \centering
    \includegraphics[width=1.0\linewidth]{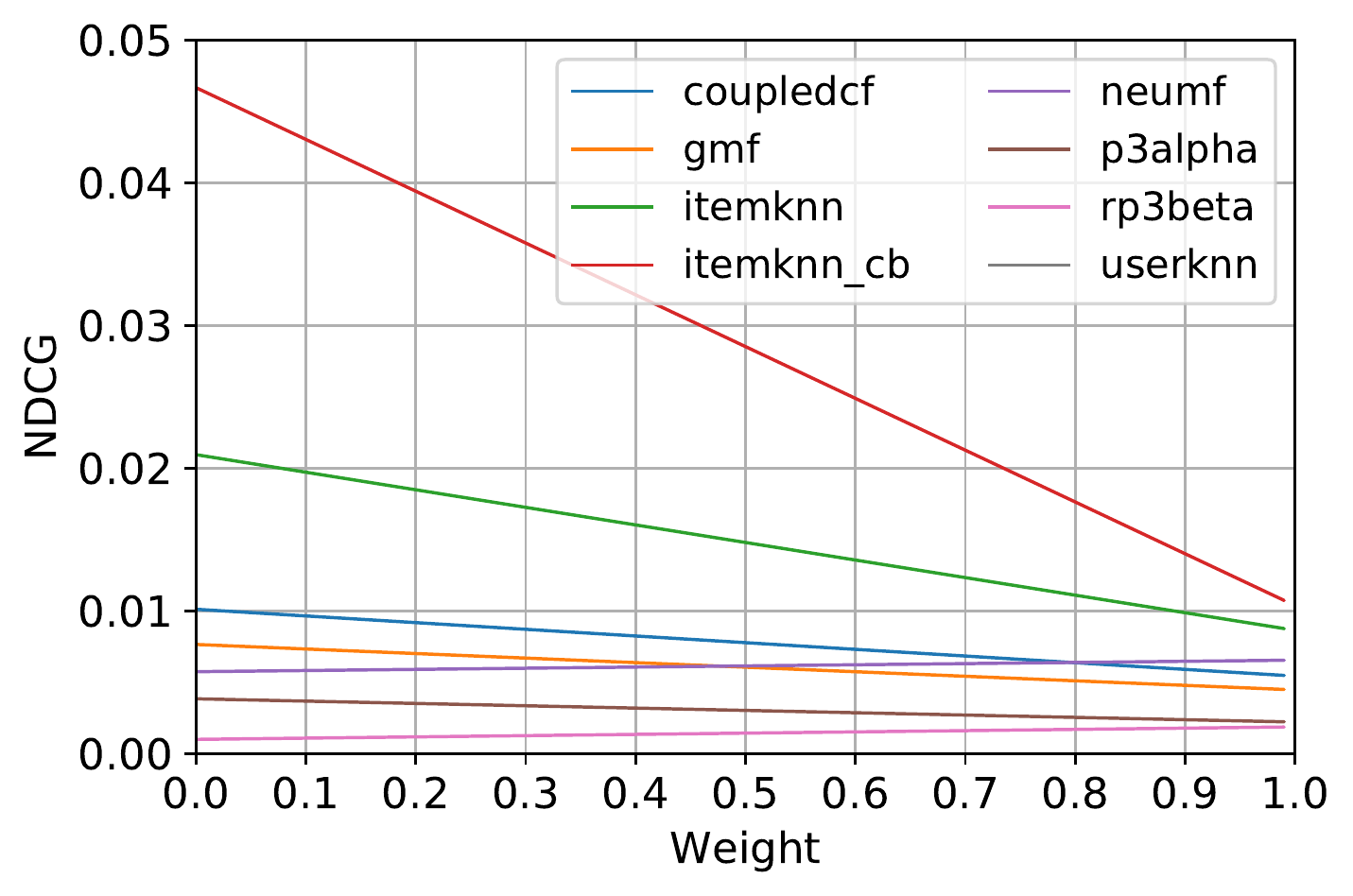}
    \caption{NDCG on Glob.}
\end{subfigure}
\begin{subfigure}[t]{0.32\linewidth}
    \centering
    \includegraphics[width=1.0\linewidth]{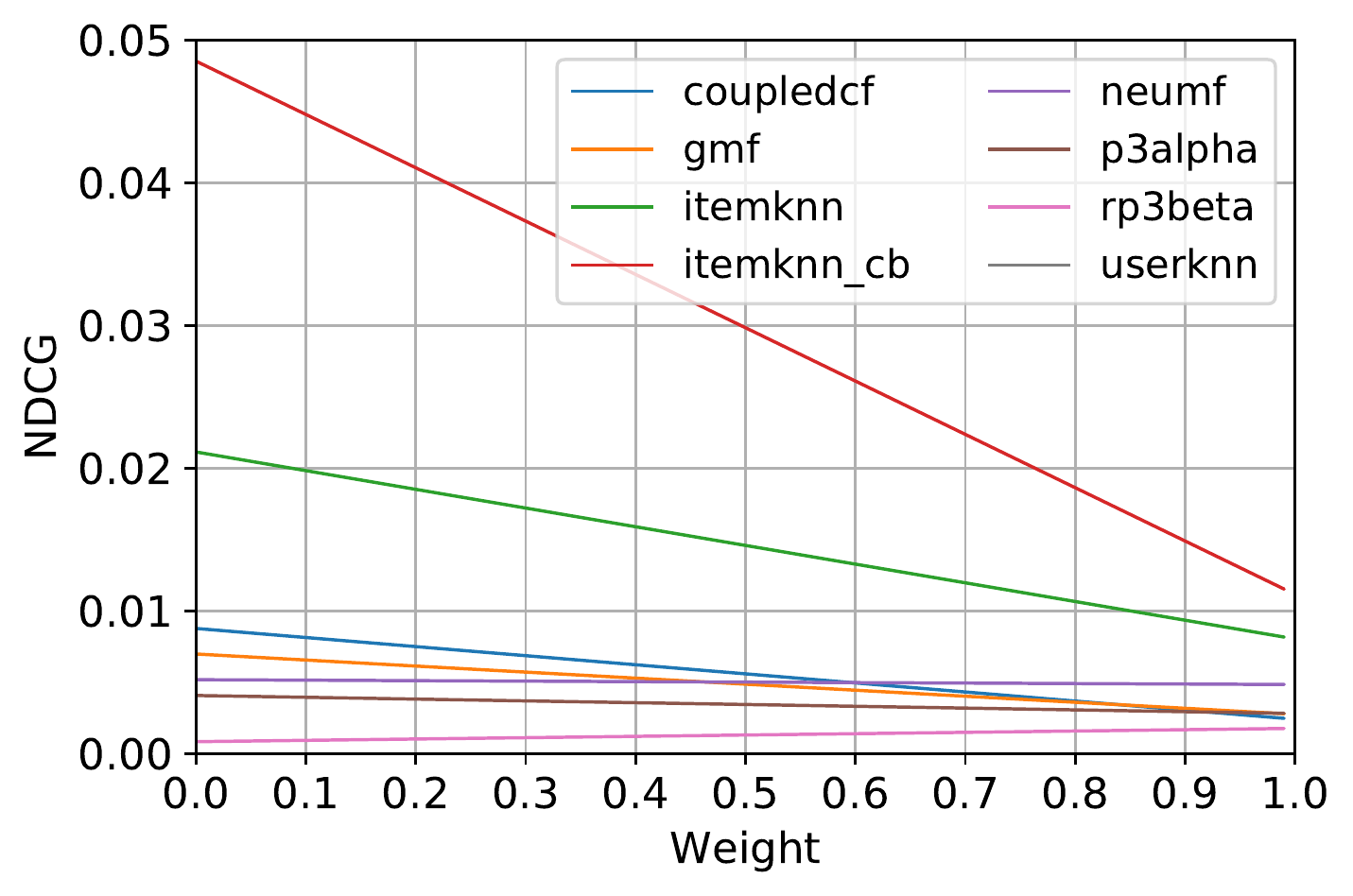}
    \caption{NDCG on User.}
\end{subfigure}
\begin{subfigure}[t]{0.32\linewidth}
    \centering
    \includegraphics[width=1.0\linewidth]{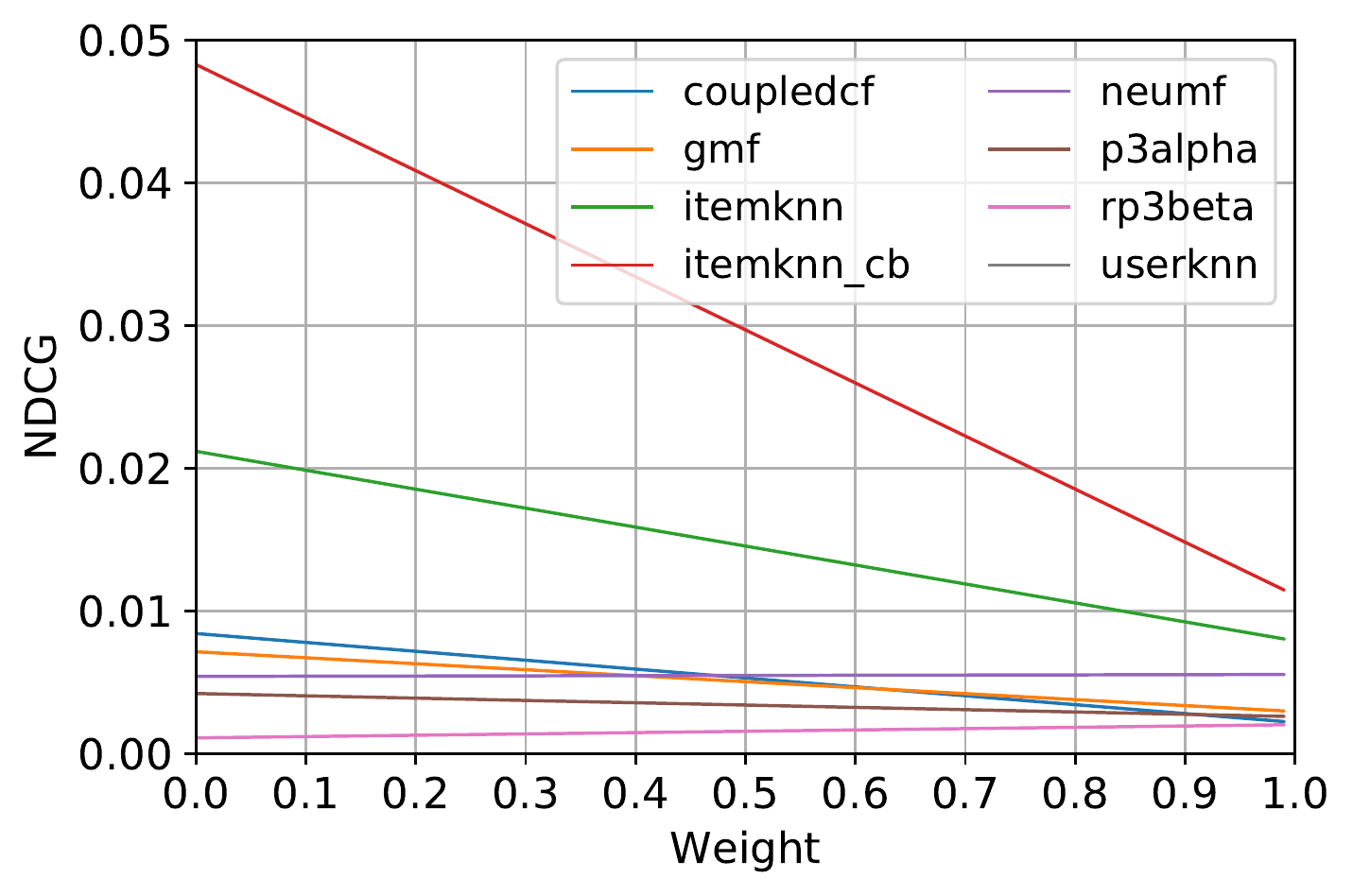}
    \caption{NDCG on Pers.}
\end{subfigure}
\begin{subfigure}[t]{0.32\linewidth}
    \centering
    \includegraphics[width=1.0\linewidth]{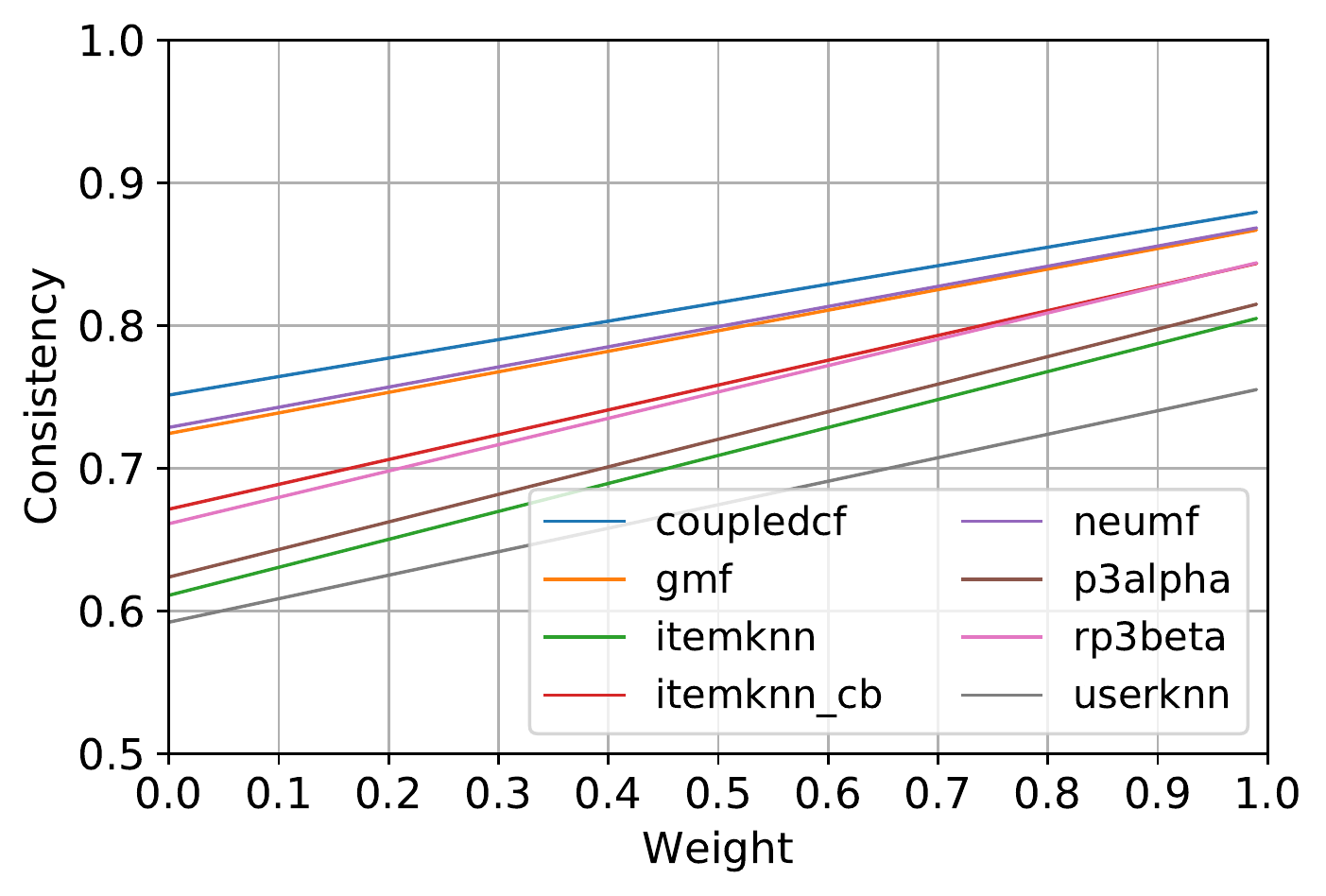}
    \caption{Consistency on Glob.}
\end{subfigure}
\begin{subfigure}[t]{0.32\linewidth}
    \centering
    \includegraphics[width=1.0\linewidth]{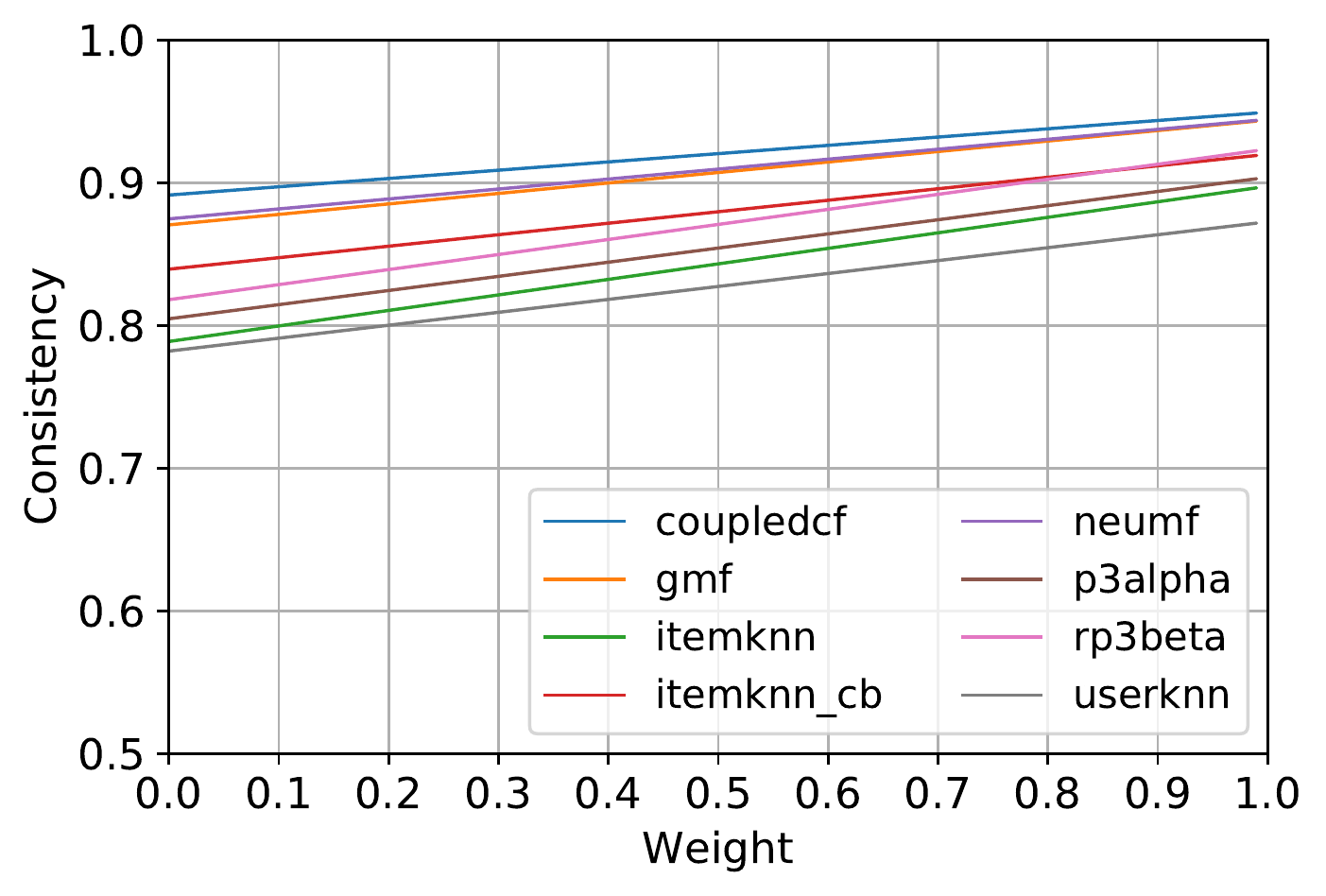}
    \caption{Consistency on User.}
\end{subfigure}
\begin{subfigure}[t]{0.32\linewidth}
    \centering
    \includegraphics[width=1.0\linewidth]{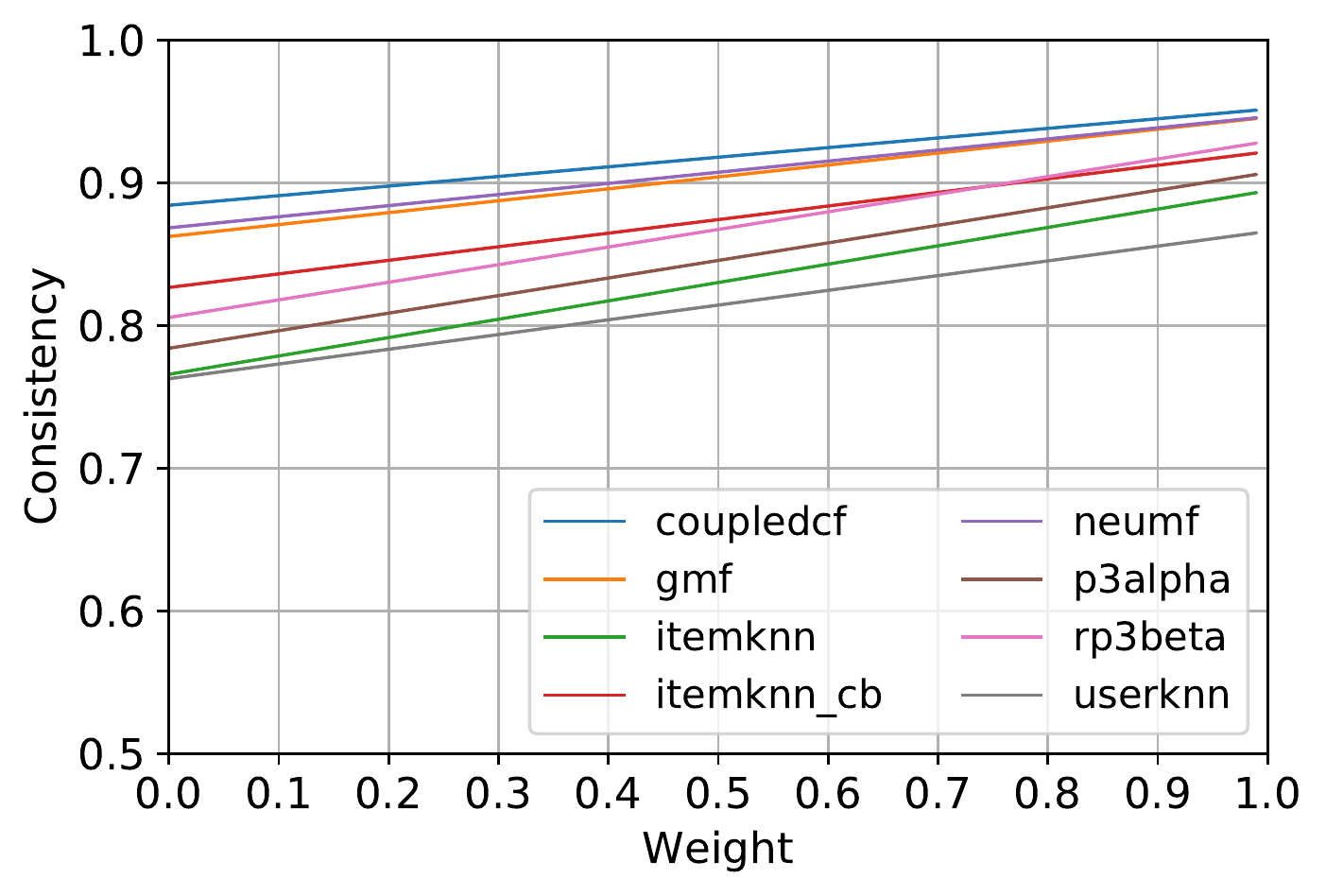}
    \caption{Consistency on Pers.}
\end{subfigure}
\begin{subfigure}[t]{0.32\linewidth}
    \centering
    \includegraphics[width=1.0\linewidth]{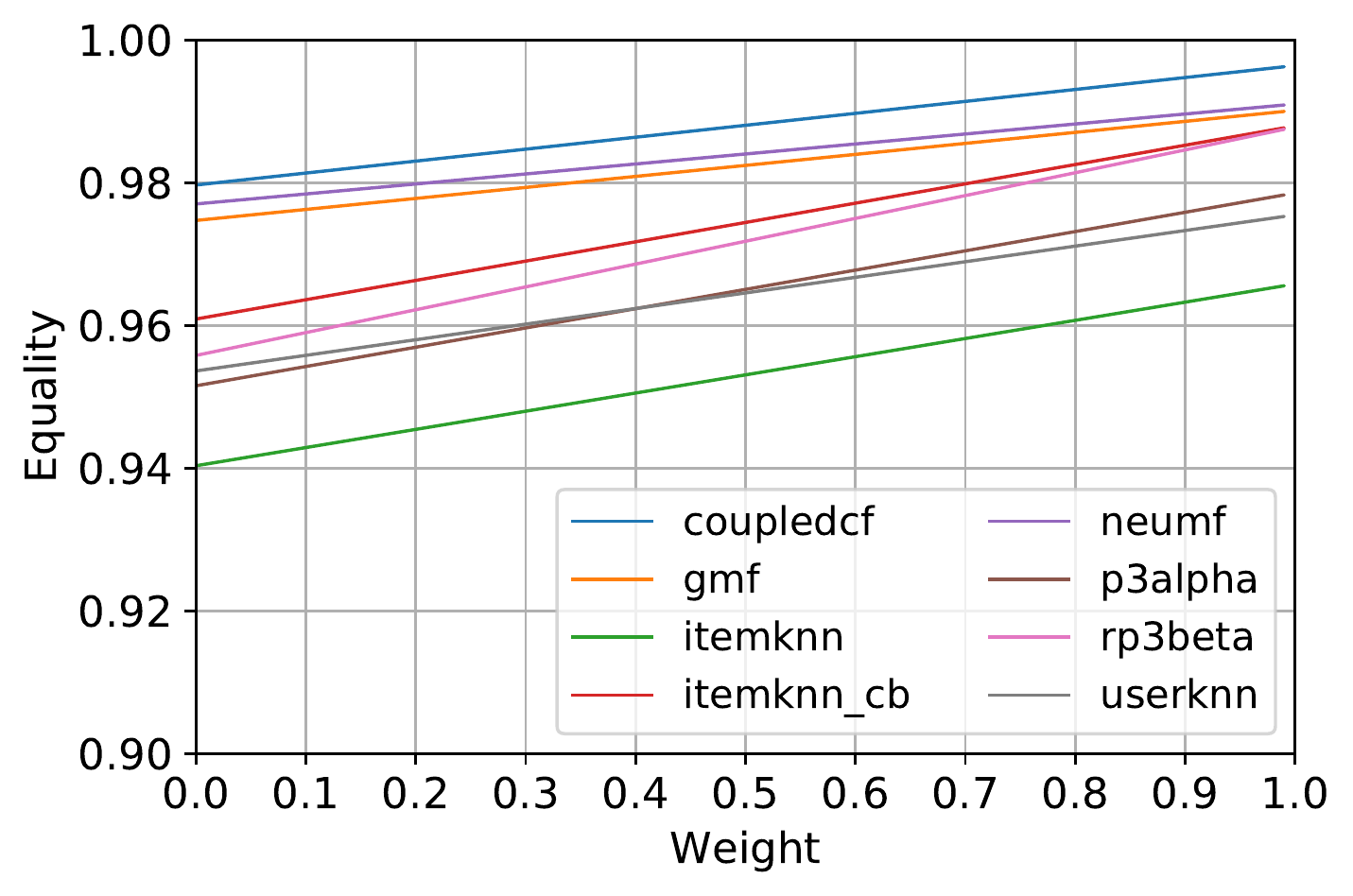}
    \caption{Equality on Glob.}
\end{subfigure}
\begin{subfigure}[t]{0.32\linewidth}
    \centering
    \includegraphics[width=1.0\linewidth]{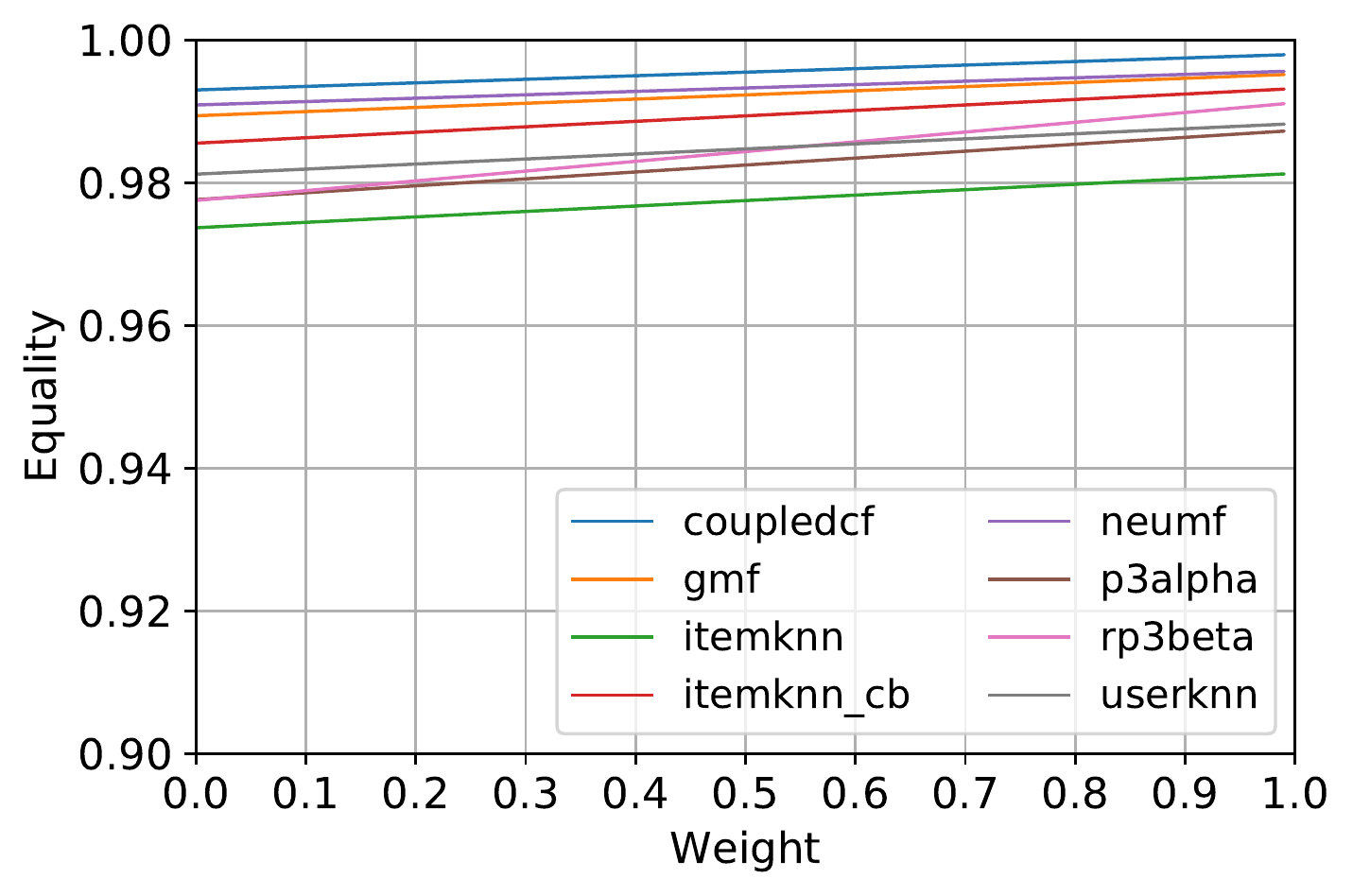}
    \caption{Equality on User.}
\end{subfigure}
\begin{subfigure}[t]{0.32\linewidth}
    \centering
    \includegraphics[width=1.0\linewidth]{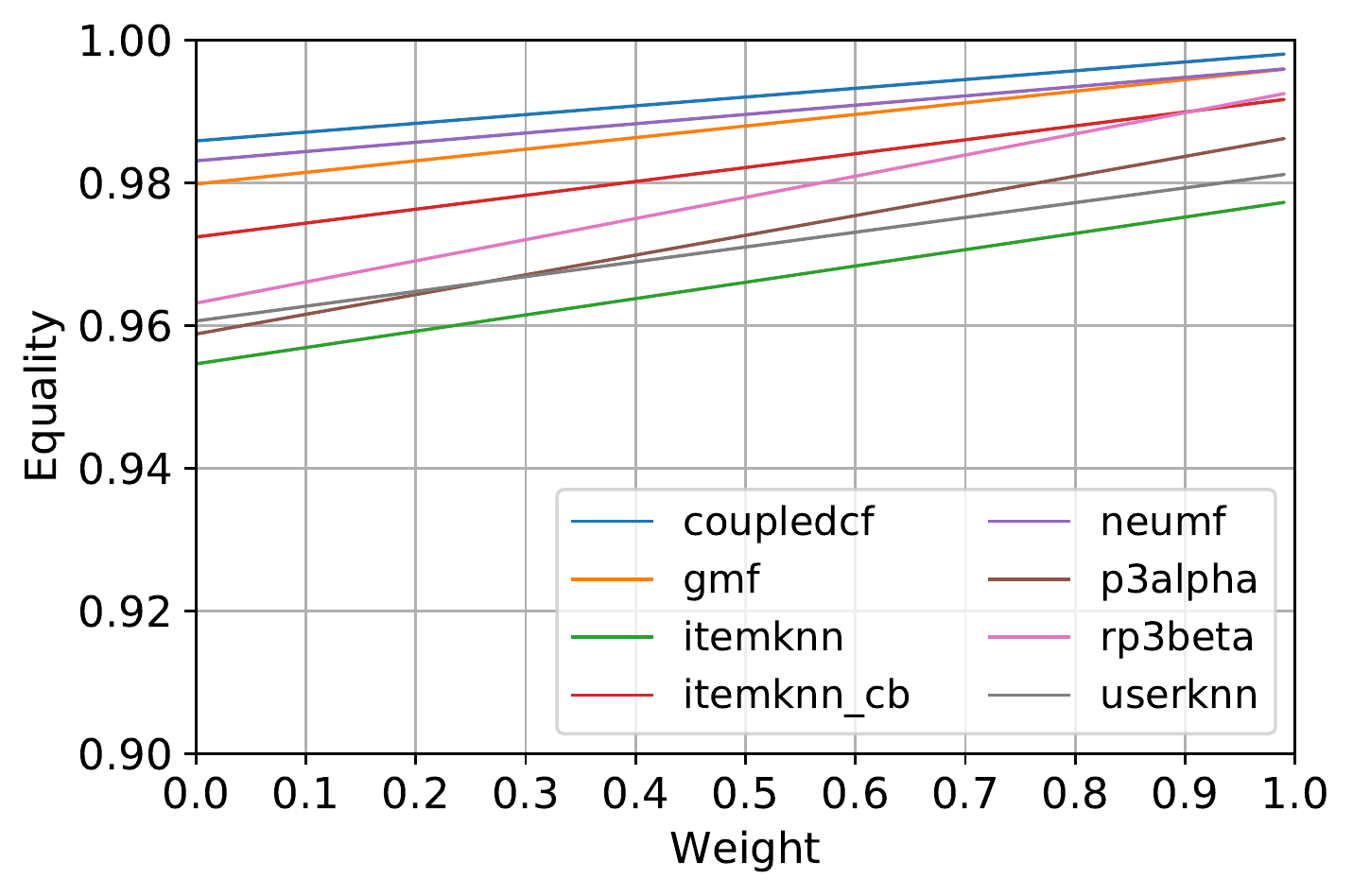}
    \caption{Equality on Pers.}
\end{subfigure}
\caption{\textbf{Controlled Performance}. Normalized Discounted Cumulative Gain (NDCG), consistency, and equality achieved by our procedure under Glob, User, and Pers weight assignment strategies. For each algorithm, our approach has been applied at various $\lambda$.} \label{fig:eq-con-ndcg-absolute}
\end{figure}

The results related to NDCG, consistency, and equality are shown in \figurename~\ref{fig:eq-con-ndcg-absolute}. 
Specifically, top-row plots on NDCG highlighted the fact that ItemKNN and ItemKNN-CB experienced the largest loss in NDCG at increasing $\lambda$. 
The rest of the algorithms showed a more stable pattern on NDCG, even though the NDCG absolute value is significantly lower for the one achieved by ItemKNN and ItemKNN-CB. 
Throughout the weight assignment strategy, no significant difference was observed for the same algorithm over the three strategies. 
On the other hand, the weight assignment strategy has an important role in consistency and equality (center and bottom rows). 
Specifically, \texttt{User} and \texttt{Pers} weight setups made it possible to achieve higher consistency and quality than \texttt{Glob}. 
It can also be observed that all the algorithms brought the same degree of improvement in consistency while varying $\lambda$. 

Interestingly, by looking at equality scores, two patterns of improvement were observed. 
Specifically, the algorithms from the graph-based, content-based, and hybrid families showed a larger improvement at each value of $\lambda$ than the other families. 
The following observation can be drawn:

\vspace{2mm} \setlength{\fboxsep}{10pt}
\noindent \textbf{Observation 6}. \textit{The considered weight assignment strategies do not differ in terms of accuracy loss. However, User and Pers lead to consistency and equality values higher than Glob, at the same $\lambda$. This means that a higher equality of recommended learning opportunities can be achieved by considering the consistency gaps experienced by the individual learner for each principle.}\\

\begin{figure}[!b]
\centering
\begin{subfigure}[t]{0.32\linewidth}
    \centering
    \includegraphics[width=1.0\linewidth]{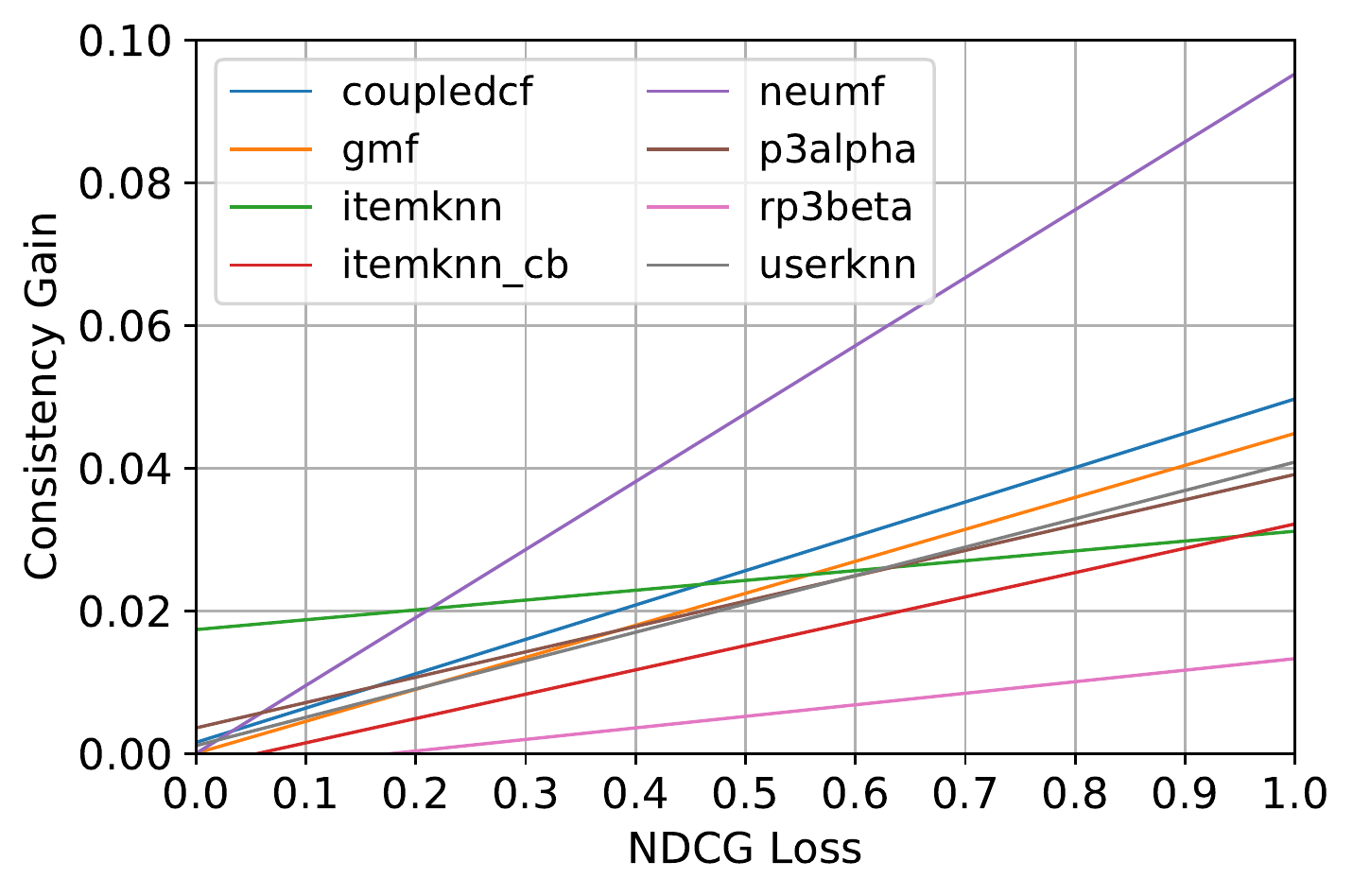}
    \caption{NDCG-Cons on Glob.}
\end{subfigure}
\begin{subfigure}[t]{0.32\linewidth}
    \centering
    \includegraphics[width=1.0\linewidth]{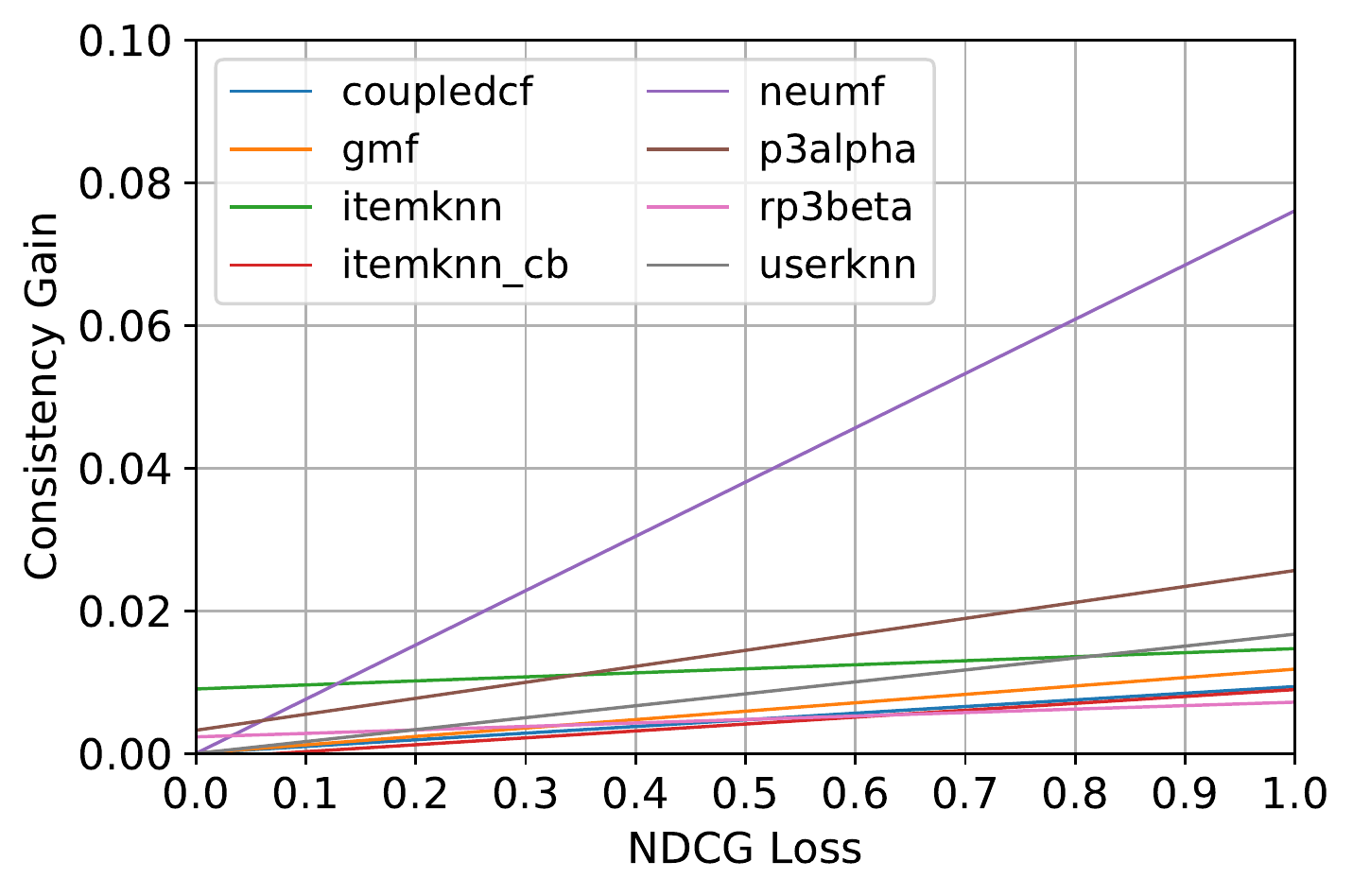}
    \caption{NDCG-Cons on User.}
\end{subfigure}
\begin{subfigure}[t]{0.32\linewidth}
    \centering
    \includegraphics[width=1.0\linewidth]{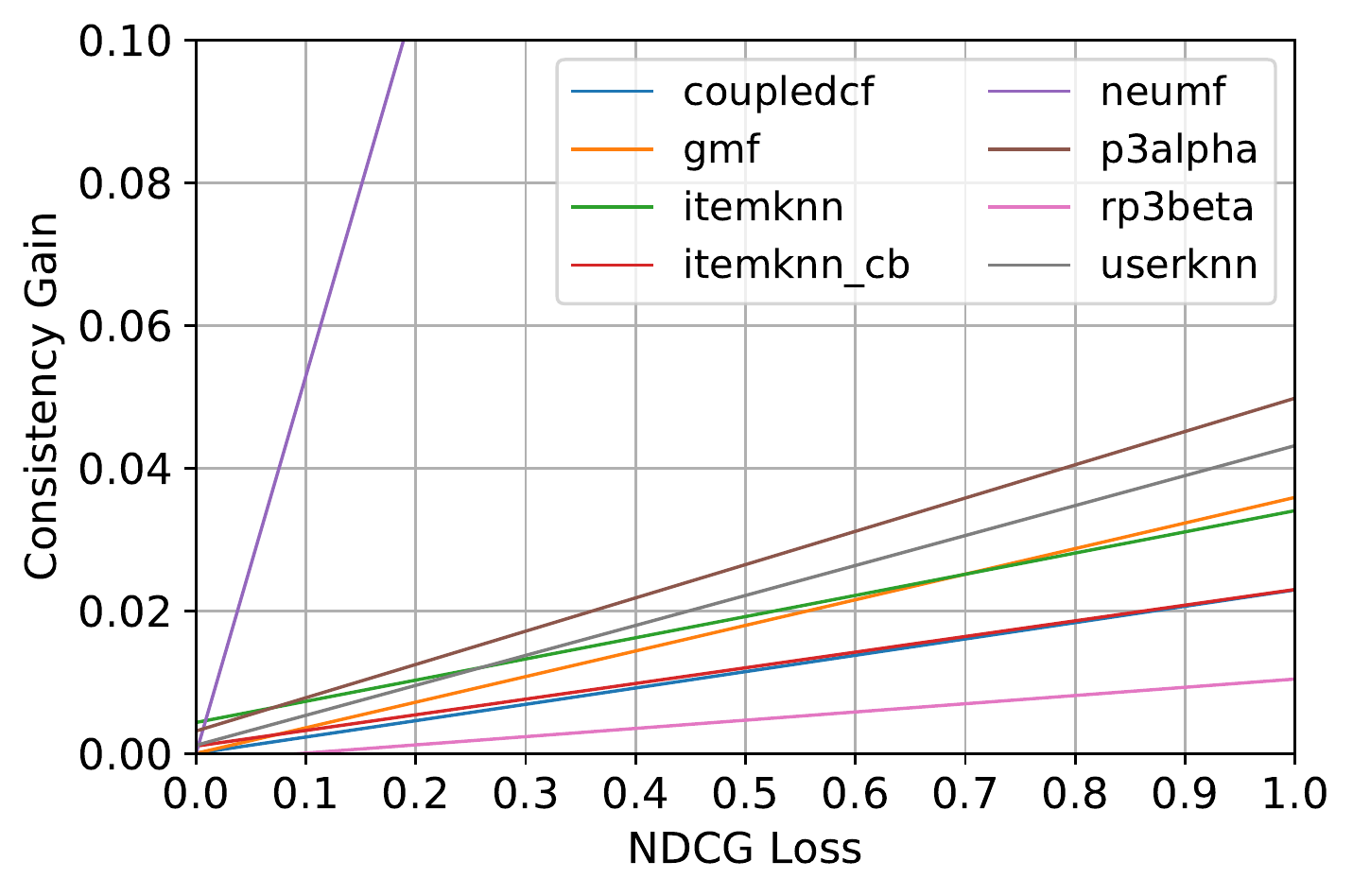}
    \caption{NDCG-Cons on Pers.}
\end{subfigure}
\begin{subfigure}[t]{0.32\linewidth}
    \centering
    \includegraphics[width=1.0\linewidth]{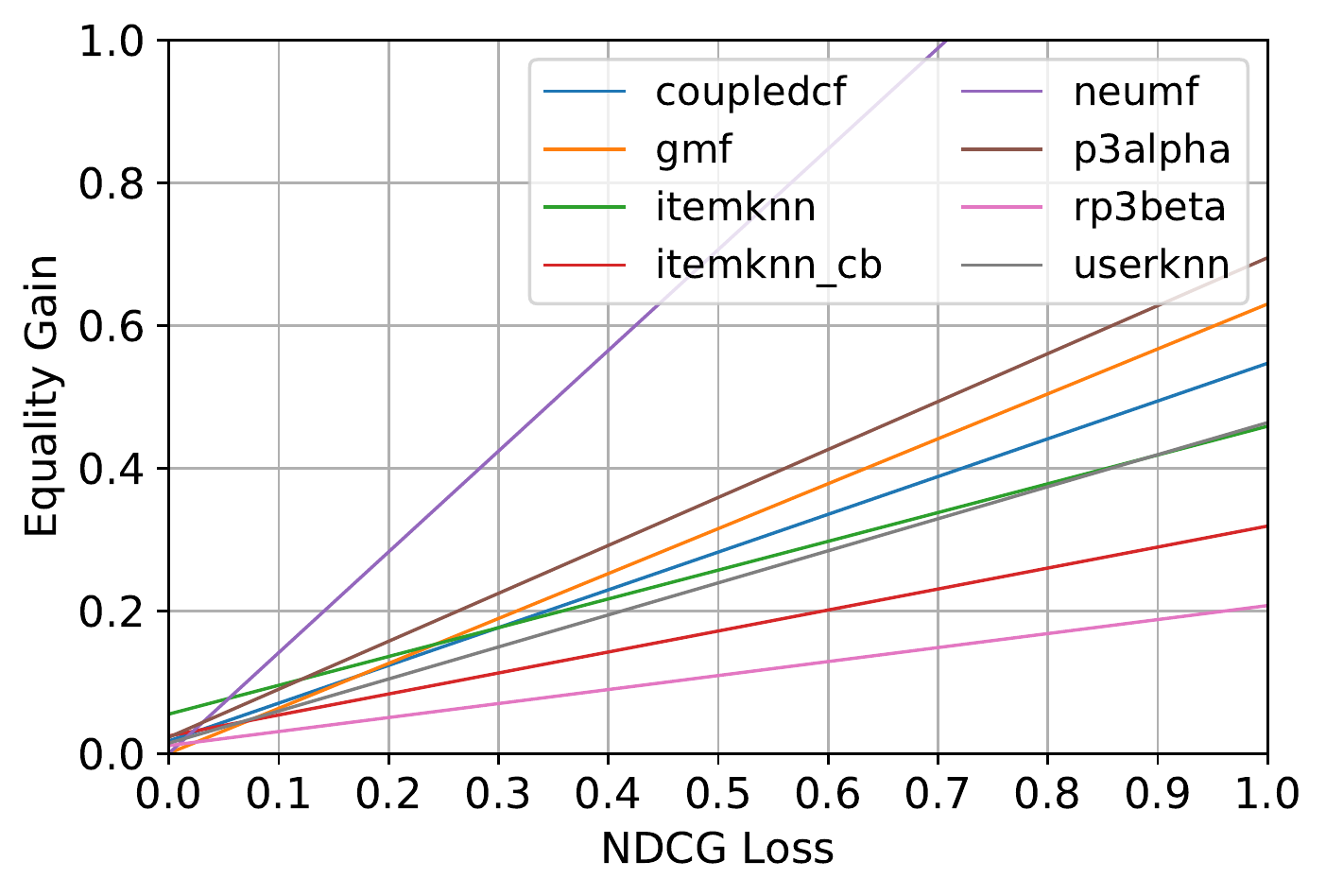}
    \caption{NDCG-Equality on Glob.}
\end{subfigure}
\begin{subfigure}[t]{0.32\linewidth}
    \centering
    \includegraphics[width=1.0\linewidth]{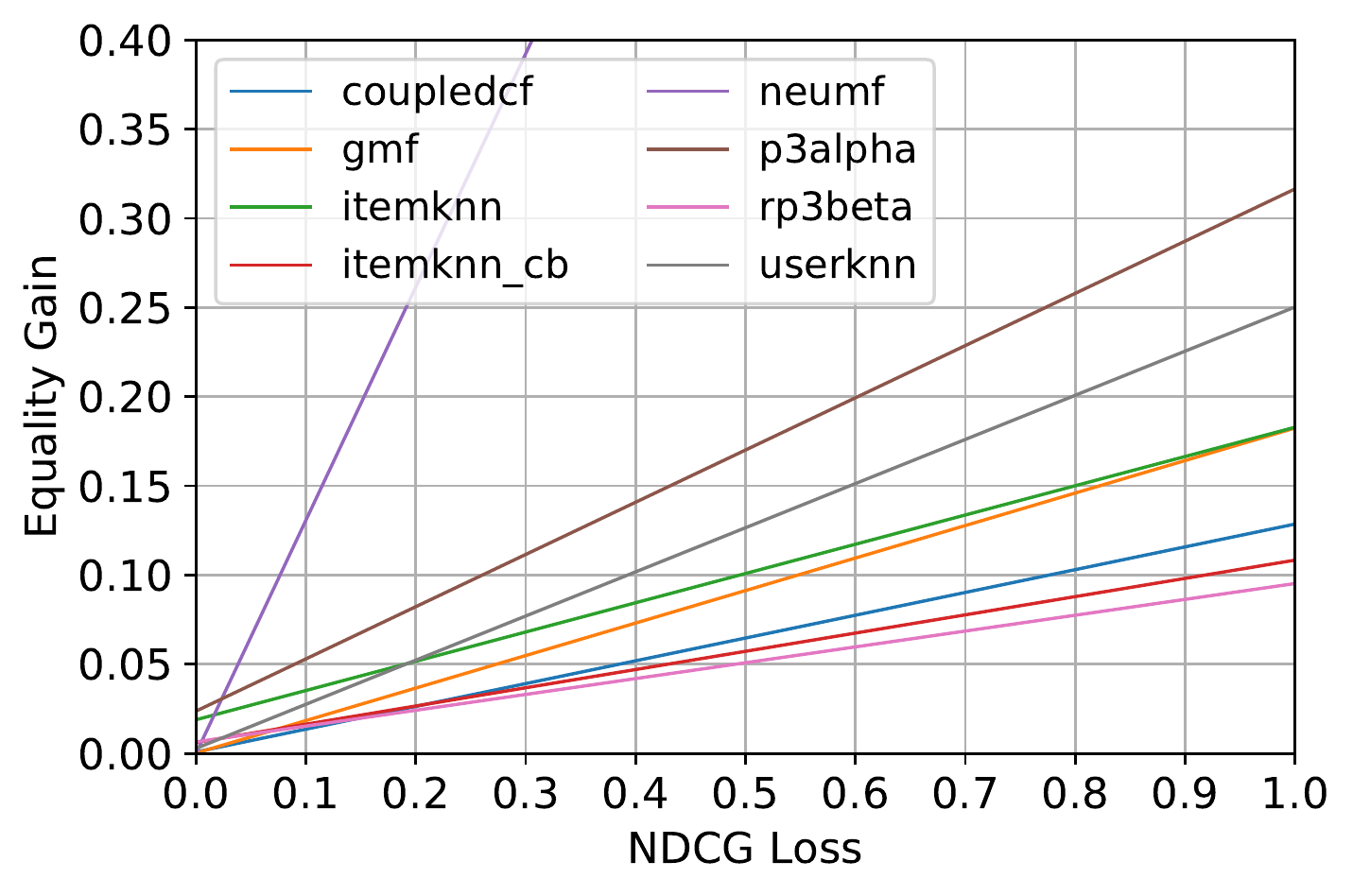}
    \caption{NDCG-Equality on User.}
\end{subfigure}
\begin{subfigure}[t]{0.32\linewidth}
    \centering
    \includegraphics[width=1.0\linewidth]{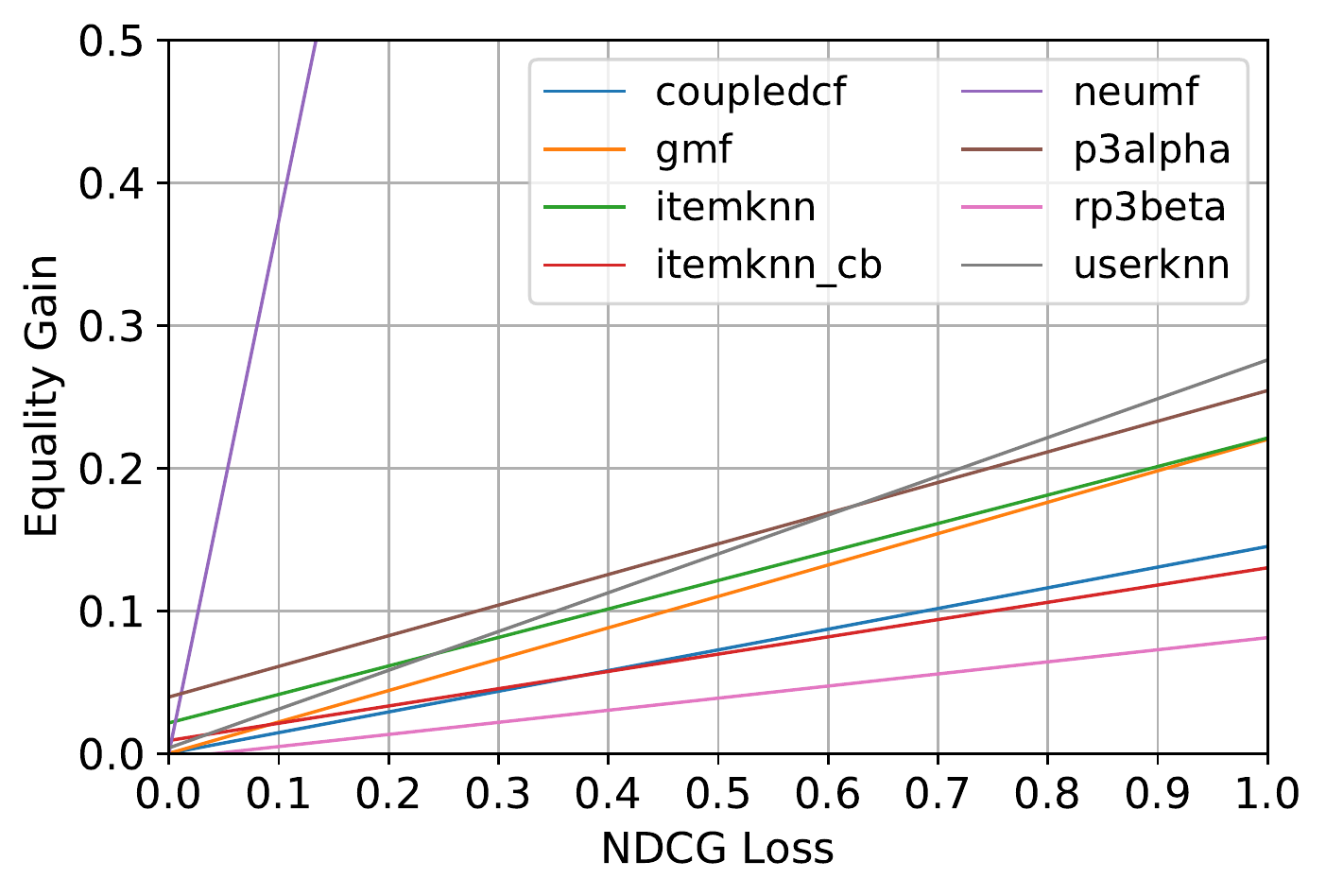}
    \caption{NDCG-Equality on Pers.}
\end{subfigure}
\caption{\textbf{Accuracy-Equality Relation}. For each algorithm and weight assignment setup, we computed the gain in equality and consistency that can be achieved at the cost of loosing a certain degree of accuracy.} \label{fig:eq-ndcg-gain-loss}
\end{figure}

To have a more detailed picture, we analyzed the connection between a loss in NDCG and between a gain in consistency and equality. 
Such a plot plays a key role in a real-world context. 
While it is the responsibility of scientists to bring forth the discussion about metrics, and possibly to design algorithms to optimize them by turning parameters, it is ultimately up to the stakeholders\footnote{Although this approval lets stakeholders control the different factors impacting on fairness, it still leaves open the questions around the intrinsic maturity of the weight settings and the accountability in the decision-making process.} (e.g., teachers, instructional designers, platform owners), depending on the targeted educational domain, to select the trade-offs most suitable for their context. 
Therefore, this kind of plot would support a decision regarding the value of $\lambda$ to set up in production, to achieve the desired trade-off. Figure~\ref{fig:eq-ndcg-gain-loss} plots the gain of consistency (top row) and quality (bottom row) resulting from the degree of loss the platform owners are willing to have. 
It can be observed that the patterns of consistency and equality within the same weight strategy show the same behavior concerning the loss in NDCG. 
This observation confirms the results of our exploratory analysis, where consistency and equality improvements were directly proportional.

\vspace{2mm} \noindent \textbf{Influence on Each Principle}. 
In this subsection, we run experiments to assess 
($i$) which principles show the largest improvement thanks to the proposed approach, and 
($ii$) which is the impact of the weighting strategy on the consistency of each principle. 
To answer these questions, for each model, we run an instance of our re-ranking procedure for each weighting strategy, varying $\lambda\in [0.0, 0.25, 0.50, 0.75, 0.99]$. 
Then, we computed the consistency of each principle achieved by an algorithm, at a given $\lambda$, with a given weighting.  

\begin{figure}[!b]
\centering
\begin{subfigure}[t]{0.32\linewidth}
    \centering
    \includegraphics[width=1.0\linewidth]{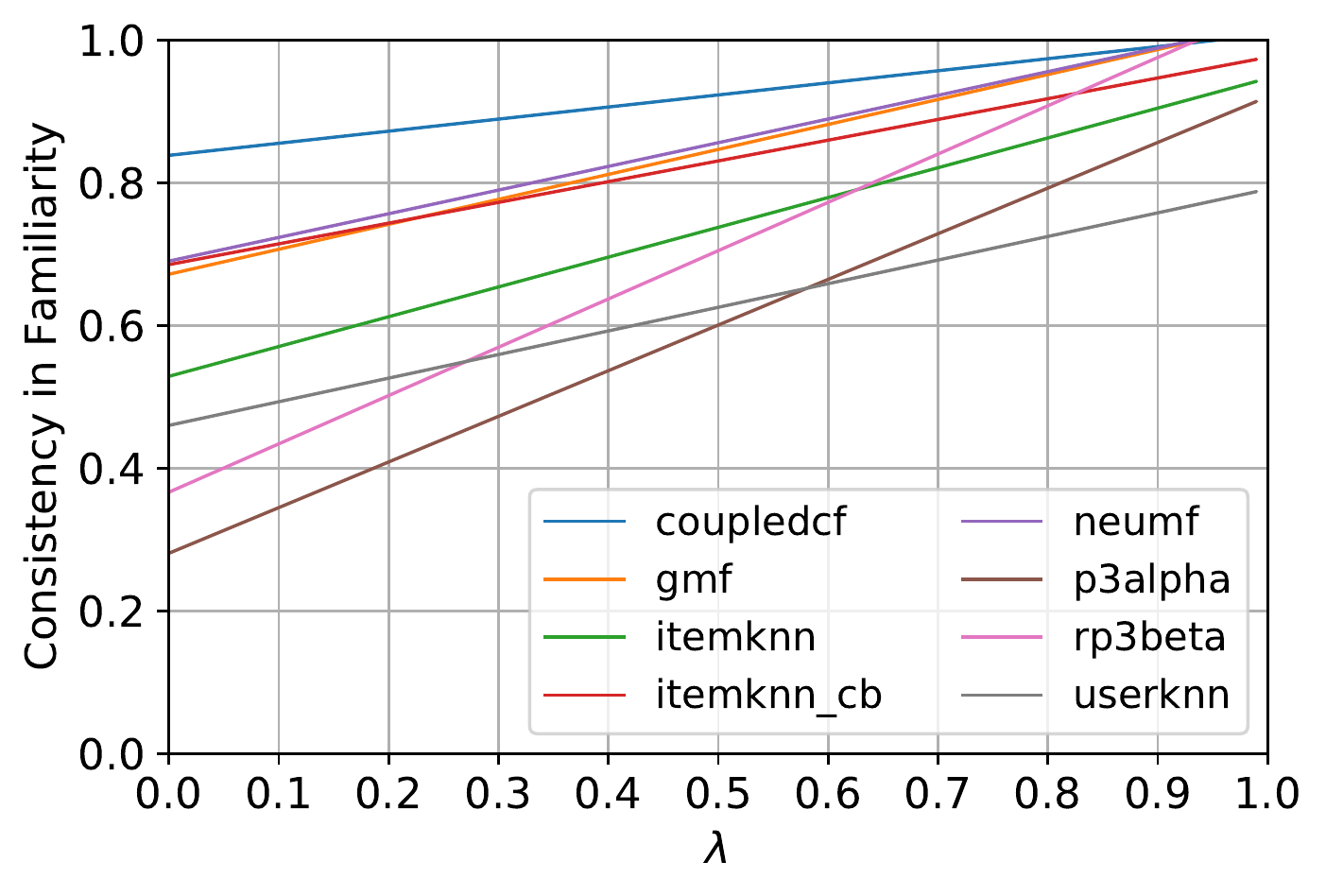}
    \caption{Familiarity.}
\end{subfigure}
\begin{subfigure}[t]{0.32\linewidth}
    \centering
    \includegraphics[width=1.0\linewidth]{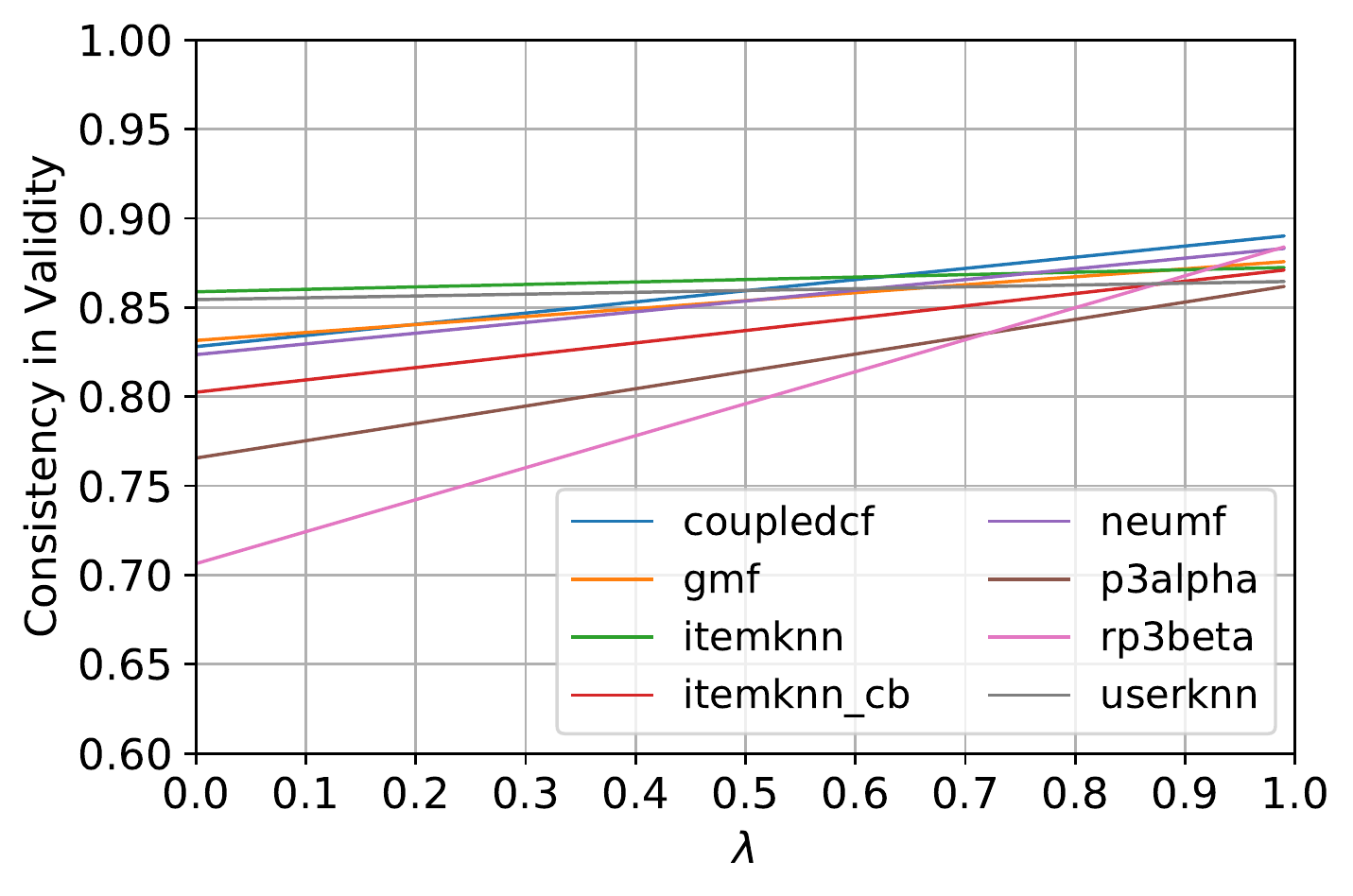}
    \caption{Validity.}
\end{subfigure}
\begin{subfigure}[t]{0.32\linewidth}
    \centering
    \includegraphics[width=1.0\linewidth]{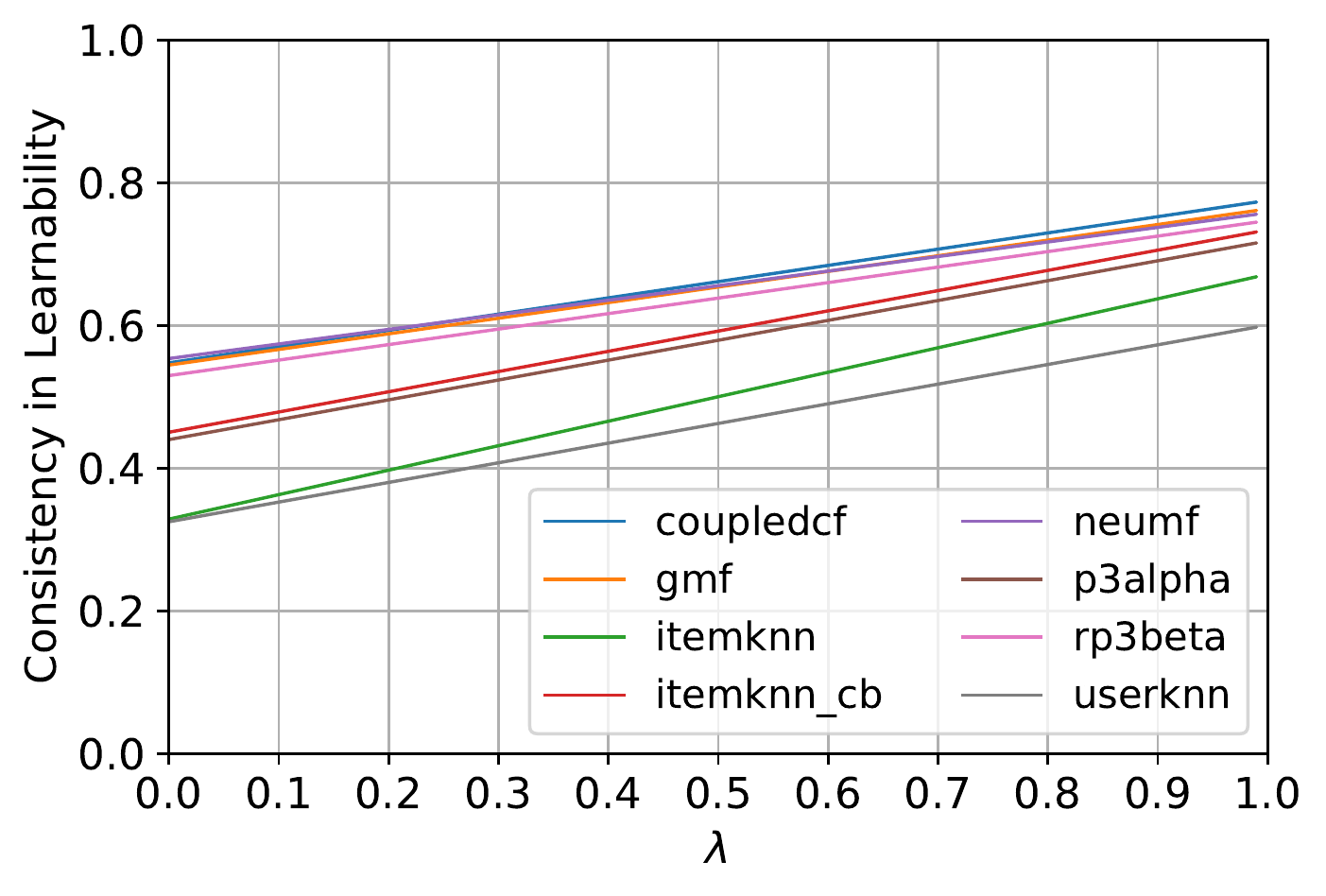}
    \caption{Learnability.}
\end{subfigure}
\begin{subfigure}[t]{0.32\linewidth}
    \centering
    \includegraphics[width=1.0\linewidth]{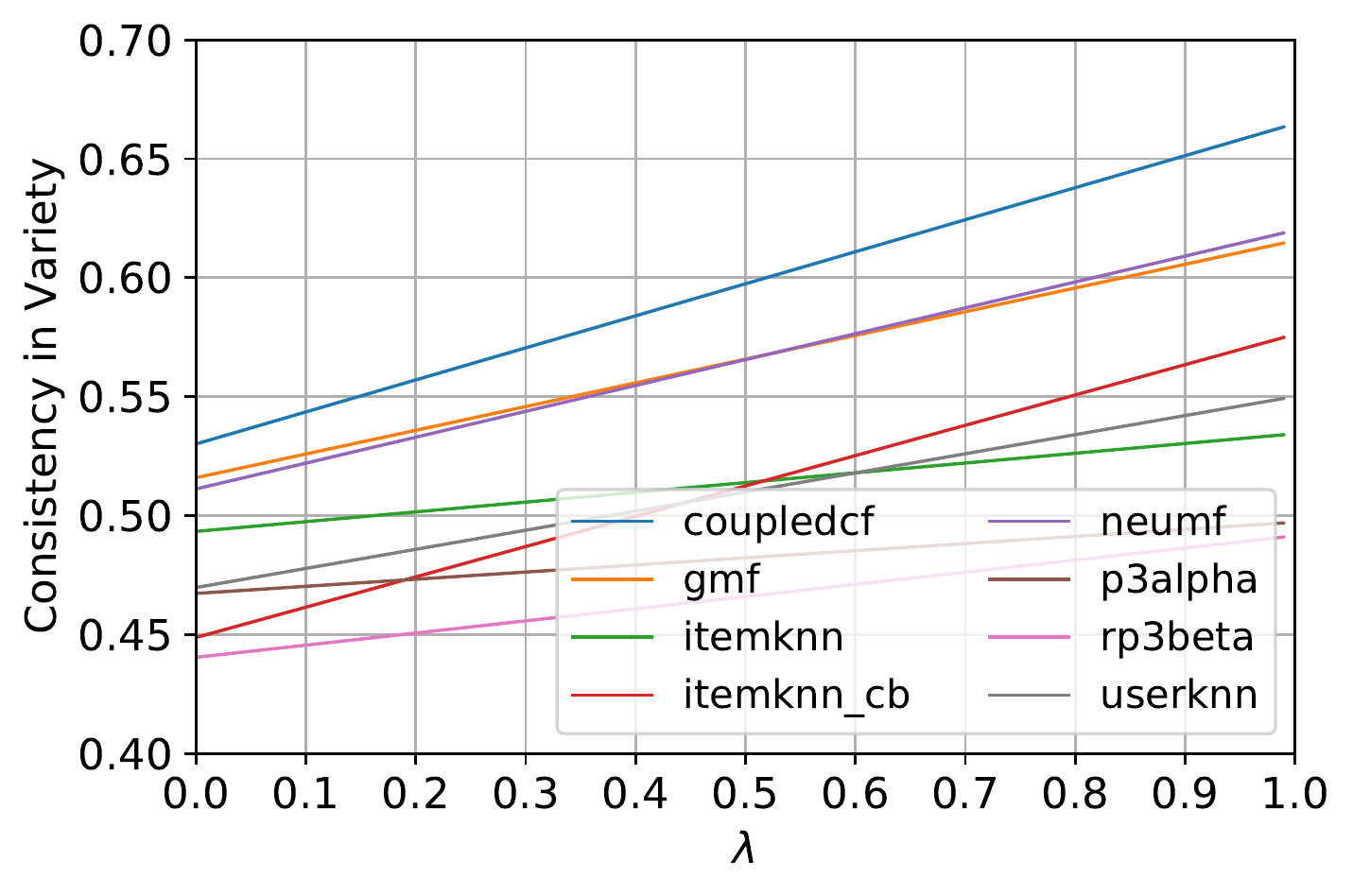}
    \caption{Variety.}
\end{subfigure}
\begin{subfigure}[t]{0.32\linewidth}
    \centering
    \includegraphics[width=1.0\linewidth]{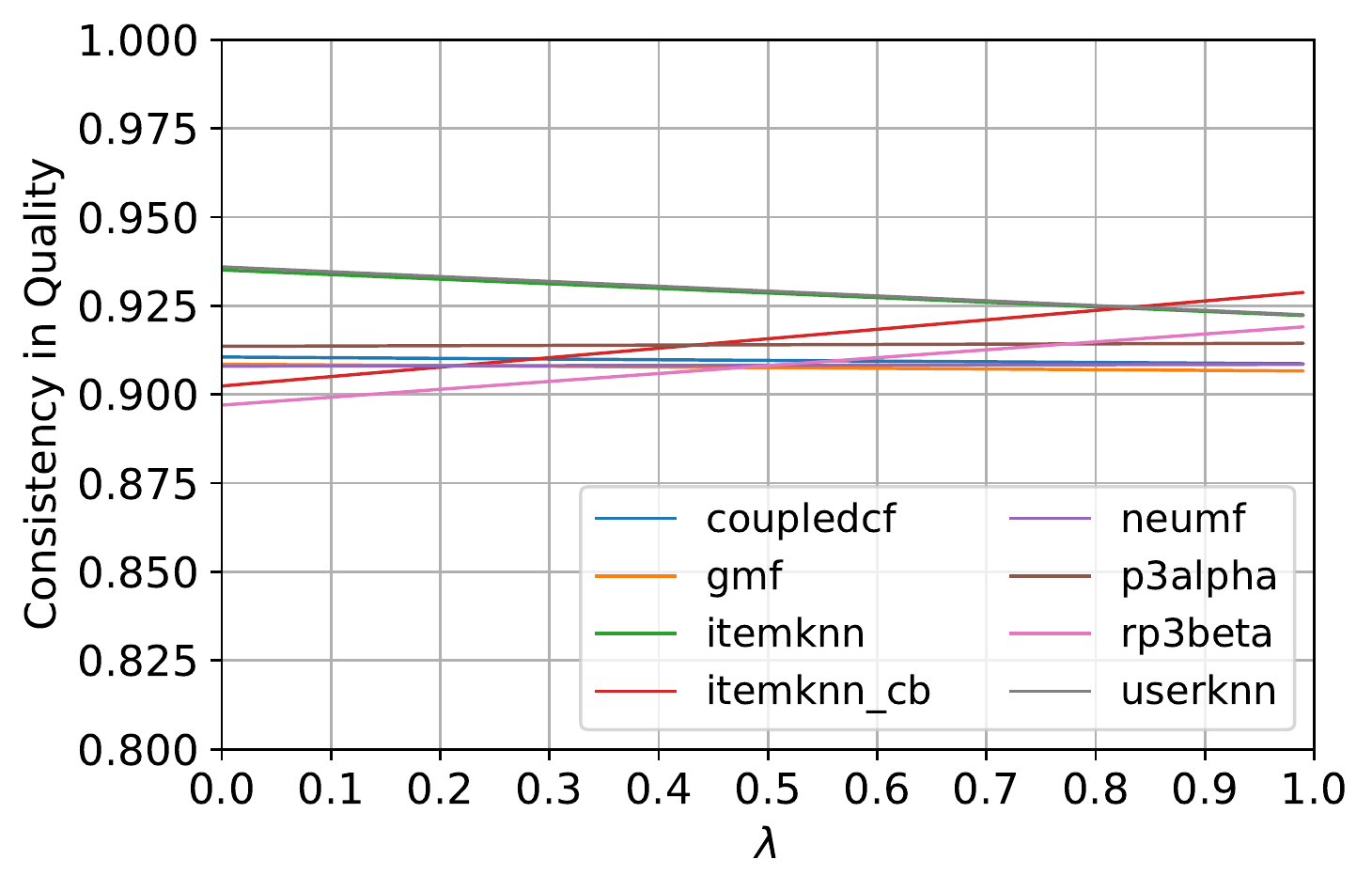}
    \caption{Quality.}
    \label{fig:control-cons-quality}
\end{subfigure}
\\
\begin{subfigure}[t]{0.32\linewidth}
    \centering
    \includegraphics[width=1.0\linewidth]{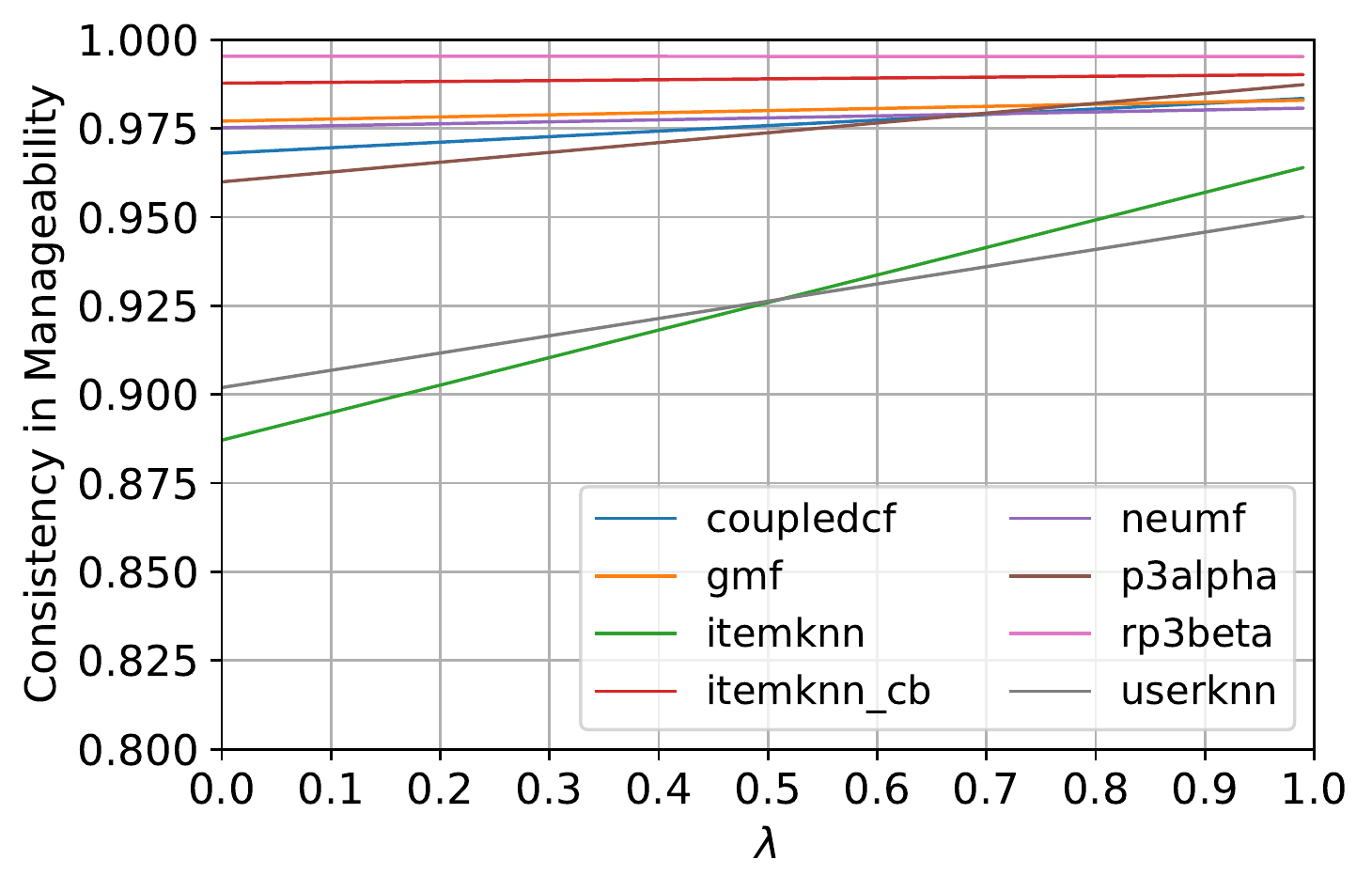}
    \caption{Massiveness.}
\end{subfigure}
\begin{subfigure}[t]{0.32\linewidth}
    \centering
    \includegraphics[width=1.0\linewidth]{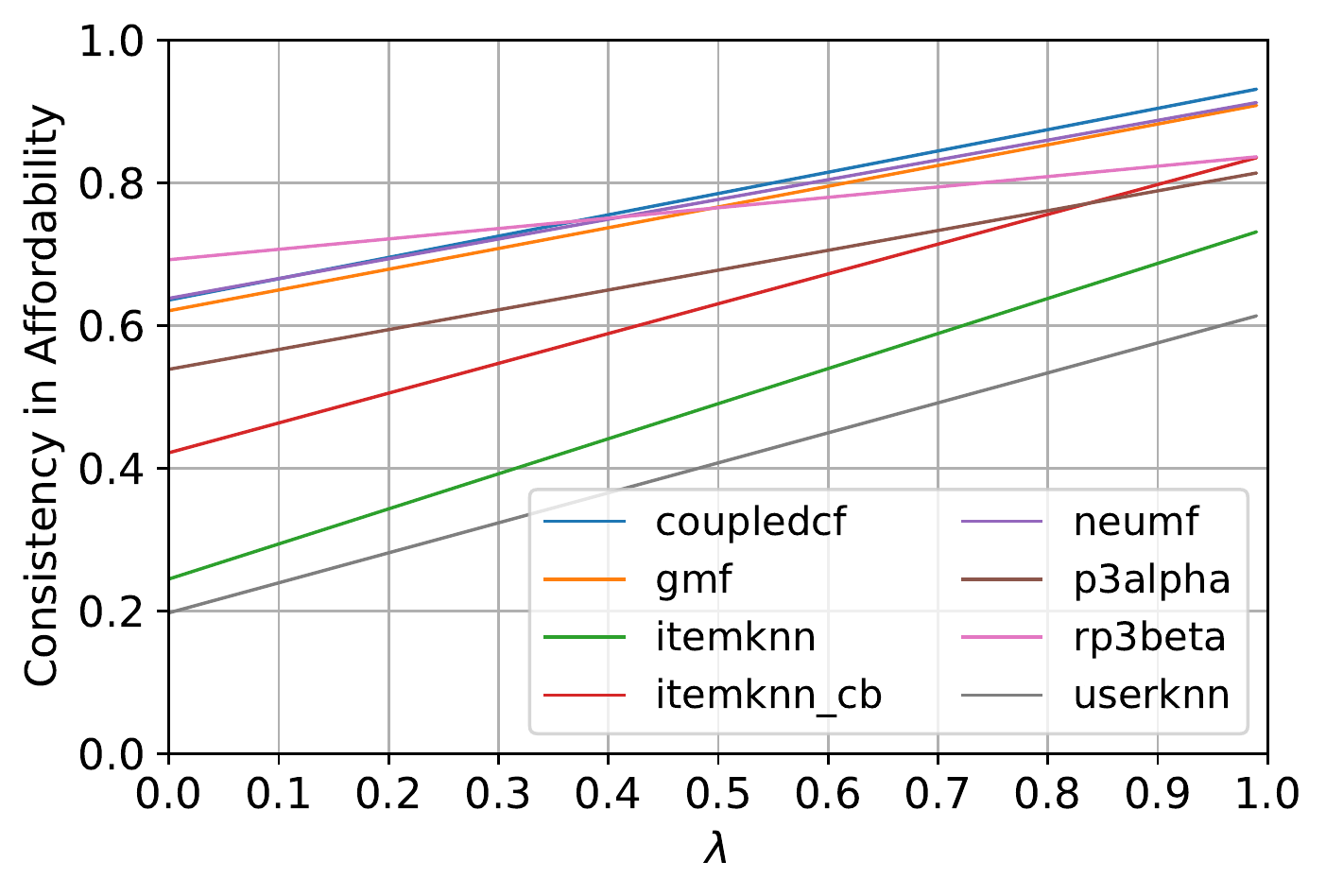}
    \caption{Affordability.}
\end{subfigure}
\caption{\textbf{Controlled Consistency}. Consistency per principle achieved by our procedure under under the Glob weight assignment strategies at various $\lambda$.} \label{fig:control-cons}
\end{figure}

\figurename~\ref{fig:control-cons} reports the impact of our procedure on the considered principle, for different algorithms. 
Overall, it can be observed that our procedure allows us to improve the consistency for all the principles, except Quality (see \figurename~\ref{fig:control-cons-quality}). 
This principle exhibited two main patterns based on the algorithms: quality increased for ItemKNN and RP3Beta, while it decreased for the other algorithms. 
Interestingly, adding course metadata information into the algorithm (ItemKNN-CB concerning ItemKNN) changes the trend in quality. 
Furthermore, our approach made it possible to improve Familiarity, Variety, and Affordability, all which achieved low consistency scores in the exploratory analysis. 
It follows that the value of $\lambda$ can be fine-tuned to reach the desired level for a given principle. 
Another observation can be drawn.      

\vspace{2mm} \setlength{\fboxsep}{10pt}
\noindent \textbf{Observation 7}. \textit{Controlling learning opportunity results in higher familiarity, variety, and affordability, while maintaining stable values for the other principles. However, quality may decrease, especially with collaborative filtering.}\\

\begin{figure}[!b]
\centering
\begin{subfigure}[t]{1.\linewidth}
    \centering
    \includegraphics[width=.7\linewidth]{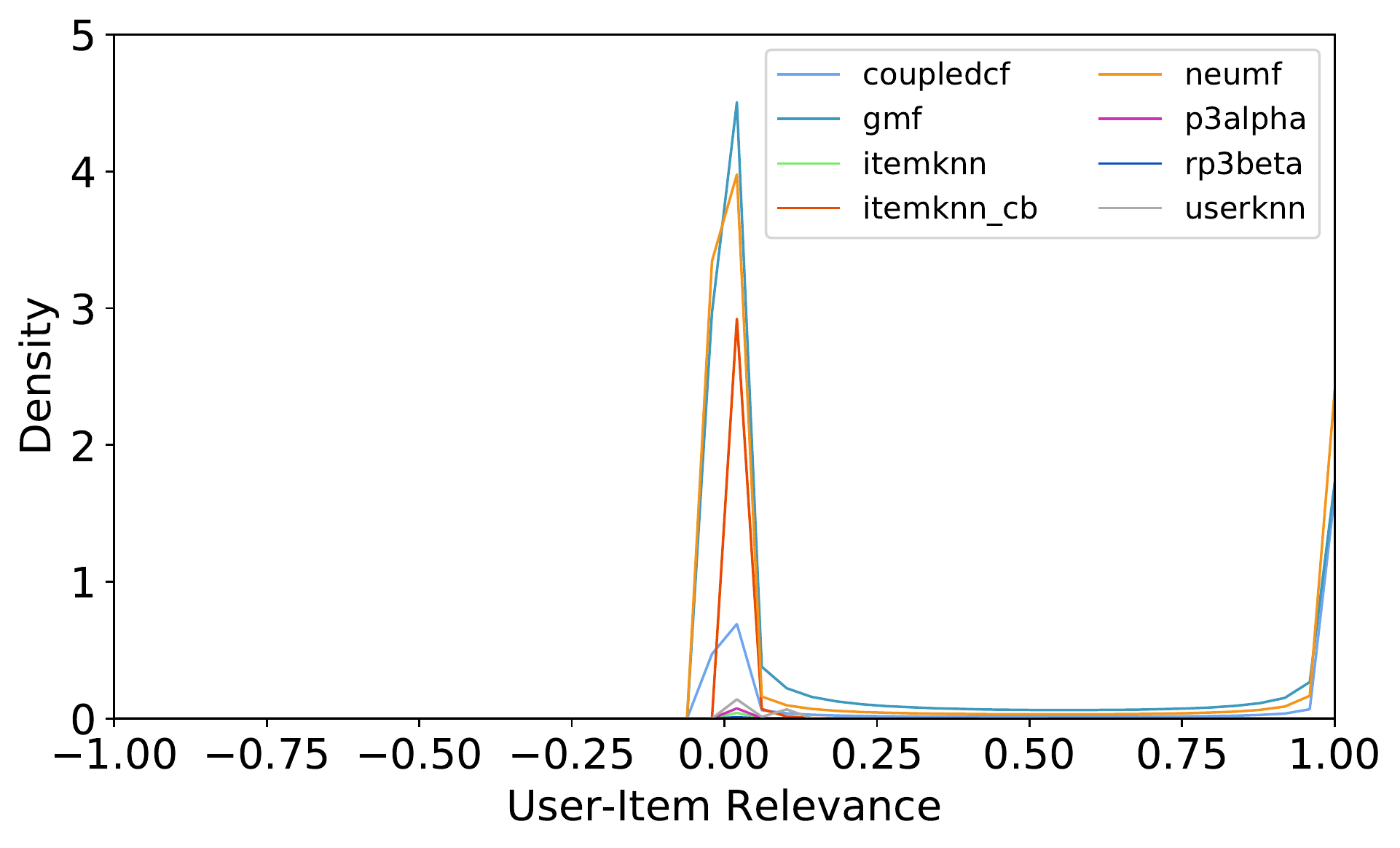}
\end{subfigure}
\vspace{-6mm}
\caption{\textbf{Relevance Score Distribution}. For each algorithm, we compute the density of the user-item relevance scores computed by the original version of the recommender.}
\label{fig:coco-sim-distr}
\end{figure}

\vspace{2mm} \noindent \textbf{Influence of Relevance Score Distribution}. 
Having observed that the improvement in consistency greatly varies among algorithms, we conjecture that the distribution of relevance scores returned by the original algorithm may influence the feasibility of our approach. 
Hence, \figurename~\ref{fig:coco-sim-distr} shows the density of relevance scores along the range $[-1, 1]$. 
It can be observed that GMF, NeuMF, and ItemKNN-CB produced relevance scores with a high density around zero. 
Therefore, in our approach, the relevance part may be dominated by the consistency part, regardless of the applied $\lambda$. 
Consequently, relevance could have a drastic drop even for low $\lambda$ values, making it harder to find a good trade-off between accuracy and consistency. 
This behavior is confirmed by the results previously reported in \figurename~~\ref{fig:eq-ndcg-gain-loss}. 
The NDCG loss compared to the Consistency gain is higher for GMF, NeuMF, and ItemKNN-CB. 
Given that the relevance scores distribution is highly dense, even a small improvement in consistency may completely overturn the list of recommended courses, leading to the following observation:   

\vspace{2mm} \setlength{\fboxsep}{10pt}
\noindent \textbf{Observation 8}. \textit{The density of the relevance score distribution returned by an algorithm influences the trade-off between accuracy and consistency, after applying our approach. The higher the density, the higher the drop in accuracy.}\\

\vspace{2mm} \noindent \textbf{Qualitative Inspection}. 
The targeted principles and the corresponding consistency and equality metrics directly monitor properties of the recommender lists and the experiments showed that our approach leads potentially\footnote{Given that it is based on assumptions and technical implementation, our approach remains to be further evaluated from a human-centered point of view (e.g., on the learners' perceptions in subjective experiences in both academic and life-long outcomes). 
There is also the need to validate the value of $\lambda$ in online settings with learners, which remains in the domain of subjective decision-making by course designers and engineers.} to more consistent and equal recommended learning opportunities. 
However, it may also be interesting to inspect some recommended lists resulting from a traditional recommendation algorithm and how they change after our approach, closing the circle for the problem that motivated this study and assessing the practical end-to-end impact of the proposed approach.     

Table \ref{tab:sample-recs} shows how the list of courses recommend to an example learner changes after applying our procedure. 
First, regarding affordability, it can be observed that the re-ranked list offers a broader range of opportunities in terms of fees, even among courses from the same category. 
This aspect may enable a learner to receive suggestions that can better fit with their current financial resources. 
Then, more diverse opportunities were proposed in terms of instructional level and asset types, in line with the targets pursued by the platform. 
Except for the course with a large class ranked in the first position, our approach leads to courses with smaller and, thus, more manageable classes. 
However, this comes at the price of a slight loss of validity and category diversity. 
This happened because the learner mostly interacted with ``development'' and ``it-and-software'' courses in the past, so our approach promoted courses aligned with those categories (i.e., increasing familiarity). 

While the proposed approach confirmed its feasibility for conveying multiple principles into a recommended list and providing more equal learning opportunities across learners, it should be noted that it is ultimately up to the stakeholders to select principles and trade-offs most suitable for their context.

\begin{table}[!t]
\begin{tabular}{lrrrrrrr}
\hline
 & \multicolumn{7}{c}{\textbf{Before}} \\
\hline
 & \textbf{Item ID} & \textbf{Category} & \textbf{Fee} & \textbf{Level} & \textbf{Update} & \textbf{Learners} & \textbf{Asset} \\
\hline
0 & 17175 & development & 99.99 & all & 2017-08-01 & 778 & V+A \\
1 & 8518 & development & 199.99 & all & 2020-03-02 & 23386 & V+A+E \\
2 & 17772 & academics & 199.99 & all & 2019-07-30 & 6501 & V+A \\
3 & 18689 & development & 199.99 & all & 2019-08-11 & 10588 & V+A \\
4 & 9364 & development & 99.99 & all & 2016-03-01 & 3825 & V+A+E \\
5 & 17735 & business & 199.99 & all & 2020-03-13 & 19615 & V \\
6 & 7932 & development & 29.99 & beg & 2015-11-06 & 4661 & V+E \\
7 & 13191 & development & 199.99 & all & 2019-02-20 & 133490 & V+A \\
8 & 15861 & development & 99.99 & all & 2020-02-24 & 24705 & V+A \\
9 & 5676 & development & 199.99 & all & 2020-02-15 & 10638 & V+A \\
\hline
\multicolumn{8}{c}{\textbf{After}} \\
\hline
 & \textbf{Item ID} & \textbf{Category} & \textbf{Fee} & \textbf{Level} & \textbf{Update} & \textbf{Learners} & \textbf{Asset} \\
\hline
\multicolumn{1}{r}{0} & 4878 & development & 0.00 & beg & 2018-07-23 & 414169 & V+A+E \\
\multicolumn{1}{r}{1} & 9797 & it-and-sw & 0.00 & int & 2019-04-12 & 5605 & V+A \\
\multicolumn{1}{r}{2} & 17175 & development & 99.99 & all & 2017-08-01 & 778 & V+A \\
\multicolumn{1}{r}{3} & 7932 & development & 29.99 & beg & 2015-11-06 & 4661 & V+E \\
\multicolumn{1}{r}{4} & 9364 & development & 99.99 & all & 2016-03-01 & 3825 & V+A+E \\
\multicolumn{1}{r}{5} & 11275 & it-and-sw & 19.99 & int & 2019-08-30 & 777 & V \\
\multicolumn{1}{r}{6} & 8518 & development & 199.99 & all & 2020-03-02 & 23386 & V+A+E \\
\multicolumn{1}{r}{7} & 3707 & development & 19.99 & beg & 2015-02-27 & 1357 & V \\
\multicolumn{1}{r}{8} & 15861 & development & 99.99 & all & 2020-02-24 & 24705 & V+A \\
\multicolumn{1}{r}{9} & 8944 & it-and-sw & 19.99 & beg & 2018-11-10 & 8282 & V+E \\
\hline
\end{tabular}
\caption{\textbf{Impact on Recommendation}. Sample top-10 recommendations provided to a learner by the traditional recommender system based on ItemKNN-CB (top) and the top-10 recommendations resulting from our approach.} \label{tab:sample-recs}
\end{table}
\subsection{Discussion} 
There has been an increase in digitalized educational systems. 
Hence, online course platforms are promptly becoming an essential tool for providing learners with the most suitable material, meeting their expectations of educational values.
Due to the highly subjective and contextual nature of this process, educational platforms need to consider multiple perspectives. 
Indeed, besides providing a wide range of course filtering options, an increasingly high number of principles for further processing such options is needed to identify the most suitable ones for a learner. 
In this view, an in-depth understanding of recommendations in online course platforms may reduce the overload of learners, improving consistency and equality of the recommendations. 

Though the learners of COCO may not be representative of general learners in the recommender system, our analysis in Section~\ref{sec:exp-analysis} indicates that optimizing recommendation algorithms only for learners' interests may result in undermining other essential properties conveyed by the learning opportunities proposed to them. 
Ranges of educational recommender systems, such as those provided by~\cite{bridges2018course,rieckmann2018learning,bhumichitr2017recommender}, can thus capitalize on our definitions, metrics, and procedures as a mean for assessing recommendations' consistency. 
However, the principles proposed in this paper, which have been derived from curriculum design beliefs, would not necessarily align with the learner's beliefs of fairness along recommendations. 

A complementary human-centered perspective of the principle design can strongly benefit from our findings, leveraging our initial principles as a good starting point. 
Nonetheless, this approach might not be sufficient, and no one-size-fits-all set of principles would exist, given their dependence on the context and the involved stakeholders. 
Despite having a range of limitations in terms of context-sensitivity, learner-centeredness, operationalization, and temporal awareness (Section~\ref{sec:limitations}), the re-ranking procedure that was proposed and assessed in Section~\ref{sec:framework} has been shown to improve equality across learners, counteracting potential pitfalls of data-driven educational recommender systems. 
This aspect becomes of paramount importance in large-scale contexts, especially while reaching out to learners reluctant to the use of data-driven procedures \citep{herold2017case}. 
However, our results cannot prove that the differences in measured metrics translate to better educational outcomes and learners' acceptance. 
Finally, our work embeds views and needs of multiple educational stakeholders into recommended lists \citep{abdollahpouri2020multistakeholder}. 

In the broad discussion on FATE in AIED, we highlight that recommender's capabilities are an important component of AIED systems. 
Moreover, our research contributes to the improvement of our understanding of fairness in the educational recommendation context, by devising ways in which we can address fairness in AIED design. 
Our study moves a step forward in understanding how equality principles can be operationalized and combined in a formal notion of equal opportunities in educational recommendations. 
This contribution serves as a foundation to investigate how learners interpret concepts such as": 
(1) the fairness of the educational resource selection decisions they make (e.g., how they select courses for their degrees); 
(ii) what they think about the fairness of the set of resources available to them;
(iii) to what extent they view the course selection process as fair; and (iv) how these decisions are influenced by available information about the given courses. 
The principles and formulations described in our study would be a great starting point for this purpose. 
Therefore, this study permits the research community to know what questions to ask as part of interviews with learners, what scenarios to explore to elicit their concepts of fairness, and how to process data in the educational platform to monitor and ensure the equality principles. 

Since we are flexible concerning the underlying principles, our approach can still be applied even if human-centered investigations reveal that the principles captured in our study need to be fine-tuned or modified. 
This paper, thus, shapes a blueprint of the decisions and processes to be done, once empirically-validated principles have been defined under the targeted educational scenario. 
To this end, we provide evidence on what kind of technologic support is needed to ensure that learners' course selection decisions lead to greater equality across learners. 
Nonetheless, there are a number of issues regarding inequality in educational opportunities that recommender systems could not fix by themselves (e.g., whether certain advanced courses are available at all). 
Consequently, our study would help better understanding what is known about the role that recommender systems could play in the bigger questions around fairness and equality, grounding the design and implementation of formal notions of equality informed by a deep understanding of how learners view equality. 

\subsection{Limitations}\label{sec:limitations}
Since our observations varied over algorithms and principles, we identified the main implications and limitations of our study. 

\begin{itemize}
\item \textbf{Limitations of data}. While our results highlight the need to consider equality of recommended learning opportunities when evaluating recommenders, the learners of COCO may not be representative of general learners in an educational platform. Unfortunately, data with enough attributes to look for sophisticated principles is hard to find. As pointed out in Section \ref{sec:data}, other datasets include few attributes of learners and courses. 
\item \textbf{Limitations of principles}. While our principles shed light on important aspects underlying the ranked courses, they may not be representative of the principles targeted by certain platforms and are based on assumptions derived from a dataset. Thus, limitations arise from different perspectives.
\begin{itemize}
\item \textbf{Context dependency}.  
Because this study does not provide formative or summative results about actual systems, it is thus more theoretical than practical and does establish a framework for work in the context of fairness in educational recommender systems. 
In traditional scenarios, the operationalization of principles is usually based on textual guidelines, and the translation into numerical indicators (when performed) is subjected to the specificity of the platform. 
Our measures are just representative examples of how traditional text-based principles could be operationalized and enriched based on the peculiarity of the online context (e.g., manageability). 
However, our framework can be adapted to any (number of) principles, based on the pursued goals.
\item \textbf{Learner-centeredness}. 
The considered principles do not relate to individuals, but to resources, educational level, number of learners, and so on, referring to courses as a group. 
Given that learning and teaching deal with changes in individuals (e.g, how each learner reacts to a problem within a course or which asset in which course was valuable for learners), our principles should be enriched to reach this level of modeling, e.g., by adding principles connected with learning outcomes, when large enough dataset will become available to the research community. 
\item \textbf{Technical operationalization}. 
Some of the operationalizations have been overly simplified due to the limitations of the data currently available in online course platforms. 
For instance, the recency of updates is used as a proxy for validity, but there is more to validity than recency of updates. 
Similarly, learner ratings are used as a proxy for quality, but learner ratings would not always correlate with other measures of quality (e.g., learning outcomes), and class size is used as a proxy for manageability, but other aspects (e.g., number of assistants) are not captured. 
Finally, the way learnability is operationalized seems to be simplified, and other aspects might be targeted. 
This opens up to more advanced operationalizations.  
\item \textbf{Temporal influence}. 
Some of the principles are sensitive to time. 
For instance, regarding familiarity, sometimes a learner may be looking for something new that broadens their horizon, rather than something familiar; at other times, they may need one more final elective for their primary major, which might mean that the preferred resource would have high familiarity. 
In other words, it is still not clear that a given learner should be viewed as having a preference for a certain level of familiarity per se. 
The same learner may, at different times, have different preferences. Similarly, regarding the measure of validity (i.e, recency of the last update), a course on foundational material has been updated many times in the past and does not benefit from recent updates. 
Similar observations apply to other principles. 
\end{itemize}
\item \textbf{Limitations of the ethical constructs}. 
Given that our methodology has been assessed in an offline setting, the real-world validity of the notion of equality that is presented still needs to be shown. 
For instance, it should be investigated whether this notion aligns with learner's notion of fairness, whether learners pay attention to all principles when assessing fairness, and whether this notion of equality relates to fairness and ethical principles, especially if enhancing equality requires to give up a given degree of personalization. 
For instance, it is worthy of exploration of how much it is fair to ask one learner to give up a degree of personalization so that the familiarity target for another learner can be met. These issues will drive future research of equality of recommended learning opportunities.
\item \textbf{Limitations of algorithms}. 
Our study involves eight representative algorithms from four families, but other types of algorithms may benefit from our procedure. 
However, to better focus on the evaluation of our contribution and due to the limitations of the data, we constrained our study to algorithms that are key building blocks of several recommender systems. 
\item \textbf{Limitations of evaluation protocol}. 
Our results cannot prove that the differences in measured metrics translate to better educational outcomes and learners' acceptance. 
Further studies with online evaluation are needed to complement these results. 
However, we conjecture that our results can provide an essential contribution to reach this goal, and offline protocols can be useful to select algorithms prior to an online deployment.
\item \textbf{Limitations of metrics}. 
Among the large number of metrics that can be used for evaluating a recommender system, we focus on consistency and equality to better assess our contribution. 
We also measured NDCG because it maps well to recommendation utility. However, consistency and equality do not consider the position of the courses in a list, which can be important in large-scale recommendation contexts (e.g., online course platforms where tons of courses are provided and having courses at the top of the recommended list is crucial to receive visibility), as an example. 
Our study focused on a more general perspective to reach a broader audience.
\end{itemize}

\section{Conclusions} \label{sec:conclusions}
In this paper, we proposed a novel fairness metric that monitors the equality of learning opportunity across learners in the context of educational recommender systems, according to a novel set of educational principles. 
Then, we explored the learning opportunities provided by ten state-of-the-art recommender systems in a large-scale online course platform, uncovering systematic inequalities across learners. 
To counteract this phenomenon, we proposed a post-processing approach that re-ranks the recommended courses originally returned by an algorithm to maximize the equality of recommended learning opportunities while preserving personalization. 
Finally, we assessed the impact of supporting learners with our approach to accuracy and beyond-accuracy metrics. Based on the results, we can conclude that: 

\begin{enumerate}
\item Recommendation algorithms tend to produce ranked lists with low equality of recommended learning opportunities across learners, especially when the algorithm uses only user-item interactions as training data.
\item Under our definition of the targeted principles, equality of quality, validity, and manageability is guaranteed by recommenders. 
Familiarity, affordability, learnability, and variety exhibit strong deviations over algorithms.
\item Optimizing recommendations for consistency concerning a set of principles leads to higher equality of recommended learning opportunities. 
This is stronger when learner-specific weights are adopted.
\item Controlling learning opportunity results in higher familiarity, variety, and affordability while maintaining stable values for the other principles. 
However, quality may experience small losses after applying our procedure.
\item The impact of our approach on accuracy and consistency depends on the density of the relevance score distribution of the original recommendation algorithm. 
The higher the density, the higher the drop in accuracy is.
\end{enumerate}

Future work will embrace our findings to study to what degree the courses currently attended by learners satisfy the notion of equality, in addition to the courses that are recommended. 
Moreover, a learner-centered approach will be carried out to investigate what learner's notions are for the fairness of the educational resource selection decisions they make, to fine-tune and adjust our original set of principles. 
By extension, learner-specific targeted degrees for each principle will be consequently elicited and applied. 
Thanks to its flexibility, the notions and procedures proposed in this study can fit with a variety of applications within both educational and non-educational contexts. 
There is also room for considering how additional algorithms respond to evaluation and what internal mechanics contribute to achieve higher consistency and equality. 
Finally, as real-world applications should consider whether their recommender systems provide consistent and equal learning opportunities across learners, we believe that there will be an increasing amount of research related to applying our study to the educational industry. 

With this study, we highlighted that our notions and procedures are quite broad and incorporate elements of societal and ethical importance. 
It may be inevitable that, as recommender systems move further into education, it is more and more necessary that they embed strategies like the one we presented.

\section*{Acknowledgments}
This work has been partially supported by the Sardinian Regional Government, POR FESR 2014-2020 - Axis 1, Action 1.1.3, under the project ``SPRINT'' (D.D. n. 2017 REA, 26/11/2018, CUP F21G18000240009), and by the Ag\`encia per a la Competivitat de l'Empresa, ACCI\'O, under the project ``Fair and Explainable Artificial Intelligence (FX-AI)''.

\appendix

\section{Mathematical Notation for Targeted Educational Principles} \label{sec:math-formulation}

In this appendix, we provide the mathematical formulations associated with the educational principles proposed in Section \ref{sec:principles}. They have been adopted for computing to what extent each principle is achieved for each learner throughout the experiments. 

\vspace{1mm} \noindent \textbf{Familiarity}.  Given a course feature $F_1 \in \mathcal{N}$ associated with integer-encoded representation of the category $g \in G$ of a resource, we consider two distributions:
\begin{itemize}
\item $x(g|u)$: the distribution over categories $G$ of the set of resources $I_u$ user $u$ interacted with in the past, defined as $x(g|u) = |I_u^g| / |I_u|$; 
\item $y(g|u)$: the distribution over categories $G$ of the set of learning opportunities $\Tilde{I}_u$ recommended to learner $u$, defined as $y(g|u) = |\Tilde{I}_u^g| / |\Tilde{I}_u|$; 
\end{itemize}

\noindent where $I_u^g$ and $\Tilde{I}_u^g$ represent the set of resources belonging to category $g$ the learner $u$ attended and the recommender system proposed, respectively. Then, we define the familiarity of $\Tilde{I}_u$ for a learner $u$ as the inverse of the Hellinger distance across $x(G|u)$ and $y(G|u)$. Specifically:

\begin{equation}
    c_{\Tilde{I}_u}(1) = 1 - H(x(G|u), y(G|u))
\end{equation}
\vspace{1mm}

\noindent where $c_{\Tilde{I}_u}(1)=1$ if $x_u$ and $y_u$ are perfectly balanced, and the highest familiarity is achieved. Conversely, the minimum familiarity $0$ is achieved when $x_u$ assigns probability 0 to every event that $y_u$ assigns a positive probability (or vice versa). In the latter situation, the recommender suggests resources opposite with respect to the user's most familiar categories. 

\vspace{1mm} \noindent \textbf{Validity}. Given a course feature $F_2 \in \mathcal{N}$ representing the last time a resource has been updated and the opening time of the platform, denoted as $T_o$, we define the validity of a set of learning opportunities $\Tilde{I}_u$ at the current time $T_c$ as follows:

\begin{equation}
    c_{\Tilde{I}_u}(2) =\frac{1}{|\Tilde{I}_u|} \sum_{i \in \Tilde{I}_u} \frac{T_c - f_{2,i}}{T_c - T_o}  
\end{equation}

\noindent where values close to 0 mean that the learning opportunities are obsolete, while values close to 1 correspond to mostly fresh opportunities in $\Tilde{I}_u$.  

\vspace{1mm} \noindent \textbf{Learnability}. Given a course feature $F_3 \in \mathcal{N}$ representing the instructional level of a resource, we define the learnability in $\Tilde{I}_u$ as:

\begin{equation}
    c_{\Tilde{I}_u}(3) = 1 - GINI \left( \frac{|{\Tilde{I}_u}^{f_3}|}{|\Tilde{I}_u|} \, \, \forall \, \, f_3 \in F_3 \right)
\end{equation}

\noindent where $c_{\Tilde{I}_u}(3)$ is the inverse of Gini inequality index over the representations of all the instructional levels in $\Tilde{I}_u$, and $\Tilde{I}_u^{f_3}$ is the set of resources in $\Tilde{I}_u$ with instructional level $f_3$. A value of 0 implies large inequality, while high balance is obtained with values close to 1.

\vspace{1mm} \noindent \textbf{Variety}. Given that each resource $j \in \Tilde{I}_u$ is composed from a set of assets $L_j$ and that the asset type of a resource $j$ is denoted by $T_j=(t_l\in T\,:\, \forall l \in L_j)$, we define the variety of the types in $\Tilde{I}_u$ as:

\begin{equation}
    c_{\Tilde{I}_u}(4) = \frac{1}{|\Tilde{I}_u|} \sum_{i \in \Tilde{I}_u} \frac{|T_i|}{|T|}
\end{equation}

\noindent where values close to 0 mean that the learning opportunities are focused on few asset types, while asset types greatly vary for values close to 1.

\vspace{1mm} \noindent \textbf{Quality}. Given a learner-resource feedback's matrix $R$ and that the platform allows for ratings between $F_{5_{min}}$ to $F_{5_{max}}$, we  define the quality of a set $\Tilde{I}_u$ as follows:

\begin{equation}
    c_{\Tilde{I}_u}(5) =\frac{1}{|\Tilde{I}_u|} \sum_{i \in \Tilde{I}_u} \frac{1}{|U_i|} \sum_{u \in U_i} \frac{F_{5_{max}} - R_{u,i}}{F_{5_{max}} - F_{5_{min}}}  
\end{equation} 

\noindent where values close to 0 mean that the learning opportunities are of low quality, while values close to 1 are measured for high-quality opportunities. 

\vspace{1mm} \noindent \textbf{Manageability}. Given a course feature $F_6 \in \mathcal{N}$ representing the number of enrolled learners in a course and that the platform allows for classes from $F_{6_{min}}$ to $F_{6_{max}}$ learners, we define the manageability in a set of learning opportunities $\Tilde{I}_u$ as follows:

\begin{equation}
    c_{\Tilde{I}_u}(6) = 1 - \frac{1}{|\Tilde{I}_u|} \sum_{i \in \Tilde{I}_u} \frac{F_{6_{max}} - f_{6,i}}{F_{6_{max}} - F_{6_{min}}}  
\end{equation}

\noindent where values close to 1 mean that the learning opportunities include  small classes, while values close to 0 refer to large classes.

\vspace{1mm} \noindent \textbf{Affordability}. Given a course feature $F_7 \in \mathcal{R}$ representing the course enrollment fee and that the platform allows for courses with a cost between $F_{7_{min}}$ and $F_{7_{max}}$, we define the affordability of a set of learning opportunities $\Tilde{I}_u$ as follows:

\begin{equation}
    c_{\Tilde{I}_u}(7) = 1 - \frac{1}{|\Tilde{I}_u|} \sum_{i \in \Tilde{I}_u} \frac{F_{7_{max}} - f_{7,i}}{F_{7_{max}} - F_{7_{min}}}  
\end{equation}

\noindent where values close to 0 mean that the learning opportunities are highly expensive, while values close to 1 correspond to free-of-charge learning opportunities in $\Tilde{I}_u$.  

\section{Optimality Proof for the Proposed Post-Processing Approach} \label{sec:proof-opt}
The combinatorial maximization problem in Eq.~\eqref{eq:opt_prob} may be efficiently approximated with a greedy approach with $(1-1/e)$ optimality if the objective function of the maximization is submodular. This statement has been proved in the following demonstration. 

\begin{theorem}\label{th:submodular} Let $Consistency(p,q|w)=1-w\|p-q\|_{|C|}^{|C|}$, with $|C|F>0$ and $w_i \ge 0 \; \forall i \in \{0, \cdots, |C|\} $, then for any $\lambda\in[0,1]$ the function in~\eqref{eq:opt_prob}, 
\[
    f(\mathcal I|w) = (1-\lambda) \sum_{i\in \mathcal I}\widetilde R_{ui}+\lambda\,Consistency(p_u,{q}_{\mathcal I}|w),
\]
is submodular. \hfill$\circ$
\end{theorem}

\begin{proof}
    First, since $\widetilde R_{ui}>0$, it follows that $f_1(\mathcal I|w)=\sum_{i\in \mathcal I}\widetilde R_{ui}$ is a modular function (i.e., hence, also submodular), because it is a sum of positive quantities. 
    Second, 
    \begin{equation*} 
    \begin{aligned}
    f_2(\mathcal{I}|w) & = Consistency(p_u,{q}_{\mathcal I})  = w\|p_u-{q}_{\mathcal I}\|_{|C|}^{|C|} \\
    & = \sum_{i = 1}^k \; w_i \; \left|[p_u]_i-[{q}_{\mathcal I}]_i\right|^{|C|} 
     = \sum_{i = 1}^k x_i,
    \end{aligned}
    \end{equation*}
    where $x_i=w_i \; \left|[p_u]_i-[{q}_{\mathcal I}]_i\right|^{|C|}>0$. Again, $f_2$ is modular because it is a sum of positive quantities. 
    Since $f(\mathcal{I}|w)=(1-\lambda) f_1(\mathcal {I}|w)+\lambda f_2(\mathcal{I}|w)$, and the convex combination of submodular functions is submodular, $f$ is submodular.
 \end{proof}
 
\bibliographystyle{plainnat}
\bibliography{bibliography}

\end{document}